\documentclass[final,3p,times]{elsarticle}
\usepackage{amssymb}
\usepackage{lineno}

\biboptions{authoryear}

\usepackage[bookmarks=false]{hyperref}
    \hypersetup{colorlinks,
      linkcolor=blue,
      citecolor=blue,
      urlcolor=blue}

\usepackage{natbib}
\usepackage{url}
\usepackage{multirow}
\usepackage{algorithm}
\usepackage{algorithmic}
\usepackage{times}
\usepackage{bm,enumerate,amsmath,amssymb,amsthm}

\DeclareMathOperator*{\argmax}{\textrm{argmax}}

\newcommand{\bvec}[1]{\mbox{\boldmath $#1$}}
\newtheorem{definition}{Definition}

\newtheorem{prop}{Proposition}
\begin{document}

%\large
%
%Highlights
%
%\begin{itemize}
%\item Characterize a dynamic activity-based traffic resource allocation problem that aims to maximize discounted social welfare considered under the elastic demand of customers.
%\item Provide an exact solution algorithm and some approximation algorithms, using the zero-suppressed binary decision diagram (ZDD), all of which guarantees to satisfy both space-time constraints of customers and capacity constraints of traffic resources at any time.
%\item Evaluate the performance of our proposed algorithms numerically, exploring the trade-off between allocation efficiency and computational costs.
%\end{itemize}
%
%\normalsize
%\clearpage

\begin{frontmatter}



%\dochead{International Symposium of Transport Simulation (ISTS'18) and the International Workshop on Traffic Data Collection and its Standardization (IWTDCS'18)}%

%\begin{document}

\title{Dynamic traffic resources allocation under elastic demand of users with space--time prism constraints}


\author[a]{Keiichiro Hayakawa} 
\author[b]{Eiji Hato}
%\author[a,b]{Third Author\corref{cor1}}

\address[a]{Toyota Central R\&D Labs., Inc., 41-1, Yokomichi, Nagakute Aichi 480-1192, Japan}
\address[b]{The University of Tokyo, Hongo, Bunkyo-ku, Tokyo, Japan}

\begin{abstract}

We present a conceptual framework for the dynamic traffic resources allocation problem in a situation of elastic demand among customers. We introduce an activity-based model to express customers' successive actions and transfers in order to capture the essential aspect of transfers as a derived demand. We focus on the decision-making of customers, in that they only use a mobility service when their space--time prism constraints represent the worst case. 
Under the setting with such elastic demand, we characterize a class of dynamic traffic resources allocation mechanisms that strictly keep space--time prism constraints of users and capacity constraints of traffic resources. Within the class of mechanisms, we show the optimal mechanism that maximizes discounted social welfare and show an exact solution algorithm for both, myopic and non-myopic settings, using the zero-suppressed binary decision diagram. We also present approximation algorithms that are still included in the class of mechanisms.
In numerical studies, we showed that our proposed algorithm works effectively in settings with high rejection rates, meaning that this algorithm can be used to focus on the behavior of latent customers who have not used mobility services so far. The proposed non-myopic algorithm is suitable to design services for limited small number of customers, for example a car-sharing services for prescribed members, while the proposed myopic algorithm is suitable for designing services for a large unspecified number of customers, for example, a ride-sharing service in a large city. 
\end{abstract}

\begin{keyword}
Activity-based; Dynamic traffic allocation; Space--time prism constraints; Strict capacity constraints; System optimal




\end{keyword}
\cortext[cor1]{Corresponding author: Keiichiro Hayakawa. E-mail: kei-hayakawa@mosk.tytlabs.co.jp}
\end{frontmatter}

%\correspondingauthor[*]{Corresponding author. Tel.: +0-000-000-0000 ; fax: +0-000-000-0000.}
%\email{author@institute.xxx}

%\linenumbers


%%%%%%%%%%%%%%%%%%%%%%%%%%%%%%%%%%%%%%%%%%%%%%%%%%%%%%%%%%%%%%%%%%%%%%%%%%%%%%%%%%%%%%%%%%%%%%%%%%%%%%%%%%%%%%%%
%
\section{\label{sec_intro}Introduction}									
%
%
%%%%%%%%%%%%%%%%%%%%%%%%%%%%%%%%%%%%%%%%%%%%%%%%%%%%%%%%%%%%%%%%%%%%%%%%%%%%%%%%%%%%%%%%%%%%%%%%%%%%%%%%%%%%%%

% MaaSの導入， ケータイアプリによる実装
In many settings, traffic operators have limited traffic resources, such as airplanes, buses, trains, and taxis. Conventionally, these traffic resources were controlled under a traffic operator and customers had to buy tickets from the operator. When users traveled, they had to arrange for tickets from various traffic operators to complete their series of trips. 
However, recently, there are often cases in which multiple operators share traffic resources and provide a combinatorial set of services to customers, and this is often called Mobility as a Service (MaaS).
The operators of these services aim to make an efficient match between limited traffic resources and heterogeneous customers.
% 事業者目線
In many cases, the booking systems for such kinds of services are developed as a mobile app, and are implemented using mobile phones that are strongly connected to each user. In fact, there are many mobile apps that recommend restaurants, driving routes, and even the time taken to rest based on the habitual behavior of the users that are logged into the app within a long-term history. 
A traffic service operator can allocate traffic resources efficiently if they can use information on the heterogeneity of customers.
% 顧客目線
In contrast, customers also benefit greatly from such services. Using such services, they can not only make arrangements for total trips by booking only once, but can also receive recommendations for activities and transfers, considering their habitual activities. Receiving such recommendations, users may change their behavior, which may enrich their lives.
% Elastic Demand 
Considering that, the demand for such services is considered to be elastic, customers may change their destination on a trip depending on the recommendations they find, and may even make a decision on whether to go out or stay at home, while availing the service. Thus, the demand greatly depends on the quality of the mobile services.
In this paper, we discuss the mechanisms and algorithms of such mobility services.

% 時空間制約のコンセプトを導入
In such applications, customers' space--time prism constraints \citep{hagerstraand1970people} on their activities have a significant meaning. Customers have a series of constraint sets of \emph{time} and \emph{location} that they have to keep. For instance, in many cases, people have to return home before or at an appropriate time at night, every day. Workers have to go to work during their prescribed working hours on weekdays and thus, all activities other than the work itself, for example, going shopping or eating dinner with friends, should be assigned to either their morning or evening commutes on weekdays, while this constraint vanishes in the weekends. In the real world setting, the space--time prism constraints of users are often quite strict. For example, users cannot be late for work because being late imposes a great penalty on the users.
% Activity-based モデル
To address such constraints, the activity-based travel analysis \citep{kitamura1996sequenced} is effective, which regards travel demand as a derived demand to realize a series of activities that have occurred at different locations. Using this approach, customers can be considered as agents that seek to maximize their utility under the given space--time prism constraints.

% Activity-based traffic assignment
Based on this policy, \citet{lam2001activity} proposed an activity-based dynamic traffic assignment model. They define \emph{activity utility function} for each activity of the users and consider the users as utility maximizers, based on the utility function and travel costs involved. Given these factors, they formulated the steady-state user equilibrium (UE) under the combined activity and route choices of users and proposed a heuristic solution algorithm using the space--time expanded network. This approach handles the activity-based elastic demand at points in which the users select a set comprising the location of the activity and the route to realize the activity, simultaneously. 
However, both, space--time prism constraints of users and the capacity constraints of traffic resources, are expressed by cost functions, and are thus, not strictly treated. 
%It is the significant elements that go into making a decision on using a mobility service, such as whether the service guarantees keeping space--time constraints strictly, or not. On the other hand, the capacity constraints of traffic resources are significant for the service operator, and they should be handled strictly. 
From the viewpoint of the service operator, it is worthwhile to know the system optimal (SO) state under strict constraints.

% Dial-a-ride
In contrast, a considerable strand of the literature related to the dial-a-ride problem (DARP) \citep{cordeau2007dial} has discussed traffic resources allocation that aims to minimize costs of the service operator under strict capacity constraints. In these works, the optimization problems are typically formulated as an integer linear programming (ILP) and a wide range of solution algorithms are provided \citep{ho2018survey}. However, the problems are formulated as optimization problems of the operator, given the deterministic nature of stochastic trip-based demand of customers, and thus, they neither consider the connectivity of a series of transfers, nor the change in demand reacting to the service quality. 

% 本研究の主張点
In this paper, we consider the setting in which user agents on mobile apps make reports on the demands and constraints of customers to the central server of the service operator that allocates a series of trips and activities to customers.
We especially focus on the situation in which the demand of customers is elastic, specifically, where customers make a decision to use the service only when their demands and constraints are satisfied within the entire series of trips and activities. For example, a user uses the mobility service to attend an event held at a museum at a certain time, to go shopping after that, and to return home by a certain deadline, and does not use the service otherwise.
The operator executes a combined activity and trip (e.g., route, traffic mode, etc.) assignment considering the elastic demand of customers.
To address this, we introduce a framework that aims to maximize social welfare under elastic demand. The framework is based on a state transition model with normalized time-step. Thus, sequential decision-making of the operator is described consistently, on the basic economic theory.
Within the framework, we characterize a class of mechanisms that strictly keep space--time prism constraints of users and capacity constraints of traffic resources. By using this class of mechanisms, the \emph{floating booking system} of mobility services can be realized. Unlike common first-come first-served booking systems, by which early users have priority over any later arrivals, it rationally reallocates previously assigned traffic resources to late-coming high-value users, while guaranteeing the worst case service quality for all customers that it once accepts. It also provides users options to change their trips on the way.
To realize this, we propose algorithms that use the zero-suppressed binary decision diagram (ZDD) \citep{minato1993zero,minato2001zero} that provides a framework with which a sparse choice-set in a discrete state space derived from various combinatorial problems, is represented compactly and set operations related to the choice-set are computed efficiently. 
\citet{radivojevic1996new} shows that graph algorithms are superior to algorithms based on ILP in settings where resource-constrained scheduling problems with non-linear constraints and the computation cost are significantly improved using the ZDD. \citet{bjorkman2013solving} shows various applications that can be solved efficiently by using the ZDD. 
We take this approach to solve the traffic resources allocation problem in MaaS settings.

Overall, we make the following novel contributions. 
\begin{itemize}
\item We characterize, for the first time, a class of dynamic activity-based traffic resources allocation mechanisms named \emph{RC (Resources Customers)-feasible} mechanisms that guarantee the satisfaction of both, space--time prism constraints of customers and capacity constraints of traffic resources at any time, in its sequential decision-making. Within this class of mechanisms, we introduce the \emph{RC-optimal} mechanism that has the objective function that maximizes the discounted social welfare each time.
\item We propose a floating booking system that allocates traffic resources efficiently by rationally reallocating previously assigned traffic resources to late-coming high-value customers. We provide exact solution algorithms for the RC-optimal mechanism using the ZDD, and provide approximation algorithms that are still RC-feasible.

\item We numerically show that our proposed solution algorithm can keep more than trillions of combinatorial trip options of current and future agents in rational computational time. We also numerically explore the trade-off between allocation efficiency and computational costs of our proposed wide range RC-feasible algorithms. 

\end{itemize}

The remainder of this paper is organized as follows. In Section~\ref{sec_review} we organize the related works in the literature. In Section~\ref{sec_overview}, we show the system overview of the considered mobility service, describing the model of users and the service operator. In Section~\ref{sec_mechanisms}, we formulate the sequential decision-making of the operator that aims to maximize social welfare and characterize the RC-feasible and RC-optimal mechanisms. In Section~\ref{sec_algorithm}, we provide the solution algorithm for both, myopic and non-myopic settings. Using ZDD, we show the exact algorithm for the RC-optimal mechanism, as well as show the approximation algorithms that are still RC-feasible. 
In Section~\ref{sec_numerical}, we numerically evaluate our proposed mechanism and discuss the selection method of myopic or non-myopic algorithm and approximation algorithms depending on the applications. Finally, in Section~\ref{sec_discussions}, we conclude the work and discuss the potential direction of the extension of this work.

\section{\label{sec_review}Literature review}									
%
%
%%%%%%%%%%%%%%%%%%%%%%%%%%%%%%%%%%%%%%%%%%%%%%%%%%%%%%%%%%%%%%%%%%%%%%%%%%%%%%%%%%%%%%%%%%%%%%%%%%%%%%%%%%%%%%%

In this section, we organize the related works in literature. Specifically, we focus on four fields, namely, the dial-a-ride problem (DARP), traffic assignment with capacity constraints, the multi-agent path finding problem (MAPF), and the activity-based model. 

\subsection{Dial-a-ride problem (DARP)}

A considerable strand of the literature has discussed traffic resources allocation that aims to minimize costs of the service operator under strict capacity constraints. For scheduling fleet management, a problem in which scheduling a set of tasks and routes for vehicles based on fully given information of demand and strict time constraints, is well-known as the Pickup and Delivery Problem with Time Windows (PDPTW) \citep{dumas1991pickup}. In this framework, \citet{horn2002fleet} proposed myopic solutions in dynamic situations in which real-time outputs should be determined based on real-time inputs.
For share-ride systems, DARP \citep{cordeau2007dial} is well-known, which is a generalization of PDPTW, that includes additional constraints related to transport passengers rather than the fleet. Specifically, it aims to provide a cost-minimum traffic operation, given multiple users' trip requests, and the related aspects of origin, destination, and time-window under limited traffic resources. This problem is typically formulated as ILP and the exact solution is proved to be NP-hard. Thus, a noticeable strand of the literature has proposed heuristic algorithms, which are extensively reviewed by \citet{ho2018survey}. As pointed out by \citet{ho2018survey}, many works with regard to DARPs discuss static situations and thus, much effort is required to address dynamic DARPs.
On the other hand, dynamic ridesharing systems have been paid much attention, recently \citep{pelzer2015partition}. This problem is similar to DARPs, but differ in some ways, as pointed out by \citet{agatz2012optimization}.

Many works in these fields have proposed myopic solutions in dynamic situations, on the lines of \citep{horn2002fleet} in PDPTW. A trivial base-line benchmark of such myopic algorithms is the first-come first-served (FCFS) algorithm, in which the system myopically determines the resources allocation each time, upon accepting user requests. To achieve more efficient allocation, requests asked within a certain time slice are collected in order to explore efficient matching within these requests in many works. Some of these have proposed graph-based algorithms to handle non-linear constraints of combinatorial requests \citep{kamar2009collaboration,santi2014quantifying,alonso2017demand}. In contrast, \citet{sayarshad2015scalable} proposed non-myopic heuristic algorithms to solve the dynamic dial-a ride and pricing problems by employing the Markov decision process (MDP) to explore the look-ahead policy. 

However, all of these works consider trip-based demands, and thus, cannot treat the correlation of successive trips. In addition, this class of problems cannot represent situations where the demand may change, depending on the service quality determined by the operator.
They assume fixed demand or stochastic demand for each time and location, and do not consider the elastic demand derived from the quality of mobility services, and the achievable activities resulting from the services. 
%However, this class of problems cannot represent the situation where the demand of customers may change depending on the service quality determined by the operator.
%\textcolor{blue}{多くのモデルにおいて、所要時間にリンクコストを導入することで、Dynamic モデルにおける時間進行との理論的な整合が取れなくなっている。Elastic Demand への対応}

%%%%%%%%%%%
\subsection{Traffic assignment with capacity constraints}

In contrast to the DARP that optimizes traffic resources under given demand, the field of traffic assignment discusses customer behavior under given traffic resources. Specifically, many works discuss UE states under given traffic networks, assuming the stochastic discrete choice models of customers. In these works, the cost of transferring each link in the network is expressed by a link cost function that is (weakly) monotonically increasing with the number of customers using the same link simultaneously. The maximum capacity of the link is not strictly considered. To address this problem, \citet{nie2004models} characterized a traffic assignment problem with capacity constraints that considers the UE state on networks with link capacity constraints. 

From the viewpoint of the service operator of mobility services, it is interesting to discuss optimal traffic resource allocation to maximize social welfare, or profit, considering customer behavior that has been described in the traffic assignment problem. For instance, \citet{allsop1974some} formulate traffic signal control as a bi-level problem in which the upper level problem expresses the traffic signal control that optimizes the objective function under given demand, while the lower level problem expresses the user equilibrium states under the given traffic signal control. Solving this bi-level problem iteratively, we can obtain the solution called \emph{Nash-Cournot Equilibrium} state. On the other hand, \citet{fisk1984game} proposed to solve this bi-level problem simultaneously. Specifically, a traffic manager optimizes their policy, incorporating the reaction of the driver to the traffic control system as constraints. This approach is classified in the \emph{Stackelberg game}, and formulated as a mathematical programming problem with equilibrium constraints (MPEC). The solution to this form is considered as SO states given the behavioral constraints of customers.
\citet{akamatsu2017tradable} proposed a demand management for static conditions as a combinatorial auction in which the SO state on networks with link capacity constraints coincide with UE states under a well-defined incentive design. The discussion is based on the Vickrey--Clarke--Groves (VCG) mechanism \citep{vickrey1961counterspeculation,clarke1971multipart,groves1973incentives} that consists of a traffic control achieving SO state and a well-defined pricing algorithm.
However, \citet{akamatsu2017tradable} only considers static settings and the extension to the dynamic settings is not trivial. Moreover, it assumes trip-based demands, and thus, cannot treat the correlation of successive trips and the related sequential decision-making of customers.

To address these problems, in this paper, we first focus on the SO states in dynamic settings, given the elastic demand with behavioral constraints of customers expressed by space--time prism constraints. Specifically, we pursue the mechanism that maximizes expected discounted social welfare, which is considered the SO state in dynamic settings in the basic economic theory\citep{rust1994structural}.
In combination with a well-defined incentive design, our proposed mechanism can also provide SO states as a result of \emph{Stackelberg game}, which we discuss in Section~\ref{sec_discussions}.

%%%%%%%%%%%
\subsection{Multi-agent Path Finding Problem (MAPF)}

The MAPF problem is a problem seeking paths of multiple agents that minimize the sum of the path cost of all agents without collision \citep{silver2005cooperative}. The MAPF is regarded as a generalization of the single-agent path finding problem and thus, the graph-based search algorithm such as $A^{\ast}$ is commonly used as a solution algorithm, while ILP or mixed integer program (MIP) are commonly chosen in DARP and traffic assignment problems, as stated earlier.
In the MAPF problems, the state transition of multiple agents is described with normalized time. Specifically, the state consists of position and the time-step, where the space is divided into grids, and each grid can be occupied at most by one agent at a time step to avoid collision. If we regard this as a strict capacity constraint of traffic resources, the problem that we consider in this paper can be regarded as a generalization of the MAPF. Indeed, \citet{ma2017lifelong} propose a lifelong algorithm within the online pickup and delivery setting based on MAPF.
As pointed out by \citet{amir2015multi}, the MAPF is formulated consistently on basic economic theory, and thus, can be generalized to combinatorial auction that includes not only the common cooperative MAPF, but also the non-cooperative MAPF. To obtain the optimal states in the setting, \citet{amir2015multi} propose an algorithm using multi-valued decision diagrams (MDDs) \citep{srinivasan1990algorithms} to represent a massive set of plans compactly, which is a family of the binary decision diagrams \citep{akers1978binary}, that are similar to the ZDD that we adopt in this paper.

Although the MAPF is interesting in its consistency in the basic economic theory, and is suitable to describe sequential decision-making on part of agents, many works related to MAPF only consider the myopic algorithms under fixed objectives. In this paper, we describe the dynamic traffic allocation problem in a similar manner, with MAPF, and discuss the sequential decision-making of the service operator aiming to maximize social welfare.

%%%%%%%%%%%
\subsection{Activity-based model}

There are earlier studies on activity-based travel analysis that regard a trip as the derived demand of users' activities, and try to model travelers' time usage within their daily activities. 
\citet{axhausen1992activity} classified activity analysis into two conceptual frameworks, namely, \emph{utility-maximization} and \emph{electric}. The former assumes that people choose to spend their time in a way that maximizes utility within their space--time prism constraints, whereas the latter addresses the scheduling process explicitly. Although the problems discussed in these frameworks are indeed difficult to come to grips with, finding solutions for them are essential in enabling policy makers to manage traffic demand generated from household activities appropriately.
\citet{kitamura1996sequenced} proposed an activity-based utility model that replicates adaptive time-of-day dynamics, and introduced a simulator that offers dynamic and integrated forecasting of elements, such as transportation and land use.

%\paragraph{Space--time expanded network}

The space--time expanded network is key to solving various problems that are formulated by the activity-based utility model. Such approaches were originally employed to solve dynamic traffic assignment (DTA) problems in many dynamic user equilibrium (DUE) models \citep{drissi1993variational, yang1998departure}. This method is well-suited to activity-based travel analysis, because it expressly considers the space--time prism constraints. Indeed, \citet{lam2001activity} formulated the combined activity/route choice problem as the ideal DUE and provided a method to solve this problem using a space--time extended network, as we have already stated in Section.~\ref{sec_intro}. Moreover, \citet{arentze2004multistate} introduced the multistate supernetwork, which can represent the multimodal transport system with sequential activities. This model was inspired by the supernetwork concept as introduced by \citet{sheffi1985urban}, which aims to enrich network representations in order to model traffic mode choices. In the multistate supernetwork, each node expresses a combination of an activity state and a vehicle state, and each edge expresses a transition between the states. Thus, the choice of sequential activities is expressed as a trajectory in the supernetwork and the complete trip chains that involve multiple transport modes can be obtained as a least-cost path. \citet{liu2015dynamic} later formalized dynamic activity travel assignment as a discrete-time DUE problem on a multistate supernetwork. 
\citet{oyama2017discounted} proposed a dynamic activity-based traffic assignment by considering users' sequential decision-making with a recursive logit model \citep{fosgerau2013link}.
However, these works use travel time as costs of edges, and consider that decisions are made at each node. Some works even consider congestion by changing the edge cost while keeping the space--time extended network as it is. This approach takes advantage of decreasing computational costs, but is not economically consistent, if we consider a sequential decision-making based on the model.
On the other hand, \citet{hara2017car} considers the optimal user and vehicle assignments for a car-sharing service, using space--time expanded network with normalized time, and considers the willingness to pay as the edge-rewards, in a manner similar to \citet{lam2001activity}. They provide algorithms that maximize the social welfare by introducing the temporal and spatial connection of users within the space--time expanded network. 
Thus, an activity-based analysis using the space--time expanded network is not only used to express multimodal and multistate behavior, but also to develop urban planning or traffic service policies that consider user heterogeneity.

\citet{recker2001bridge} discussed the relationship between trip-based and activity-based travel analysis and demonstrated that activity-based travel analysis can be formalized with mathematical programming based on traditional trip-based modeling methodologies with the addition of temporal and spatial constraints. Although we formulate problems that aim to maximize social utility under rational utility-maximizing customers, the discussion can be applied to problems that aim to minimize total travel time under rational cost-minimizing customers.

%\color{red}
%正規化時間に言及
%\textcolor{blue}{多くのモデルにおいて、所要時間にリンクコストを導入することで、Dynamic モデルにおける時間進行との理論的な整合が取れなくなっている。}

%\citet{oyama2017discounted} proposed an activity-based traffic assignment in this setting, focusing on space--time prism constraints. 
%Several works have considered traffic assignment without capacity constraints; for instance, stochastic traffic assignment in static situations is often considered using logit models, which can be extended to dynamic situations by considering users' sequential decision-making with a recursive logit model \citep{fosgerau2013link}. \citet{oyama2017discounted} proposed an activity-based traffic assignment in this setting, focusing on space--time prism constraints. 

%It is the significant elements that go into making a decision on using a mobility service, such as whether the service guarantees keeping space--time constraints strictly, or not. On the other hand, the capacity constraints of traffic resources are significant for the service operator, and they should be handled strictly. 

%\clearpage
%%%%%%%%%%%%%%%%%%%%%%%%%%%%%%%%%%%%%%%%%%%%%%%%%%%%%%%%%%%%%%%%%%%%%%%%%%%%%%%%%%%%%%%%%%%%%%%%%%%%%%%%%%%%%%%%
%
\section{\label{sec_overview}System Overview}									
%
%
%%%%%%%%%%%%%%%%%%%%%%%%%%%%%%%%%%%%%%%%%%%%%%%%%%%%%%%%%%%%%%%%%%%%%%%%%%%%%%%%%%%%%%%%%%%%%%%%%%%%%%%%%%%%%%%

\begin{figure}[t]
\begin{center}
	    \includegraphics[width = \hsize]{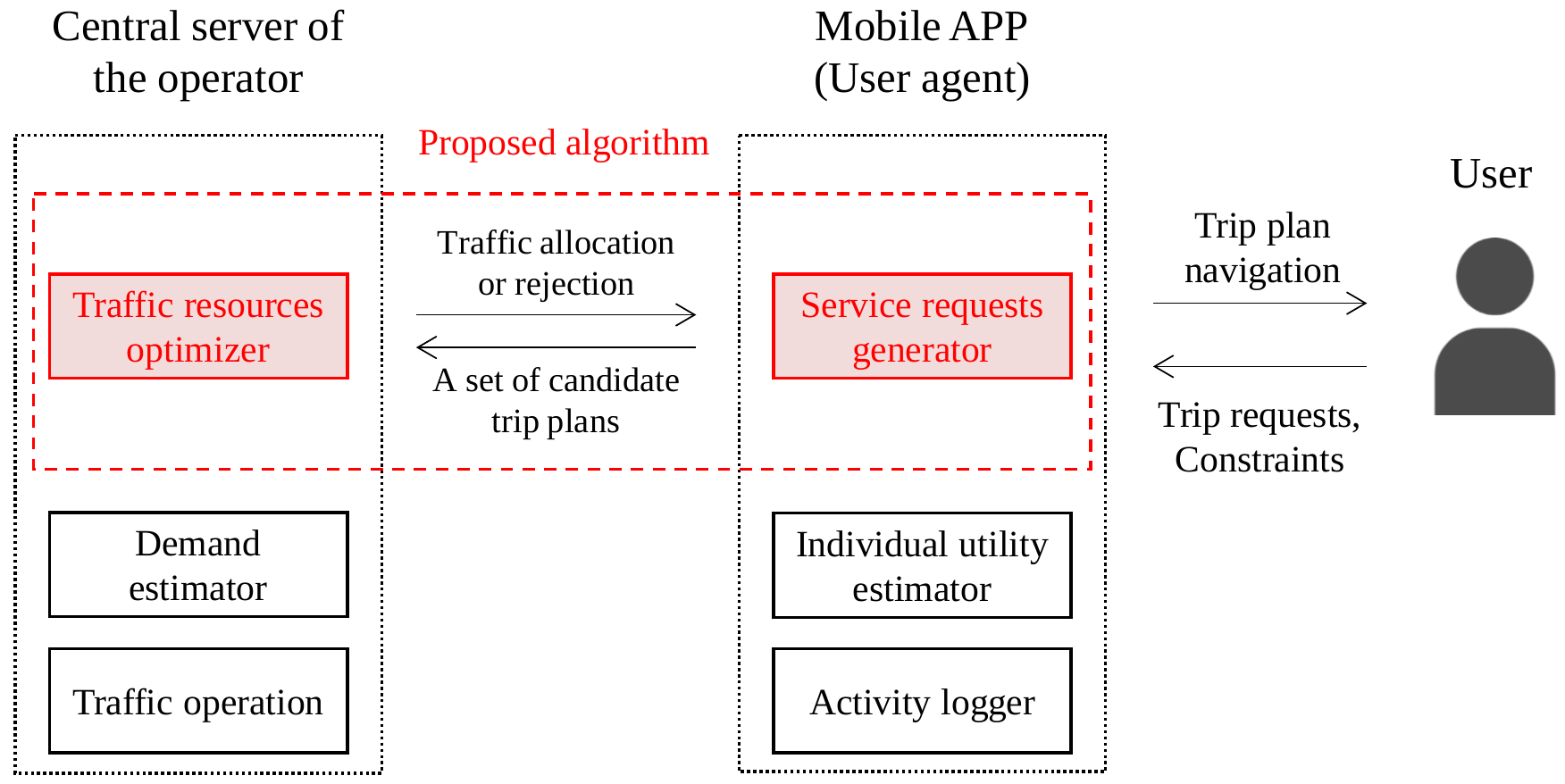}\\

          \caption{System overview} 
	    \label{fig_SystemOverview}
\end{center}
\end{figure}

We consider a MaaS system, which is implemented by a service operator. We show the system overview in Fig.~\ref{fig_SystemOverview}. The traffic allocation algorithm is implemented separately from the central server of the operator and app on the users' mobile phones. A service request generator is implemented on the mobile app. It generates a set of candidate trip-plans based on the users' direct requests as well as information generated from the activity loggers and individual utility estimator. This is then reported to the operator. In the following part of this paper, we call this service request generator as \emph{user agent}, or \emph{agent} for short.
Receiving requests from a set of agents, the traffic resources optimizer implemented in the operator makes a decision as to whether it allocates or rejects each request, and allocates traffic resources efficiently to the agents it accepts. The decision is notified to users through the mobile app and the traffic resources are operated as the decision, such as for example, taxis are dispatched. The operator may use the information of estimated future demand in its decision. In this paper, we focus on the traffic resources optimizer, and call it the \emph{operator}. 
Thus, we propose algorithms implemented for the operator and the agents that appropriately allocate traffic resources to users. 

We assume that decisions are made at discrete time steps $t \in T = \{1, 2, \ldots, \bar{T}\}$, where $T$ is the set of all time steps. 
The traffic network is expressed by a directed graph $\mathcal{G}=(\mathcal{N},\mathcal{E})$. It can represent multiple traffic modes by a concept of supernetwork \citep{sheffi1985urban,arentze2004multistate}. The service operator, having a limited capacity of traffic resources for each edge and each time in the network, allocates the resource for each user. Since the traffic capacity constraints are strictly kept in the allocation, we do not consider congestion and assume that the travel time of each edge remains constant. \footnote{This assumption is adopted in \citet{lam2001activity} as well, which discusses an activity-based traffic assignment model that aims to express time dependent UE states.} We use $\tau_e$ to denote the travel time of edge $e \in \mathcal{E}$. 
Some facilities on node $n \in \mathcal{N}$ may have an active time duration, such as for example, the opening time of a restaurant. We use $b(n) \subset T$ to denote the active time duration of node $n$.

In the following part, we introduce the model of user agent and the service operator, which is used in Section~\ref{sec_mechanisms}, where dynamic traffic allocation mechanisms are characterized.

%\begin{itemize}
%\item User agent　と Central server of the traffic operator からなる仕組み。
%\item 離散時刻$t \in T = \{ 0,1,\ldots, \bar{T}\}$
%\end{itemize}

%%%%%%%%%%%%%%%%%%%%%%
%\subsection{Traffic resources}
%
%\begin{itemize}
%\item directed graph $\mathcal{G}=(\mathcal{N},\mathcal{E})$であらわされる交通ネットワークを考える。
%\item マルチモーダルを考える場合、Super network \cite{sheffi1985urban,arentze2004multistate}を用いる。
%\item Traffic operator は、各時刻ごと、各リンクごとの交通リソーセスを有する。時刻$t \in T$にリンク$e \in \mathcal{E}$へ流入可能な容量を$C_{e,t}$とする。
%\item 本研究では、各リンクの所要時間は一定とし、リンク$e \in \mathcal{E}$の所要時間を$\tau_e$とする。
%\item Strict な容量制約を前提としており、混雑は考えない。この仮定は、\cite{lam2001activity}と同様。
%\end{itemize}

%%%%%%%%%%%%%%%%%%%%%
\subsection{User Agent Model}

We use $I = \{1, 2, \ldots, \bar{I} \}$ to denote a set of user agents. A type of user $i \in I$ is expressed by origin $O_i \in \mathcal{N}$, destination $D_i \in \mathcal{N}$ and active time duration $T_i = \{ t^B_i, t^B_i+1, \ldots, t^E_i \} \subset T$, where $t^B_i$ denotes the earliest possible time for starting activities and $t^E_i$ denotes the strict deadline of finishing all the activities. Here, destination $D_i$ means the final destination after visiting several locations for various activities. Typically, both, origin $O_i$ and destination $D_i$ are the user's home. The service provides a series of trips from leaving his/her home after $t^B_i$ and until returning his/her home before $t^E_i$. 

%ユーザーエージェントの集合を$I = \{0,1,\ldots, \bar{I} \}$で表す．エージェント$i \in I$は出発地$O_i$，最終目的地$D_i$，および，活動開始可能時間$t^B_i$と活動終了時間$t^E_i$であらわされる活動可能時間帯$T_i = \{ t^B_i, t^B_i+1, \ldots, t^E_i \} \subset T$を有しているものとする．
%ここで、最終目的地とは、いくつかの目的地を経由して最終的にたどり着きたい場所。典型的には、出発地も目的地も共に自宅ということが考えられる。

%%%%%%%%%%%%%%%%%%%%%
\subsubsection{\label{sec_st_constraints}Space--time prism constraints}

\begin{figure}[t]
\begin{center}
	    \includegraphics[width = 0.4\hsize]{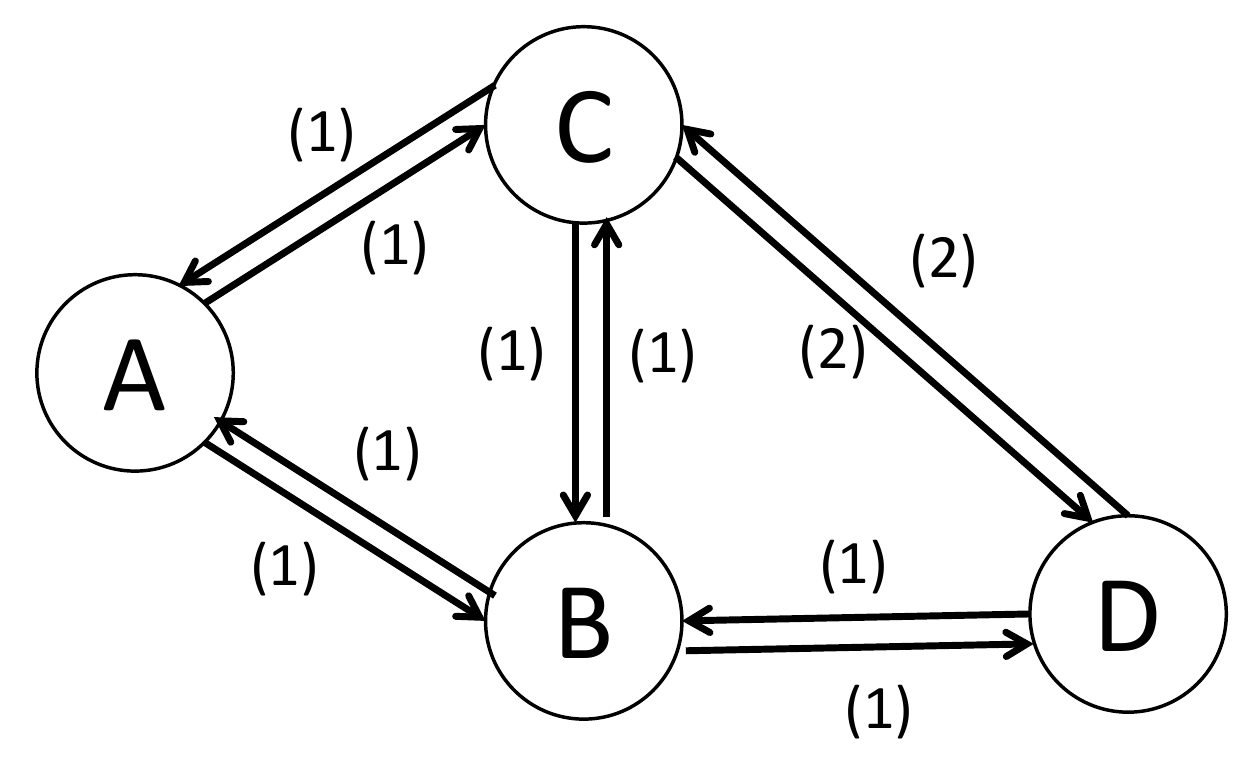}\\

          \caption{Sample network and the link costs} 
	    \label{fig_SampleNetwork}
\end{center}
\end{figure}

\begin{figure}[p]
	%2
	\begin{center}
	    \includegraphics[width = 0.8\hsize]{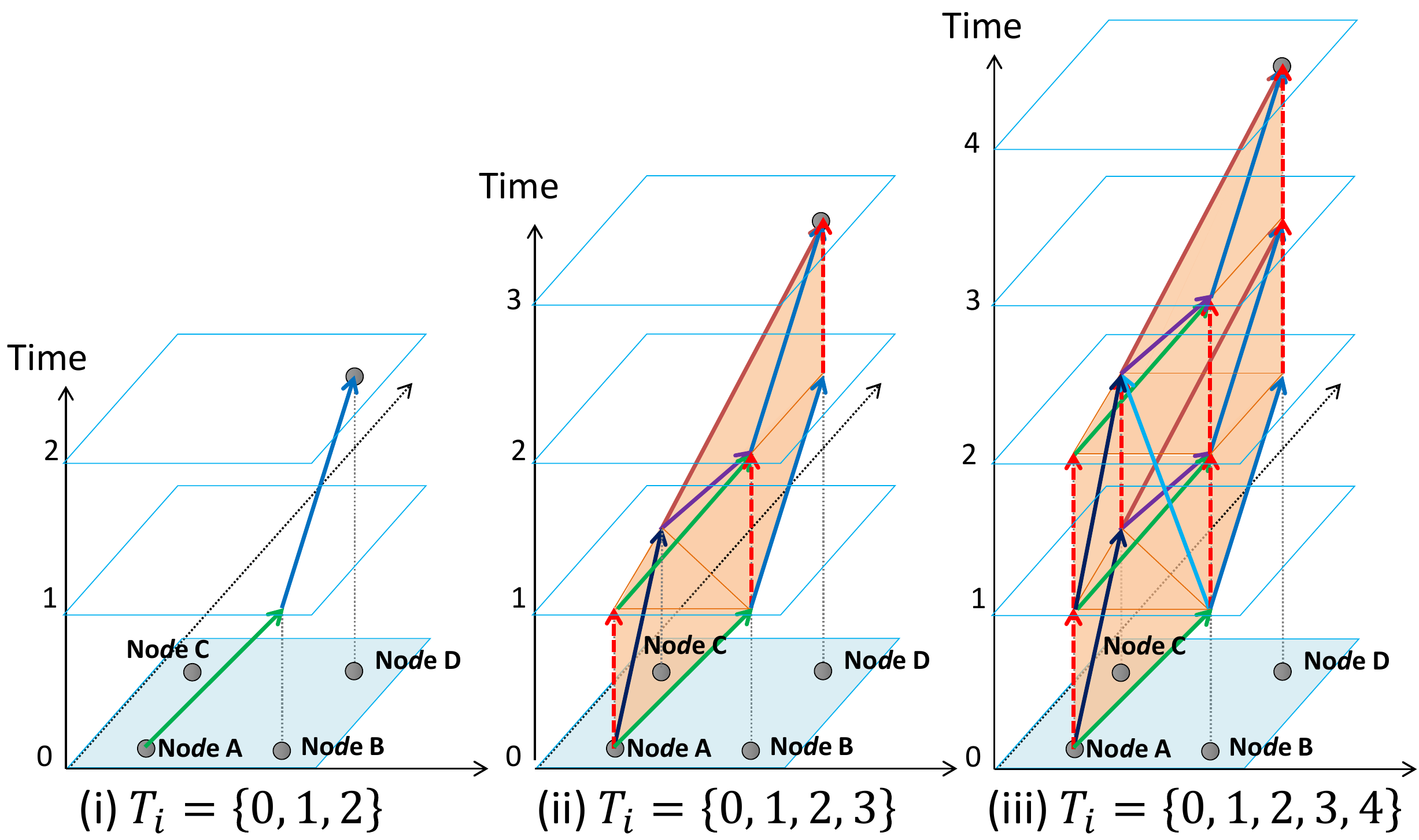}\\
		
		(a) Space--time constraints at origin and destination

	\vspace{1.0cm}

	    \includegraphics[width = 0.8\hsize]{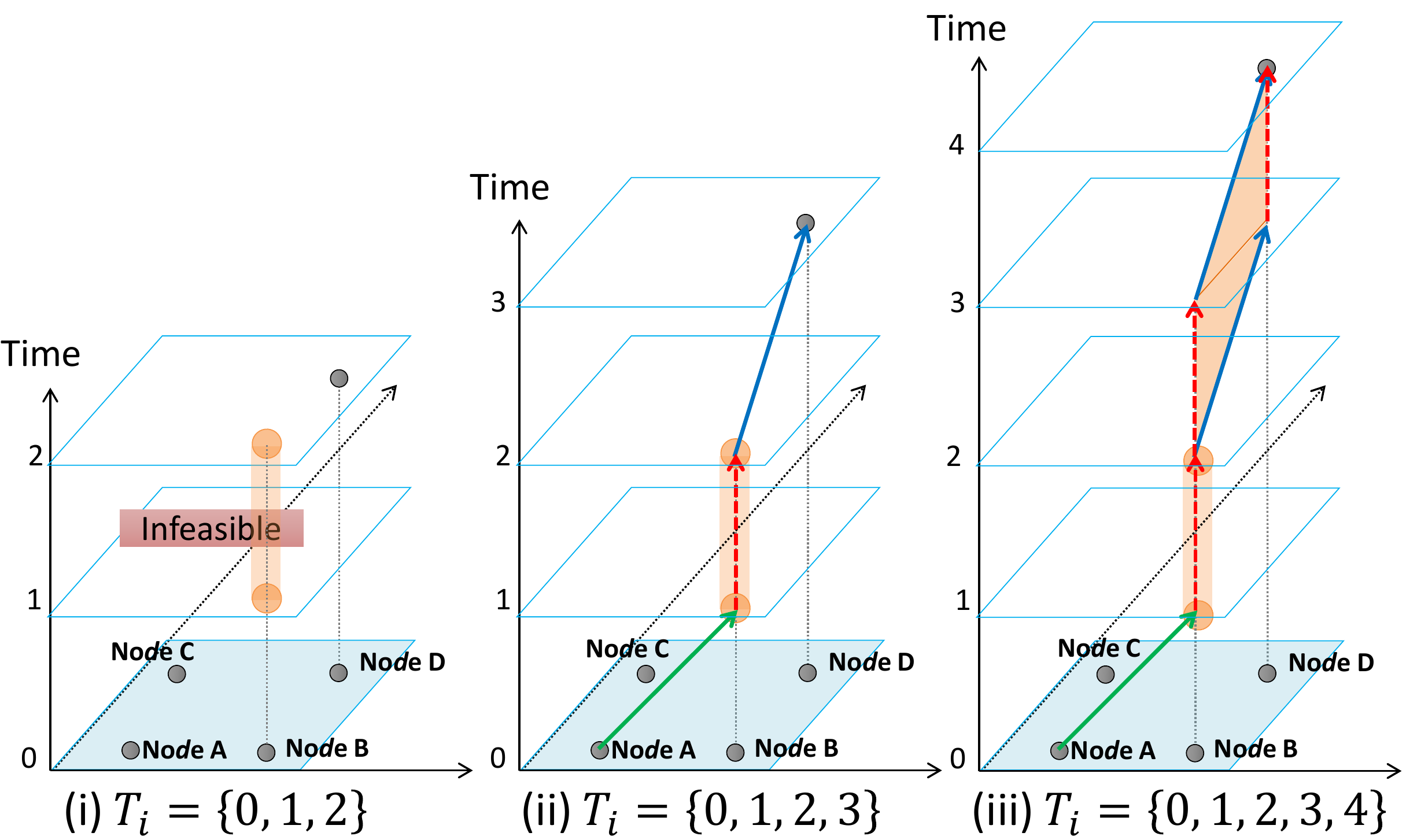}\\

		(b) Space--time constraints at origin, destination, and visiting places

	\vspace{1.0cm}

         \caption{A set of trip trajectories that keep the space--time prism constraints} 
%         \caption{時空間プリズム制約を満たす行動軌跡} 
         \label{fig_TimeSpacePrism_sample}
	\end{center}
\end{figure}

In generating trip-plans, including activities and transfers, agent $i$ has strict space--time prism constraints \citep{hagerstraand1970people} to guarantee that the user must be on $O_i$ at $t^B_i$ and on $D_i$ at $t^E_i$.
Here, we consider a sample network shown in Fig.~\ref{fig_SampleNetwork}. The network has 4 nodes $\mathcal{N} = \{A, B, C, D\}$ and 10 directed edges. The numbers in brackets show the required travel time on each edge. Specifically, two time steps are required to move between Nodes C and D, and one time step is required for the other edges.

Now, we consider an agent $i$ with origin $O_i = A$, destination $D_i = D$. Using space--time expanded network \citep{lam2001activity}, the candidate trip trajectories with given three different active time durations $T_i$ are shown as in Fig.~\ref{fig_TimeSpacePrism_sample}(a). In this figure, solid lines express the transfer along links and dashed lines express the staying at nodes.
In cases (i), (ii), and (iii), the final time step $t^E_i$ is set to 2, 3, and 4 and the number of candidate trajectories are 1, 5, and 15, respectively.
The space--time prism constraints are tightened if the agent has another constraint related to its activity. For instance, assume that the active time duration of a facility on node B is limited as $b(B) = \{ 1\}$, and the above agent $i$ has to visit Node B during its active time duration, the candidate trajectories are decreased as shown in Fig.~\ref{fig_TimeSpacePrism_sample}(a). Under this assumption, in case (i) with $t^E_i = 2$, the trip cannot be generated. In cases (ii) and (iii) with $t^E_i = 3$ and $t^E_i = 4$, the number of candidate trajectories are 1 and 2, respectively.
Many users have non-linear constraints on their series of activities, such as, for example, constraints stated by \emph{if-then-else} rule, and the candidate trip-plans are tightened by such constraints. 

%このとき，エージェント$i$は，時刻$t^B_i$に出発地$O_i$に，時刻$t^E_i $に目的地$D_i$に居る必要があるため，活動時間帯の間に実行可能な行動については，時空間的な制約が生じる．本研究では，Hagerstrand\cite{hagerstraand1970people}にならって，これを時空間プリズム制約と呼ぶ．

%例として，図~\ref{fig_SampleNetwork}に示すサンプルネットワークを考える．図中のかっこ内に示された数値は当該リンクを移動するのにかかる時間であり，リンクCDの移動には2時刻ステップ，その他のリンクの移動には1時刻ステップを必要とする．
%このとき、時空間プリズム制約を考慮したエージェントの行動軌跡の例を図~\ref{fig_TimeSpacePrism_sample}に示す．この例では，出発地がノードA，到着地がノードDであるエージェントについて，３種類の異なる活動時間帯を与えた場合の実行可能な行動軌跡を、space--time expanded network \cite{lam2001activity}上で図示したものである．活動時間帯が短い(i)のケースでは，エージェントが取り得る行動軌跡は1通りであるのに対して，活動時間帯が増加した(ii)および(iii)のケースでは，それぞれ4通り，9通りの実行可能な行動軌跡が考えられる．

%時空間プリズム制約をさらに強める条件として，各ノード$n \in \mathcal{N}$上の施設が利用可能である時間帯$\tau_n \subset T$と，各エージェントの必須立ち寄り施設条件を考える．例えば，上述のサンプルケースで，$\tau_B=\{1\}$，すなわちノード$B$上の施設が時刻$t = 1$のときのみ利用可能であり，ノードAからノードDに向かうあるエージェントにとって施設$B$への立ち寄りが必須であると仮定すると，活動時間帯が上述の(i)のケースではこの制約を満たす行動軌跡は存在せず，(ii)および(iii)のケースで実行可能な行動軌跡はそれぞれ1通り，2通りに絞り込まれる．
%このように，エージェントは，出発地，目的地，活動可能時間帯に加えて，必須の滞在地点，滞在時刻などに起因する非線形制約を持つ。

%%%%%%%%%%%%%%%%%%%%%
\subsubsection{State transition}

We introduce $s_{i,t} \in \mathcal{S}_i$ to denote the state of agent $i \in I$ at time $t \in T$, where the state is a multi-dimensional index including for example, location $y$, traffic mode $m$, and so on. $\mathcal{S}_i$ denotes a set of states that agent $i$ can take. We introduce the location function $\lambda(\cdot)$, namely the location $y$ of an agent with state $s_{i,t}$ is obtained by $y = \lambda(s_{i,t})$. Note that $y$ can be the middle of an edge when $\tau_{e} \geq 2$.
The action that agent $i$ with state $s_{i,t}$ takes at time $t$ is denoted by $a_{i,t} \in \Gamma(s_{i,t}) \subset \mathcal{A}_i$, where $\Gamma(s_{i,t})$ is a set of actions that the agent can take at that time and $\mathcal{A}_i:\mathcal{S}_i \rightarrow \mathcal{S}_i$ denotes a set of actions that agent $i$ can take under all possible states. Taking action $a_{i,t} \in \Gamma(s_{i,t})$, the state of the agent transits from state $s_{i,t}$ to state $s_{i,t+1} \in \mathcal{S}_i$.
We use $\mathcal{S} = \bigcup_{i \in I}\mathcal{S}_i$ to denote a set of states for all agents and $\mathcal{A} = \bigcup_{i \in I}\mathcal{A}_i$ to denote a set of actions taken by all agents.

As actions of agents, we consider \emph{Moving} and \emph{Staying}. For instance, an agent located on Node C in the sample network shown in Fig.~\ref{fig_SampleNetwork} at time $t$ can take one of the following actions, namely, staying at Node C, or start moving toward Node A, B, or D. Formally, we use $\mathcal{A}_i^M \subset \mathcal{A}_i$ to denote a set of \emph{Moving} actions of agent $i$ while $\mathcal{A}_i^S \subset \mathcal{A}_i$ to denote a set of \emph{Staying} actions, where $\mathcal{A}_i = \mathcal{A}_i^M \cup \mathcal{A}_i^S$. A set of staying actions $\mathcal{A}_i^S$ expresses a stay of agent $i$ on any node $n \in \mathcal{N}$, potentially, doing some activity. We assume that the staying actions are taken only on any node $n \in \mathcal{N}$ in the traffic network $\mathcal{G}$ and are not taken anywhere else such as, for example, in the middle of a edge.
In contrast, moving action $a_{i,t} \in \mathcal{A}_i^M$ is taken along any edge $e \in \mathcal{E}$ in the traffic network $\mathcal{G}$. We introduce the edge function $\eta(\cdot)$. The edge $e \in \mathcal{E}$, related to action $a_{i,t} \in \mathcal{A}_i^M$ of an agent, is obtained by $e = \eta(a_{i,t})$. For any staying action $a_{i,t} \in \mathcal{A}_i^S$, $\eta(a_{i,t}) = \emptyset$.

Given that, we introduce a function $\delta: \mathcal{S} \times \mathcal{A} \times \mathcal{E} \rightarrow \{0,1\}$ to express the relationship between actions and edges, as follows:
\begin{equation}
\delta(s_{i,t},a_{i,t},e)=
\begin{cases}
1 & \text{if} \hspace{0.5cm} \lambda(s_{i,t}) \in \mathcal{N}, \hspace{0.2cm} \text{and} \hspace{0.2cm} \eta(a_{i,t}) = e \\
0 & \text{otherwise}.
\end{cases}
\end{equation}
Specifically, $\delta(s_{i,t},a_{i,t},e)=1$ if an agent $i$ that is located on any node in the traffic network at time $t$ take a moving action related to edge $e$ at that time, and otherwise $\delta(s_{i,t},a_{i,t},e)=0$.
Given that, the traffic volume $F_{e,t}$ that flows into edge $e$ at time $t$ is expressed as:
\begin{equation}
F_{e,t} = \sum_{i \in I} \delta(s_{i,t},a_{i,t},e).
\label{eq_link_flow}
\end{equation}
We assume that an agent starting the moving action at time $t \in T$ along a edge $e$ continues the action until it reaches the end of the edge at time $t + \tau_e$ and does not change the action in the middle of the edge.

%%%%%%%%%%%%%%%%%%%%%
\subsubsection{Trip plan}

As shown in Section~\ref{sec_st_constraints}, actions of each agent are bounded by non-linear space--time prism constraints. We use $L_i$ to denote a set of executable state-action trajectories that keep the constraints. A series of state-action transition $l_i$ during the active time duration $T_i = [ t^B_i, t^E_i ]$ is expressed as below:
\begin{equation}
	l_i = \{ s_{i, t^B_i}, a_{i, t^B_i}, \ldots, s_{i, t^E_i-1}, a_{i, t^E_i-1}, s_{i, t^E_i}\} \in L_i.
\end{equation}
By the constraints about origin and destination, $\lambda(s_{i, t^B_i}) = O_i$ and $\lambda(s_{i, t^E_i}) = D_i$ are satisfied.
The state-action trajectory $l_i$ provides a plan for the trip-chain from the origin to the destination including multiple activities and transfers. In the following parts of the paper, we call the trajectory a trip-plan.
Considering a time $t$ such that $t^B_i < t < t^E_i$, we use $l_i^{\langle t \rangle}$ to denote a trip-plan after time $t$, such that:
\begin{equation}
	l_i^{\langle t \rangle} = \{ s_{i, t}, a_{i, t}, \ldots, s_{i, t^E_i-1}, a_{i, t^E_i-1}, s_{i, t^E_i}\} \in L_i^{\langle t \rangle},
\end{equation}
where $L_i^{\langle t \rangle}$ denotes a set of executable trip-plans after time $t$ while keeping the space--time prism constraints. 

%前節で示したように、各エージェントの行動は、時空間プリズム制約に代表される非線形制約を受ける。
%本研究では，このような非線形制約を満たすエージェント$i$の実行可能なstate-action transitionの集合を$L_i$で表し，活動可能時間帯$T_i = [ t^B_i, t^E_i ]$の間の一連のstate-action transitionを
%\begin{equation}
%	l_i = \{ s_{i, t^B_i}, a_{i, t^B_i}, \ldots, s_{i, t^E_i-1}, a_{i, t^E_i-1}, s_{i, t^E_i}\} \in L_i
%\end{equation}
%で表す．
%出発地と目的地の制約から、$\lambda(s_{i, t^B_i}) = O_i$, $\lambda(s_{i, t^E_i}) = D_i$である。

%また、エージェント$i$が時刻$t^B_i < t < t^E_i$以降にとる軌跡を、
%\begin{equation}
%	l_i^{\langle t \rangle} = \{ s_{i, t}, a_{i, t}, \ldots, s_{i, t^E_i-1}, a_{i, t^E_i-1}, s_{i, t^E_i}\} \in L_i^{\langle t \rangle}
%\end{equation}
%とする。ただし、$L_i^{\langle t \rangle}$は、時空間プリズム制約を考慮したうえでエージェント$i$が時刻$t$以降に実行可能な軌跡の集合である。

%$l_i $は、出発から到着までの複数の移動と滞在を含むトリップチェーンのプランを示している。

%%%%%%%%%%%%%%%%%%%%%
\subsubsection{Time dependent reward function}

The utility that the agents receive from activities are often dependent on the time of day. Thus, we introduce time dependent reward functions, in  a manner that is similar to \cite{axhausen1992activity,supernak1992temporal,lam2001activity}.
Formally, we define the agents' reward function by $R_i:\mathcal{S}_i \times \mathcal{A}_i  \times T \rightarrow \mathbb{R}$. Thus, the utility $U_i(l_i)$ that agent $i$ derives upon completing a trip-plan $l_i= \{ s_{i, t^B_i}, a_{i, t^B_i}, \ldots, s_{i, t^E_i-1}, a_{i, t^E_i-1}, s_{i, t^E_i}\}$ is expressed as:
\begin{equation}
U_i(l_i) = \sum_{t=t^B_i }^{t^E_i-1} R_i(s_{i, t}, a_{i, t},t).
\end{equation}
%\color{blue}
The time-discounted utility $DU_i(l_i^{\langle t \rangle})$ of agent $i$ that takes a trip-plan $l_i^{\langle t \rangle} = \{ s_{i, t}, a_{i, t}, \ldots, s_{i, t^E_i-1}, a_{i, t^E_i-1}, s_{i, t^E_i}\}$ at time $t$ is expressed as:
\begin{equation}
DU_i(l_i^{\langle t \rangle}) = 	\sum_{t'=t }^{t^E_i-1} \beta^{t'-t} R_i(s_{i, t'}, a_{i, t'},t'),
\end{equation}
where $\beta$ is time discount rate. We assume that $\beta$ is common for all users and the operator.
%\color{black}
%
%エージェントの効用は、time-of-day dependent であると考える\cite{axhausen1992activity,supernak1992temporal,lam2001activity}。
%エージェントの報酬は、報酬関数$R_i:\mathcal{S}_i \times \mathcal{A}_i  \times T \rightarrow \mathbb{R}$で表される。
%エージェント$i$がプラン$l_i$をとるときの効用$U_i(l_i)$は、
%\begin{equation}
%U_i(l_i) = \sum_{t=t^B_i }^{t^E_i-1} R_i(s_{i, t}, a_{i, t},t)
%\end{equation}
%であらわされる。
%時間割引率$\beta$を考慮すると、時刻$t$の時点でエージェント$i$がプラン$l_i^{\langle t \rangle}$を取ることに対しての割引効用$DU_i$は、
%\begin{equation}
%DU_i(l_i^{\langle t \rangle}) = 	\sum_{t=t' }^{t^E_i-1} \beta^{t'-t} R_i(s_{i, t'}, a_{i, t'},t')
%\end{equation}
%であらわされる。

%%%%%%%%%%%%%%%%%%%%%
\subsubsection{Reports to the operator}

In the previous parts, we have defined non-linear space--time constraints and the reward function of agents. In this paper, we assume that this information is given exactly, but is available only after the beginning of the active time duration of agents. That is, information about the constraints and the reward function related to agent $i$ is available only after $t^B_i$. This assumption is reasonable because many users input their demands into mobile apps at the time of starting their trips. At time $t^B_i$, user agents generate a set of executable trip-plans $L_i$ and report the information $\theta_i = \{ R_i, L_i\}$, namely the reward function as well as a set of trip-plans, to the central server of the operator.

%%%%%%%%%%%%%%%%%%%%%
\subsubsection{Decision-making on participation}

Now, we consider the decision-making of users on participation with respect to mobility services. We consider users behaving based on minmax strategy to decide whether they use a service or not. Specifically, a user uses the service only when his/her space--time prism constraints are kept and the discounted utility $DU_i$ is larger than $h_i$ in the worst case scenario, where $h_i$ is the discounted utility that the user can obtain without using the service. In the following parts in this paper, we set $h_i=0$ for all agent $i \in I$, without loss of generality.

The agent makes this decision at its start time $t^B_i$. If an agent decides to use the service, the agent accepts the trip-plan that the operator assigns to the agent. In contrast, if an agent decides not to use the service, the agent does not use any traffic resource of the operator at any time $t \geq t^B_i$, and the operator is not concerned with the agent. We use $l_i^{C}=\{s_{i,t^B_i}, a_{i}^{cancel}\}$ to denote the trip-plan of agent $i$ that decides not to use the service at $t^B_i$, where $a_{i}^{cancel}$ denotes the action that the agent decides not to participate.

%%%%%%%%%%%%%%%%%%%%%
\subsubsection{\label{sec_flexible_behavior}Flexible behavior of users}

In the previous part, we assumed \emph{static agent-type} settings, where an agent reports its type $\theta_i = \{ R_i, L_i\}$ at the beginning of the trip and it remains unchanged until it finishes the trip. However, agents may change their minds, for example, they may want to return home earlier, or may want to visit another place. Our proposed mechanism can easily be extended to such situations. We call this the \emph {dynamic-agent type} settings. In these settings, an agent can report a new type $\theta'_i = \{ R'_i, L'_i\}$ to the operator in an arbitrary timing.

However, to show the essential property of our proposed mechanism in simple terms, in the following part of this paper, we consider the \emph{static agent-type} setting, except in Section~\ref{sec_mechanism_dynamic_agent} in which we shortly state the mechanism for the \emph {dynamic-agent type} setting, and in Section~\ref{sec_dynamic_type} and \ref{sec_results_dynamic_type} in which we numerically explore the effectiveness of our proposed mechanism in the \emph {dynamic-agent type} setting.

\subsection{Service Operator Model}

As stated in the previous part, the service operator sequentially receives the agents' report and is required to allocate traffic resources for each user adequately at each time. In this part, we focus on the operator and introduce the model of the traffic resources capacity and its allocation.

%%%%%%%%%
\subsubsection{Capacity of traffic resources}

\begin{figure}[t]
\begin{center}
	%1
	\begin{minipage}{0.33\hsize}
	\begin{center}
	    \includegraphics[height = 5.0cm]{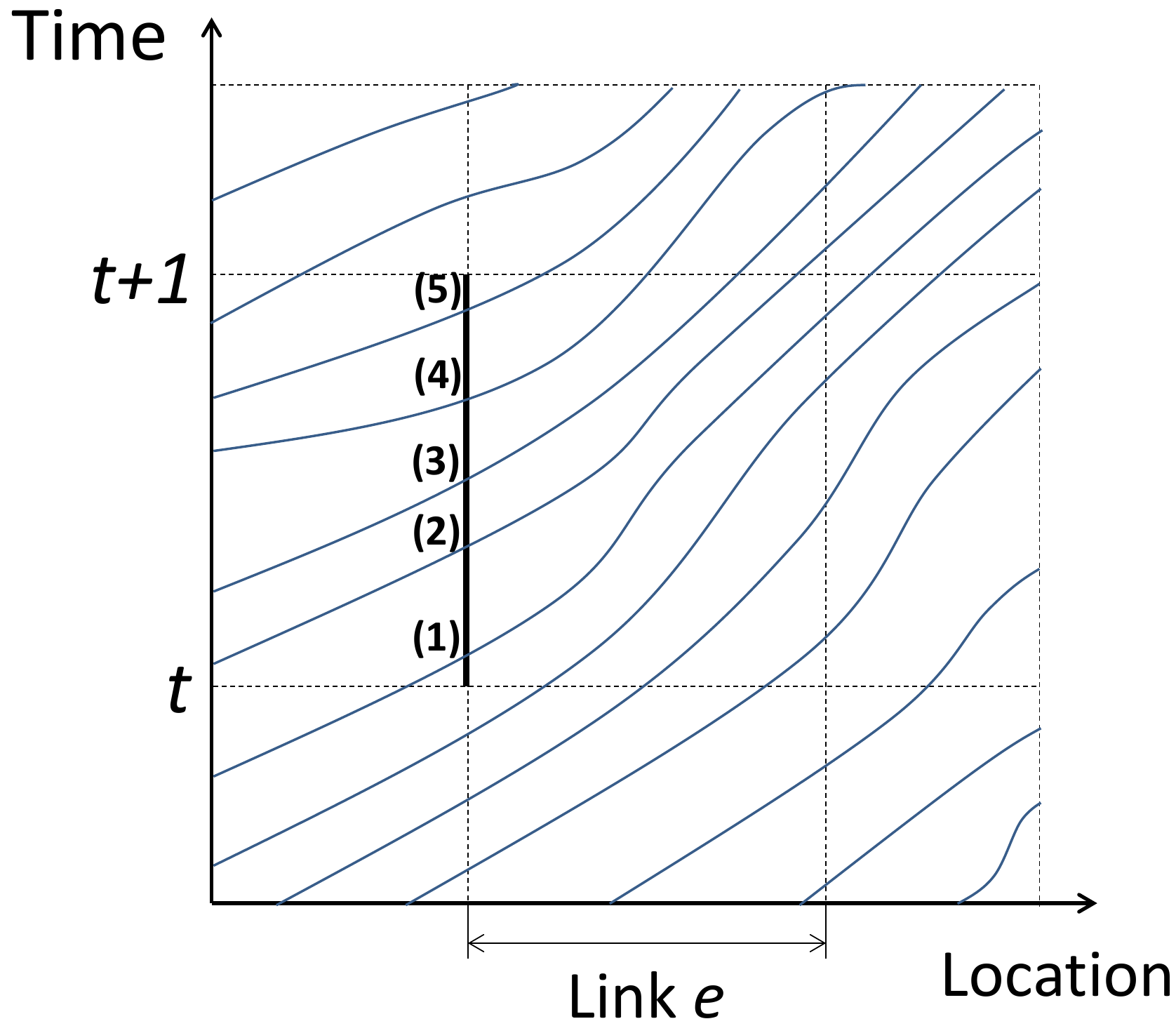}\\

	(a) Private vehicles
	\end{center}
	\end{minipage}
	%2
	\begin{minipage}{0.33\hsize}
	\begin{center}
	    \includegraphics[height = 5.0cm]{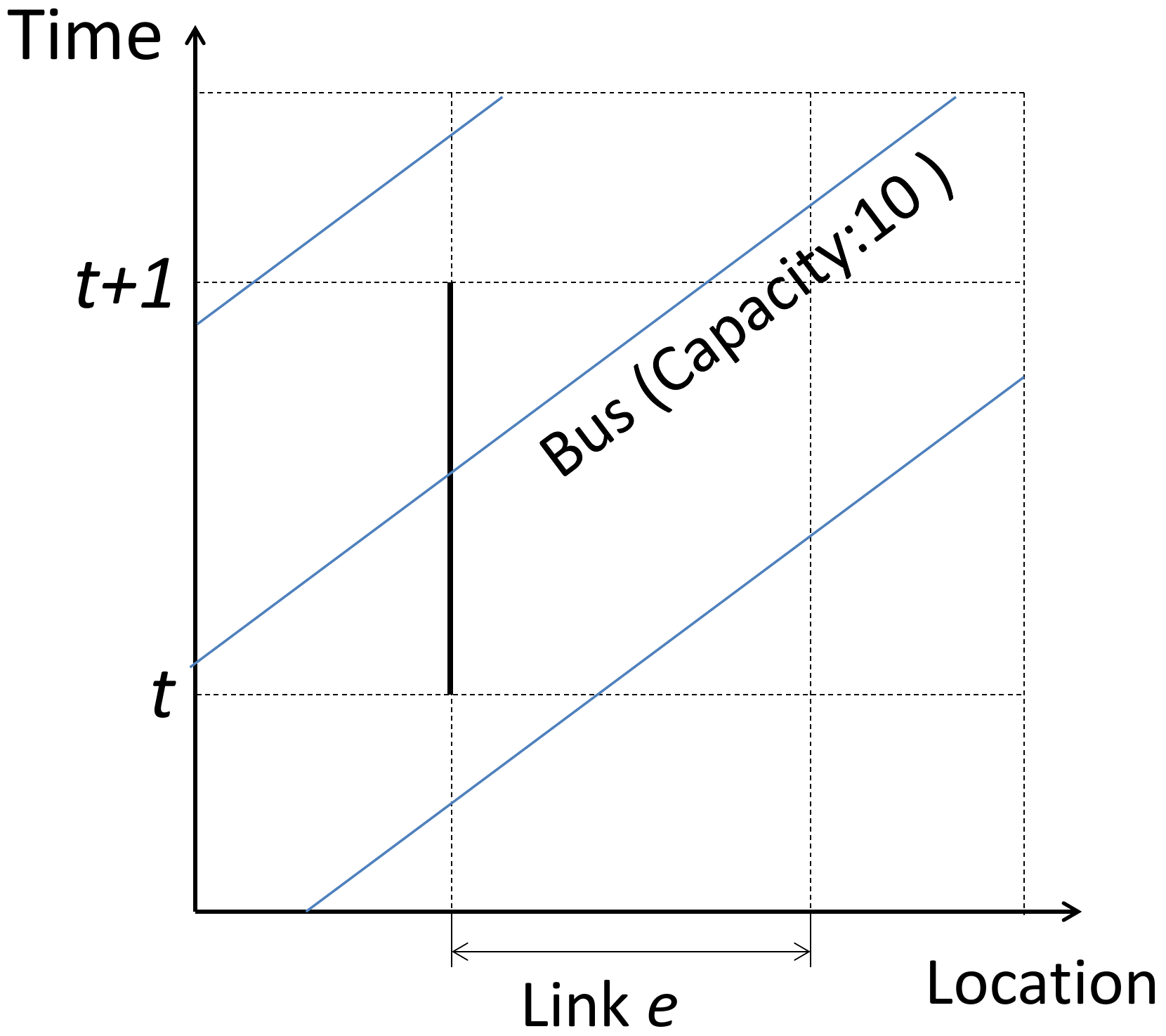}\\

	(b) Scheduled buses
	\end{center}
	\end{minipage}
	%2
	\begin{minipage}{0.33\hsize}
	\begin{center}
	    \includegraphics[height = 5.0cm]{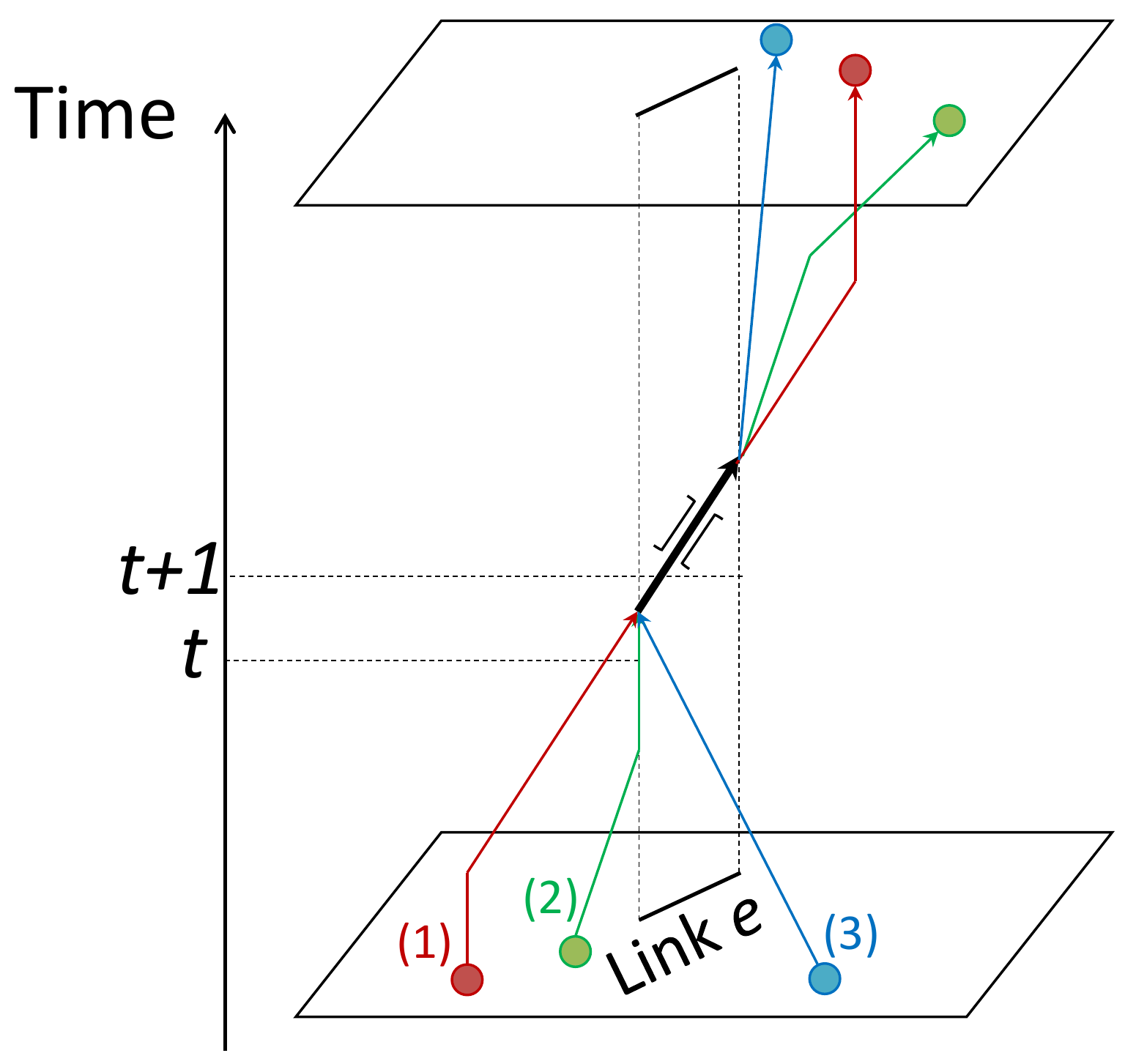}\\

	(C) Ride sharing
	\end{center}
	\end{minipage}

          \caption{Traffic resource capacity of various traffic mode} 
	    \label{fig_Capacity}
\end{center}
\end{figure}

In this part, we explain the \emph{capacity} of traffic resources. We use $C_{e,t}$ to denote the traffic capacity of edge $e \in \mathcal{E}$ at time $t \in T$, namely at most $C_{e,t}$ users can flow into edge $e$ at time $t$. Here, the capacity of traffic resources is a generalized concept of the road capacity for private vehicles, transportation capacity of scheduled bus services, acceptable volume of users of ride-share services, and so on. The traffic volume $F_{e,t}$ on edge $e$ at time $t$ is defined as $F_{e,t} = \sum_{i \in I} \delta(s_{i,t},a_{i,t},e)$ by Eq.~\ref{eq_link_flow}. For example, given the space-time trajectories of private vehicles as shown in Fig.\ref{fig_Capacity}(a), five vehicles flow into link $e$ at time $t$, and thus the traffic volume is obtained as $F_{e,t}=5$, assuming that each vehicle is occupied by only one person. In this case, the traffic capacity $C_{e,t}$ coincides with the link capacity, that is, the maximum number of vehicles that can flow into the link within a unit time, and is constant over time. In the case of scheduled buses, the space--time trajectories of which are shown in Fig.\ref{fig_Capacity}(b), the traffic volume is obtained as $F_{e,t} =10$, since one bus with the capacity of $10$ flows in to the link at time $t$. In this case, the traffic capacity $C_{e,t}$ is controlled by the service operator by changing the frequency and size of buses, and can be time-dependent. Finally, in the example of ride-share services, the space--time trajectories of which are shown in Fig.\ref{fig_Capacity}(c), the traffic volume is obtained as $F_{e,t}=3$ since three users share the vehicles that starts moving along the link $e$ at time $t$. In this case, the traffic capacity $C_{e,t}$ is also time-dependent and is determined by the combinatorial matching of trips by users. The method of bundling the trips of users are studied in many works \citep{santi2014quantifying,alonso2017demand,hara2017car}.

The assignment of the spatial road capacity for each traffic mode is also an important problem to be studied. For example, traffic operators can assign road capacity to private vehicles, scheduled buses or ride share services, by setting bus priority lanes or high-occupancy vehicle (HOV) lanes. With the widened reach of automated vehicles (AVs) as a result of which road capacity can possibly be increased, the assignment of priority lanes for AVs are also considered \citep{chen2016optimal}.
To make the discussion simple, in this paper, we assume that the traffic capacity of each traffic mode is given and represented by the capacity of edge $C_{e,t}$, note that an edge represents the transfer made by a traffic mode. Specifically, our proposed framework discusses the allocation of limited traffic resources to users with heterogeneous space--time prism constraints. However, this framework can be generalized to the settings where the capacity $C_{e,t}$ of each traffic mode is not given and only the constraints, such as road capacities, the number of vehicles, or total budgets, are given, which we discuss in Section~\ref{sec_discussions}.

% このキャパシティーは、自動車を走らせるための道路容量，路線バスの単位時間当たり輸送量，ライドシェアの最大許容数などの概念を一般化したものである。

% サービスオペレータは、空間的な道路を、路線バス、ライドシェア、自動車のそれぞれの割り当てることが可能である。例：バス専用道路、HOV、さらには自動運転車両専用道路。
%\citep{chen2016optimal}
%空間をどのように各モードのキャパシティーに割り振るかというのは、大きなテーマである。

%ここでは、各モードを代表する各エッジの容量$C_{e,t}$の下での、交通リソース配分を考える。

%%%%%%%%%
\subsubsection{Allocation of traffic resources}

Given the agents' report $\theta_i = \{ R_i, L_i\}$ at each time, the operator dynamically makes decisions on whether it accepts or rejects the requests and also makes decisions on the allocation of traffic resources to the accepted agents.
This problem is formulated as a bi-level problem in which the upper problem describes the decisions of the operator aiming to maximize social welfare, while the lower problem describes the decisions of agents regardless of whether they use the service or not. We provide an optimization model to solve this bi-level problem simultaneously, by introducing constraints for the operator, to guarantee that any agent $i$ can finish its trip within its time-space prism constraints in any situation, once the request is accepted. Specifically, one of the reported set of trip-plans $L_i$ is allocated ex-post to agent $i$ if accepted.

We use $I_t^{reported} = \{1,2,\ldots, \bar{I}_t \} \subset I$ to denote the set of all agents that report to the operator by time $t \in T$, and use $\bvec{\theta}_t = \{\theta_1,\theta_2,\ldots, \theta_{\bar{I}_t} \}$ to denote all the reports that the operator has received from the agents until time $t$. In $\bvec{\theta}_t$, the information about the reward functions $\bvec{R}_t = \{R_1,R_2,\ldots, R_{\bar{I}_t} \}$ and the information about the sets of trip-plan $\bvec{L}_t = \{L_1,L_2,\ldots, L_{\bar{I}_t} \}$ are included.
The operator makes all decisions based on the information it has at that time, specifically, the information $\bvec{\theta}_t$ and the observed states $\bvec{s}_t = \{s_{1,t},s_{2,t},\ldots, s_{{\bar{I}_t},t} \}$. The actions of all agents determined by the operator are denoted by $\bvec{\pi}_t = \{\pi_{1,t},\pi_{2,t},\ldots, \pi_{{\bar{I}_t},t} \}$, where $\pi_{i,t} \in \Gamma(s_{i,t})$ denotes the action that the operator allocates to agent $i$ at time $t$. 
Specifically, $\pi_{i,t} = l_i^C$ if the request of an agent $i$ is rejected. We use $I_t$ to denote the set of agents using the service at time $t$, meaning all agents with $t_i^{B} \leq t \leq t_i^{E}-1$, except agents that are rejected before time $t$.
We assume that the operator knows the state $s_{i,t}$ of agent $i \in I_t$ at time $t$, which is natural under the assumption that the operator provides mobility services for all accepted users. 
%We also assume that agents act based on the trip-plan generated by the operator.

%%%%%%%%%%%%%%%%%%%%%%%%%%%%%%%%%%%%%%%%%%%%%%%%%%%%%%%%%%%%%%%
\section{\label{sec_mechanisms}Mechanism}

Given the agent and the operator model stated in the previous section, we now show the traffic allocation mechanism of the service operator that guarantees to keep capacity constraints of traffic resources and space--time constraints of customers, even in the worst case, while aiming to maximize social welfare. We first show the offline optimal mechanism in a static setting and then discuss sequential mechanisms in a dynamic setting.

%%%%%%%%%%%%%%%%%%%%%
\subsection{Offline optimal mechanism}
First, we show the offline-optimal problem in which the operator knows all the information $\bvec{\theta} = \{ \bvec{R}, \bvec{L}\}$ of all agents beforehand, where $\bvec{\theta} = \{\theta_1,\theta_2,\ldots, \theta_{\bar{I}} \}$, $\bvec{R} = \{R_1,R_2,\ldots, R_{\bar{I}} \}$, and $\bvec{L} = \{L_1,L_2,\ldots, L_{\bar{I}} \}$. We define the social welfare as follows.

\begin{definition}[Social welfare] 
\label{def_SW} 
Social welfare is defined as the summation of rewards obtained by all agents over all time horizons. 
\end{definition}
Formally, social welfare $SW$ is given by:
\begin{equation}
SW(\bvec{\pi}|\bvec{s}_1)  = \sum_{i \in I} \sum_{t=t^B_i }^{t^E_i-1} R_i(s_{i, t}, \pi_{i, t},t) = \sum_{t =1}^{\bar{T}-1} R(\bvec{s}_t, \bvec{\pi}_t),
\label{eq_def_SW}
\end{equation}
where $\bvec{\pi} = \{\bvec{\pi}_1, \ldots, \bvec{\pi}_{\bar{T}}\}$ denotes the allocation to all agents at all times and $R(\bvec{s}_t, \bvec{\pi}_t) = \sum_{i \in I_t} R_i(s_{i,t}, \pi_{i,t},t) $ denotes the sum of the rewards of all agents given the allocation $\bvec{\pi}_t$ under the state $\bvec{s}_t$. Here, $\bvec{s}_1$ denotes the given initial state.

Then, we define the constraints in static settings. First, we define space--time prism constraints of users as follows.
\begin{definition}[Space--time prism constraints in static settings] 
\label{def_st_constraints_static} 
We say space--time prism constraints are satisfied in static settings, if the allocated trip plans for all accepted agents are included in a set of executable trip-plans for each agent. 
\end{definition}
Formally, it is given by;
\begin{equation}
\forall{i \in I}: (\pi_i|s_{i,t_i^{B}})= \{s_{i,t_i^{B}}, \pi_{i,t_i^{B}}, \ldots, s_{i,t_i^{E}-1}, \pi_{i,t_i^{E}-1}, s_{i,t_i^{E}} \} \in \{L_i \cup l_i^{C}\}
\label{eq_def_constraints}
\end{equation}
Specifically, the trip-plan allocated to agent $i$ is selected within a set of executable trip-plans $L_i$ or the request is rejected, otherwise.
We also define the capacity constraints as follows.
\begin{definition}[Capacity constraints in static settings] 
\label{def_cap_constraints_static} 
We say capacity constraints are satisfied in static settings, if the allocated trip plans for all agents do not violate the capacity constraints of traffic resources at any time and at any edge in the network. 
\end{definition}
Formally, it is given by;
\begin{equation}
\forall{e \in \mathcal{E}}, \forall{t \in T}: \sum_{i \in I} \delta(s_{i,t},\pi_{i,t},e) \leq C_{e,t},
\label{eq_def_cap_constraints}
\end{equation}

The left term of these constraints expresses the traffic volume that flows into the edge $e$ at time $t$, given by Eq.~\ref{eq_link_flow}.
Given these definitions, the offline optimal mechanism is defined as follows.

\begin{definition}[Offline optimal mechanism]
\label{def_offline_optimal} 
Offline optimal mechanism in our setting has an objective function that maximizes the social welfare $SW$ as defined by Definition~ \ref{def_SW} under space--time prism constraints defined by Definition~\ref{def_st_constraints_static} and capacity constraints defined by Definition~\ref{def_cap_constraints_static} . 
\end{definition}

Owing to the assumption that the operator has the perfect knowledge of future agents, the social welfare achieved by the offline optimal mechanism is the upper bound of all mechanisms in settings with limited information of future agents.

%Since we assume that the operator has the perfect knowledge of future agents, the offline optimal mechanism shows the upper bound of all mechanisms in our settings with limited information of future agents, in terms of the achieved social welfare.

%%%%%%%%%%%%%%%%%%%%%
\subsection{Sequential mechanisms}

Then, we show the sequential mechanism in which the operator can use reports $\bvec{\theta}_t = \{\theta_1,\theta_2,\ldots, \theta_{\bar{I}_t} \}$ from a limited set of agents $I_t^{reported} = \{1,2, \ldots, \bar{I}_t \} \subset I$ that has already reported, to decide the allocation $\bvec{\pi}_t = \{\pi_{1,t},\pi_{2,t},\ldots, \pi_{{\bar{I}_t},t} \}$ at time $t$. 
First, we introduce the definition of discounted social welfare based on dynamic programming \citep{rust1994structural}, as follows.
\begin{definition}[Discounted social welfare] 
\label{def_DSW}
The discounted social welfare at time $t$ is defined by the summation of discounted rewards obtained by all agents at and after time $t$. 
\end{definition}
Formally, the discounted social welfare ${DSW}_t$ at time $t$ is given by:
\begin{equation}
{DSW}_t = {DSW}(\bvec{\bvec{\pi}_t},\bvec{\bvec{\pi}_{t+1}}, \ldots, \bvec{\bvec{\pi}_{\bar{T}}}|\bvec{s}_t) = \sum_{t' = t}^{\bar{T}} \beta^{t'-t} R(\bvec{s}_{t'}, \bvec{\pi}_{t'}).
\label{eq_def_DSW}
\end{equation}

Then, we define the constraints in dynamic settings. First, we define space--time prism constraints of users as follows.
\begin{definition}[Space--time prism constraints in dynamic settings] 
\label{def_st_constraints_dynamic}
We say that space--time prism constraints are satisfied under the decision $\bvec{\pi}_t$ at time $t$, if there exists at least one executable trip plan for all existing agents, except agents that are rejected at time $t$.
\end{definition}
Formally, it is given by:
\begin{equation}
\forall{i \in I_t}, \exists{l_i^{\langle t \rangle}} \in \{L_i^{\langle t \rangle} \cup l_i^{C}\}, \pi_{i,t} \in l_i.
\label{eq_dynamic_constraints}
\end{equation}
We also define the capacity constraints in dynamic settings as follows.
\begin{definition}[Capacity constraints in dynamic settings] 
\label{def_cap_constraints_dynamic}
We say that capacity constraints are satisfied at time $t$, if there are joint trip plans across all accepted agents, that do not violate the capacity of traffic resources at any time in the future.
\end{definition}
Formally, it is given by;
\begin{equation}
\exists{\bvec{l}_{I'_t}^{\langle t \rangle} \in \bvec{L}_{I'_t}^{\langle t \rangle}}, \forall{e \in \mathcal{E}}, \forall{t' \geq t}: \sum_{i \in {I'_t}} \sum_{(s_{i,t'},\pi_{i,t'}) \in l_i^{\langle t \rangle}}\delta(s_{i,t'},\pi_{i,t'},e) \leq C_{e,t'},
%\exists{\bvec{l}^{\langle t \rangle} \in \bvec{L}^{\langle t \rangle}}, \forall{e \in \mathcal{E}}, \forall{t' \geq t}: \sum_{i \in {I'_t}} \sum_{\pi_{i,t'} \in l_i^{\langle t \rangle}}\delta(\mathcal{T}(\bvec{s}_{t'-1}, \pi_{t'-1}),\pi_{i,t'},e) \leq C_{e,t'},
\label{eq_dynamic_cap_constraints}
\end{equation}
where $I'_t$ denotes all agents under the service at time $t$ except agents that are rejected at that time, $\bvec{l}_{I'_t}^{\langle t \rangle}$ denotes a combination of trip plans by $I'_t$, and $\bvec{L}_{I'_t}^{\langle t \rangle}$ denotes a set of executable combinations of trip plans by $I'_t$. 

%Specifically, except the case where it is rejected, the allocated action $\pi_{i,t}$ for agent $i$ at time $t$ is included in any trip-plan $l_i$ in its executable set of plans $L_i$, and the capacity constraints are kept at any time in the future by all agents that take the plan.

Given these definitions, we define a class of mechanisms that always keep both, space--time prism constraints and capacity constraints in our dynamic settings, and call it \emph{Resources-Customers(RC)-feasible} mechanism.

\begin{definition}[RC-feasible mechanism] 
We say the mechanism is RC-feasible in our dynamic settings, if the mechanism satisfies the space--time prism constraints as defined by Definition~\ref{def_st_constraints_dynamic} and capacity constraints as defined by Definition~\ref{def_cap_constraints_dynamic} at any time. 
\end{definition}

It is the class of mechanisms that have arbitral objective functions and strictly keep the constraints given by Eqs.~\ref{eq_dynamic_constraints} and \ref{eq_dynamic_cap_constraints}. 
Moreover, we define \emph{Resources-Customers(RC)-optimal} mechanism as follows:

\begin{definition}[RC-optimal mechanism] 
\label{def_optimal_dynamic}
The RC-optimal mechanism in our dynamic setting has an objective function that maximizes the discounted social welfare $DSW_t$ as defined by Definition~\ref{def_DSW} and satisfies the space--time prism constraints as defined by Definition~\ref{def_st_constraints_dynamic} and capacity constraints as defined by Definition~\ref{def_cap_constraints_dynamic} . 
\end{definition}

To maximize the discounted social welfare $DSW_t$ given by Eq.~\ref{eq_def_DSW}, we introduce the V-value function $V(\bvec{s}_t)$ that represents the expected discounted social welfare under the state $\bvec{s}_t$ such that;
\begin{equation}
V(\bvec{s}_t) = \mathbb{E} \Biggl[ \sum_{t' = t}^{\bar{T}} \beta^{t'-t} R(\bvec{s}_{t'}, \bvec{\pi}_{t'}) \Biggr],
\end{equation}
where $\mathbb{E}[\cdot]$ denotes the mathematical operator that expresses the expectation.
We also introduce Q-value function $Q(\bvec{s}_t, \bvec{\pi}_t)$ defined as;
\begin{equation}
Q(\bvec{s}_t, \bvec{\pi}_t) = R(\bvec{s}_t, \bvec{\pi}_t) + \beta \cdot V(\mathcal{T}(\bvec{s_t}, \bvec{\pi}_t)),
\label{eq_def_Q_value}
\end{equation}
where $\mathcal{T}(\bvec{s_t}, \bvec{\pi}_t)$ denotes the state at time $t+1$ given the state $\bvec{s_t}$ and action $\bvec{\pi}_t$ at time $t$. This Q-value function $Q(\bvec{s}_t, \bvec{\pi}_t)$ represents the expected discounted social welfare ${DSW}_t$ given the allocation $\bvec{\pi}_t$ under the state $\bvec{s}_t$. 
By Bellman's principle of optimality \citep{bellman1957dynamic}, the objective function of the \emph{RC-optimal} mechanism as defined by Definition~\ref{def_optimal_dynamic} is reformulated in a recursive-form such that:
\begin{equation}
 \max_{\bvec{\pi}_t \in \Gamma(\bvec{s_t})} Q(\bvec{s}_t, \bvec{\pi}_t).
\label{eq_obj_DSW}
\end{equation}
The RC-optimal mechanism maximizes the Q-value function under the constraints given by Eqs.~\ref{eq_dynamic_constraints} and \ref{eq_dynamic_cap_constraints} at each time.
This mechanism is an instance of the RC-feasible mechanism. Another instance of the RC-feasible mechanism is the FCFS mechanism that allocates the traffic resource myopically to the early coming agents and rejects in case the request violates the capacity constraints, which is common in this kind of a setting.
In addition to these two extreme examples, there exists a wide range of RC-feasible algorithms that reasonably gives priority to later-coming high-value agents in its decision of allocation, unlike the FCFS mechanism that always gives priority to the early coming agents. 
We call a booking system operated by a RC-feasible algorithms as the \emph{Floating booking system}, in the sense that it does not totally fix the trip-plan of agents at the time of accepting the booking, but has some flexibility to accept any other high-valued agent, while guaranteeing to keep the constraints of all agents and traffic resources. We show algorithms for establishing such systems in Section~\ref{sec_algorithm}.

%%%%%%
\subsection{\label{sec_mechanism_dynamic_agent}RC-feasible algorithm considering flexible behaviors of agents}

In this part, we briefly present an example of the RC-feasible algorithm in \emph{dynamic agent-type} settings stated in Section~\ref{sec_flexible_behavior}, assuming that the agent can report a new type $\theta'_i = \{ R'_i, L'_i\}$ to the operator in an arbitrary timing.
We consider an agent $i$ whose request $\theta_i = \{ R_i, L_i\}$ has been accepted and is on the trip at time $t$. At this point, an RC-feasible algorithm guarantees that this agent can reach its destination $D_i$ by time $t_i^E$, within a set of trip-plan $L_i^{\langle t \rangle}$ based on its space--time constraints. Assume that this agent changes its mind and reports a new type $\theta'_i = \{ R'_i, L'_i\}$ to the operator.
Receiving this report, the mechanism makes the following decision.
\begin{itemize}
\item(1) Case 1, where $L'_i \supseteq L_i^{\langle t \rangle}$, namely the space--time constraints are relaxed, the mechanism accepts the new type of agents immediately and makes a decision based on this new type $\theta'_i = \{ R'_i, L'_i\}$.
\item (2) Case 2, where $L'_i \nsupseteq L_i^{\langle t \rangle}$, namely a new set of trip-plans are not a superset of the original set of trip-plans, the mechanism makes a judgment on whether it accepts this request or not. Specifically, the mechanism temporarily processes the decision at that time based on new type $\theta'_i = \{R'_i, L'_i\}$ and reports by other agents. If the mechanism could find a solution that does not violate the capacity constraints, the operator accepts the new request and makes the following decisions based on the new type $\theta'_i = \{R'_i, L'_i\}$. Otherwise, the operator makes the following decisions based on the original type $\theta_i = \{ R_i, L_i\}$.
%with the option that it does not accept the new request but keeps the original constraints $L_i$ for agent $i$. If the resulting trip-plan for agent $i$ is included in the new set of trip-plans $L'_i$, 
\end{itemize}

The algorithm shown above is RC-feasible in the \emph{dynamic agent-type} settings. Note that if the agent's new request is $L'_i=l_i^C$, that is the agent want to renounce the trip and not use any traffic resources in the future, the request is always accepted and thus, users can drop off the trip at any time.

%%%%%%%%%%%%%%%%%%%%%%%%%%%%%%%%%%%%%%%%%%%%%%%%%%%%%%%%%%%%%%%%%%%%%%%%%%%%%%%%%%%%%%%%%%%%%%%%%%%%%%%%%%%%%%%%
%
\section{\label{sec_algorithm}Solution Algorithm}									
%
%
%%%%%%%%%%%%%%%%%%%%%%%%%%%%%%%%%%%%%%%%%%%%%%%%%%%%%%%%%%%%%%%%%%%%%%%%%%%%%%%%%%%%%%%%%%%%%%%%%%%%%%%%%%%%%%

In the previous sections, we show a MaaS model and characterize the RC-feasible mechanism for establishing the floating booking system in the service. The mechanism solves two problems simultaneously, namely, deciding a set of accepted agents and allocating traffic resources. 
In this section, we propose solution algorithms to realize such a system that is implemented to the central server of the operator and the distributed servers connecting to users' mobile apps.
As shown in Eqs.~\ref{eq_dynamic_constraints} and \ref{eq_dynamic_cap_constraints}, this problem includes non-linear constraints. Thus, we propose using the ZDD \citep{minato1993zero} to solve this problem.

%%%%%%%%%%%%%%%%%%%%
%\subsection{User agent algorithm}

\begin{figure}[t]
\begin{center}
	%2
	\begin{minipage}{0.6\hsize}
	\begin{center}
	    \includegraphics[height = 5.0cm]{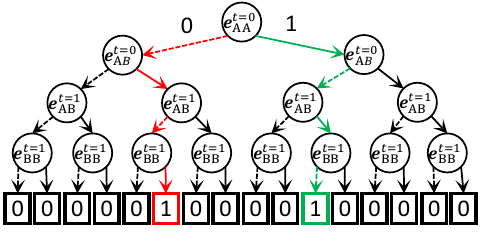}\\

	(a) Trip-plans expressed by the binary decision tree
	\end{center}
	\end{minipage}
	%3
	\begin{minipage}{0.38\hsize}
	\begin{center}
	    \includegraphics[height = 5.0cm]{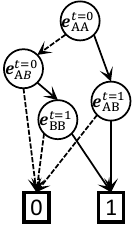}\\

	(b) Trip-plans expressed by the ZDD
	\end{center}
	\end{minipage}
          \caption{Trip-plan enumeration using the ZDD} 
	    \label{fig_ZDD_example}
\end{center}
\end{figure}

\begin{figure}[p]
\begin{center}

   \includegraphics[width = 0.8\hsize]{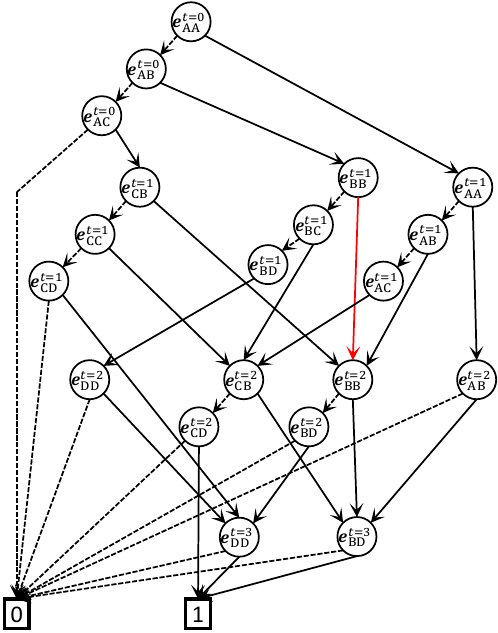}\\
   \caption{Enumerated trip-plans for the setting in Fig.~\ref{fig_TimeSpacePrism_sample}(a - iii), by using the ZDD} 
   \label{fig_ZDD_example2}

\end{center}
\end{figure}

We first show the algorithm of user agents that is rather simple. The user agents make a set of executable trip-plans $L_i$ under space--time prism constraints, based on the users' request, or the information of the active logger. To do this, agents execute a full search of the routes in the space--time expanded network, such as for example, as shown in Fig.~\ref{fig_TimeSpacePrism_sample}. 
This process is computed efficiently by using ZDD \citep{minato1993zero}. For example, when we consider the network with only 2 nodes, Node A and Node B, and a link AB with the link cost $\tau_{AB} = 1$, the trip-plans of an agent that start from Node A at $t=0$ and arrive at Node B at $t=2$ is enumerated as shown in Fig.~\ref{fig_ZDD_example}(a), using the binary decision tree. In the figure, for example, $e_{AB}^{t=0}$ express the action moving from Node A to Node B at $t=0$, and  $e_{AA}^{t=0}$ express the action staying at Node A  at $t=0$. The arrows express the binary-decisions, specifically, the solid arrows mean that the action is taken while the dashed arrows mean that the action is not taken. The numbers in $\{0,1\}$ shown in the bottom express whether the plan is executable or not. For example, a trip plan displayed in red in the figure shows a plan that an agent moves from Node A to Node B at $t = 0$ and stays at Node B at $t = 1$. In this case, two trip-plans shown in red and green in Fig.~\ref{fig_ZDD_example}(a) are executable. By using the ZDD, this binary-decision tree is represented in a compact form as shown in Fig.~\ref{fig_ZDD_example}(b). Various manipulations such as subset, union, and cross-product are efficiently computed while keeping the compact form.
As another example, in Fig.~\ref{fig_ZDD_example2}, we show the enumerated trip-plans in the setting as shown in Fig.~\ref{fig_TimeSpacePrism_sample}(a-iii). By using the ZDD, all of the 15 trip-plans considered in the settings where $T_i = \{ 0, 1, 2, 3, 4\}$ are compactly enumerated. If the agent has to stay at Node B at $t=1$, as considered in Fig.~\ref{fig_TimeSpacePrism_sample}(b-iii), the trip-plans have to include the arrow shown in red in Fig.~\ref{fig_ZDD_example2} and thus only 2 plans are extracted.
This process is computed on distributed small servers that communicate with users' mobile apps, and the obtained set of trip-plans are held in the compact form. The user agent sends this set of plans $L_i$ as well as the reward function $R_i$ to the central server.
%\textcolor{blue}{とはいえ、携帯端末というよりかは、計算は何らかの分散サーバー側で}

In the following parts, we present the algorithm of the central server. We first show the algorithm of the RC-optimal mechanism in myopic and non-myopic settings. Then, we show the approximation algorithms to solve it in a computationally low cost manner, which are classified into the RC-feasible mechanism. After that, we discuss the trade-off among the wide range of algorithms that we provide. In this section, we assume \emph{static agent-type} settings, where an agent reports its type $\theta_i = \{ R_i, L_i\}$ at the beginning of the trip and it remains unchanged until it ends, but the discussion can easily be extended to \emph{dynamic agent-type} settings.

%%%%%%%%%%%%%%%%%%%%
\subsection{Myopic algorithm}

First, we discuss the myopic algorithm that myopically makes decisions without using any information on future agents. The overall algorithm is shown in Algorithm~\ref{algo_main}.

%まず、将来の需要予測などを考慮しない、 Myopic なアルゴリズムについて述べる。最初に、アルゴリズムの全体像を示す。

\begin{algorithm}
\caption{Main (Myopic)}
\label{algo_main}
\begin{algorithmic}[1]
\REQUIRE $T, \mathcal{G}=(\mathcal{N},\mathcal{E}), I, \bvec{\theta}=\{\bvec{R}, \bvec{L} \}$
\ENSURE $\bvec{\pi}$
\STATE $t \leftarrow 0$
\STATE $\mathcal{Z} \leftarrow \emptyset $
\FOR{$t \in T$}
	\STATE $I_t, \bvec{\theta}_t \leftarrow \text{\sc{RenewAgent}}(I, \bvec{\theta}, t)$
	\STATE $\mathcal{Z} \leftarrow \text{\sc{MyopicUpdate}}(\mathcal{Z}, \bvec{L}_t)$
	\STATE $\bvec{\pi}_t \leftarrow \text{\sc{MyopicOptimalAllocation}}(\mathcal{Z}, \bvec{R}_t , \bvec{s}_t)$
	\STATE $\mathcal{Z} \leftarrow \text{\sc{NarrowDown}}(\mathcal{Z}, \bvec{\pi}_t)$
\ENDFOR
\end{algorithmic}
\end{algorithm}

The initial inputs to the algorithm are control time duration $T = \{1,2,\ldots, \bar{T}\}$ and spatial network $\mathcal{G}=(\mathcal{N},\mathcal{E})$. In addition, while receiving reports from agents at each time, the information on a set of agents $I$ and their type $\bvec{\theta}=\{\bvec{R}, \bvec{L}\}$ are dynamically updated (Line 4).
At each time, a set of executable trip-plans for all agents are updated based on $\bvec{L}_t$ and saved to $\mathcal{Z}$, using ZDD \citep{minato1993zero} (Line 5). We use $|\mathcal{Z}|$ to express the number of plans saved in $\mathcal{Z}$. 
Based on the reward functions $\bvec{R}_t$ and the state $\bvec{s}_t$, the algorithm determines an allocation $\bvec{\pi}_t$ within a set of executable  trip-plans $\mathcal{Z}$ (Line 6). Given the fixed allocation $\bvec{\pi}_t$ at each time, the set of executable plans $\mathcal{Z}$ is narrowed down (Line 7).
In the following parts, we highlight each part of this algorithm in detail.

%%%%%%%%%
\subsubsection{Function {\sc MyopicUpdate}}

\begin{algorithm}
\caption{MyopicUpdate}
\label{algo_update}
\begin{algorithmic}
\REQUIRE $\mathcal{Z}, \bvec{L}_t = \{L_1, L_2, \ldots L_{\bar{I}_t} \}$
\ENSURE $\mathcal{Z}$
\FOR{$i \in I_t$}
	\STATE $\mathcal{Z} \leftarrow \text{\sc{ClossProduct}}(\mathcal{Z}, L_i)$
	\STATE $\mathcal{Z} \leftarrow \mathcal{Z} \setminus \text{\sc{CapacityViolationSet}}(\mathcal{Z})$
\ENDFOR
\RETURN $\mathcal{Z}$
\end{algorithmic}
\end{algorithm}

We show the pseudo code of Function {\sc MyopicUpdate} in Algorithm~\ref{algo_update}. This function updates the set of plans $\mathcal{Z}$ based on the reported sets of trip-plans $\bvec{L}_t = \{L_1, L_2, \ldots L_{\bar{I}_t} \}$. Given a set of trip-plans $L_i$ of agent $i$, the algorithm first takes the cross-product of $L_i$ and $\mathcal{Z}$, which is a set of plans for agents that have already been processed. At this point, all trip-plans included in $\mathcal{Z}$ keeps the space--time prism constraints given by Eq.~\ref{eq_dynamic_constraints} for all agents.
However, some plans obtained by this process violate the capacity constraints given by Eq.~\ref{eq_dynamic_cap_constraints}, so the algorithm extracts such plans by the function {\sc CapacityViolationSet}$(\mathcal{Z})$ and removes these plans from $\mathcal{Z}$. This process is repeated for all agents. 
Thus, we can get the full-searched executable trip-plans that keep constraints, as given by Eqs.~\ref{eq_dynamic_constraints} and \ref{eq_dynamic_cap_constraints}.
The algorithm that is described above can be computed efficiently using ZDD. 

%%%%%%%%%
\subsubsection{Function {\sc MyopicOptimalAllocation}}

\begin{algorithm}
\caption{MyopicOptimalAllocation}
\label{algo_allocation}
\begin{algorithmic}
\REQUIRE $\mathcal{Z}, \bvec{R}_t , \bvec{s}_t$
\ENSURE $\bvec{\pi}_t$
\STATE $\Gamma(\bvec{s_t}) \leftarrow \text{\sc{JointAction}}(\bvec{s}_t)$
\STATE $\mathcal{Z}^w \leftarrow \text{\sc{Weight}}(\mathcal{Z}, \bvec{R}_t)$
%\STATE $\forall{\bvec{a} \in \bvec{A}_t}: V_{\bvec{a}} \leftarrow 0$
\STATE $\bvec{\pi}_t \leftarrow \text{\sc{Optimize}}(\bvec{s}_t, \Gamma(\bvec{s_t}), \mathcal{Z}^w)$
\RETURN $\bvec{\pi}_t$
\end{algorithmic}
\end{algorithm}

We show the pseudo code of Function {\sc MyopicOptimalAllocation} in Algorithm~\ref{algo_allocation}. 
First, the algorithm enumerates a set of joint actions $\Gamma(\bvec{s_t})$ that could be taken by all agents based on state $\bvec{s}_t$ at the current time step $t$. 
Then, the algorithm weights the reward $\bvec{R}_t$ to the set of plans $\mathcal{Z}$. The obtained $\mathcal{Z}^w$ is a weighted graph saved in the ZDD structure. Using this, the algorithm determines the best allocation $\bvec{\pi}_t$ based on Eq.~\ref{eq_obj_DSW}. 
These processes are executed by search-based optimization using the ZDD.

%%%%%%%%%
\subsubsection{Function{\sc NarrowDown}}

\begin{algorithm}
\caption{NarrowDown}
\label{algo_narrowdown}
\begin{algorithmic}
\REQUIRE $\mathcal{Z}, \bvec{\pi}_t$
\ENSURE $\mathcal{Z}$
	\STATE $\mathcal{Z} \leftarrow \text{\sc{Including}}(\mathcal{Z}, \bvec{\pi}_t)$	
\RETURN $\mathcal{Z}$
\end{algorithmic}
\end{algorithm}

Finally, we show the pseudo code for the {\sc NarrowDown} function in Algorithm~\ref{algo_narrowdown}. As the joint action of all agents is determined as $\bvec{\pi}_t$ by the mechanism at each time, this function extracts the executable plans, including the determined joint action from $\mathcal{Z}$. This process can also be computed efficiently using the ZDD.

%%%%%%%%%%%%%%%%%%%%%
\subsection{Non-myopic algorithm}

Then, we discuss the non-myopic algorithm that makes decisions using information of estimated future demand. We assume that the stochastic and elastic demand $\tilde{\mathcal{D}}$ is given to the operator in advance although the actual demand is not revealed until the reports of user agents are received. For instance, in a car-sharing service for a small number of predescribed members, the space--time constraints of each member can be well estimated by observing the daily-activity of the member.

%次に、交通管理者は将来の需要に対する確率的な推定モデル$\tilde{\mathcal{D}}$を有していることを想定して、Non-myopic なアルゴリズムを考える。
%例えば、会員制のサービスなどでは、会員の日常的なルーティンを観測することによって、その時空間プリズム制約を高い確率で予測できると考えられる。

\begin{algorithm}
\caption{Main (Non-myopic)}
\label{algo_main_nonmyopic}
\begin{algorithmic}[1]
\REQUIRE $\tilde{\mathcal{D}}, T, \mathcal{G}=(\mathcal{N},\mathcal{E}), I, \bvec{\theta}=\{\bvec{R}, \bvec{L}\}$
\ENSURE $\bvec{\pi}$
\STATE $t \leftarrow 0 $
\STATE $\mathcal{Z}, \tilde{\bvec{L}} \leftarrow \text{\sc{Enumerate}}(\tilde{\mathcal{D}}, T, \mathcal{G} )$
\FOR{$t \in T$}
	\STATE $I_t, \bvec{\theta}_t \leftarrow \text{\sc{RenewAgent}}(I, \bvec{\theta}, t)$
	\STATE $\mathcal{Z} \leftarrow \text{\sc{NonMyopicUpdate}}(\mathcal{Z}, \tilde{\bvec{L}}_t, \bvec{L}_t)$
	\STATE $\bvec{\pi}_t \leftarrow \text{\sc{NonMyopicOptimalAllocation}}(\tilde{\mathcal{D}}, \mathcal{Z}, \bvec{R}_t , \bvec{s}_t)$
	\STATE $\mathcal{Z} \leftarrow \text{\sc{NarrowDown}}(\mathcal{Z}, \bvec{\pi}_t)$
\ENDFOR
\end{algorithmic}
\end{algorithm}

We show an overview of the algorithm in Algorithm~\ref{algo_main_nonmyopic}.
Basically, this algorithm is similar to the myopic algorithm shown in Algorithm~\ref{algo_main}. The differences are that the non-myopic algorithm enumerates an estimated set of plans based on the given stochastic demand $\tilde{\mathcal{D}}$ (Line 2), it uses the constraints $\tilde{\bvec{L}_t}$ derived from $\tilde{\mathcal{D}}$ when updating the set of plans (Line 5), and it uses $\tilde{\mathcal{D}}$ when making decisions on the allocation (Line 6). 
% and it also uses $\tilde{\mathcal{D}}$ in the process of updating the set of plans (Line 5) and making decisions on the allocation (Line 6). 
In the following parts, we show these processes.

%アルゴリズムは、myopic アルゴリズムと非常に似ている。違いは、事前に予測需要$\tilde{\mathcal{D}}$に基づいてトリッププラン集合を生成すること(Line 2)，The algorithm first enumerates all possible allocation plans in control duration $T$ and saves the set of plans to $\mathcal{Z}$ using ZDD\citep{minato1993zero} (Line 2). 
%
%および，実行可能プランのアップデート (Line 5)および，アロケーションの決定 (Line 6)の時に、需要モデル$\tilde{\mathcal{D}}$を参照することである。これらの違いについて、以下のパートで詳述する。

%%%%%%%%%
\subsubsection{Function {\sc Enumerate}}

\begin{algorithm}
\caption{Enumerate}
\label{algo_enumerate}
\begin{algorithmic}
\REQUIRE $\tilde{\mathcal{D}}, T, \mathcal{G}=(\mathcal{N}, \mathcal{E})$
\ENSURE $\mathcal{Z}$
\STATE $\mathcal{Z} \leftarrow \emptyset$
\STATE $\tilde{\bvec{L}} \leftarrow \text{\sc{SampleConstraints}}(\tilde{\mathcal{D}}, T, \mathcal{G})$
\STATE $\mathcal{Z} \leftarrow \text{\sc{MyopicUpdate}}(\mathcal{Z}, \tilde{\bvec{L}})$
\RETURN $\mathcal{Z}, \tilde{\bvec{L}}$
\end{algorithmic}
\end{algorithm}

We show the pseudo code for the {\sc Enumerate} function in Algorithm~\ref{algo_enumerate}. This function aims to enumerate executable plans by possibly attending future agents in control time $T$. 
First, it generates space--time constraints of future agents using the stochastic demand model $\tilde{\mathcal{D}}$, and enumerates the sets of trip-plans $\tilde{\bvec{L}}$ for each agent that keeps the constraints. We call this set of plans $\tilde{\bvec{L}}$ as virtual plans.
 Then, it generates executable trip-plans $\mathcal{Z}$ of all agents by applying the process of function {\sc MyopicUpdate} that we have shown in Algorithm~\ref{algo_update}.

%ここでは、将来にやってくる可能性が高いエージェントについて、需要モデルからその時空間プリズム制約を推定し、実行可能なプランのセットを生成する（関数{\sc SampleConstraints}）。これらに対して、Algorithm~\ref{algo_update}に示した関数{\sc MyopicUpdate}を適用することによって、それらのエージェントの全てが同時に実行可能な行動プランを列挙する。

%%%%%%%%%
\subsubsection{Function {\sc NonMyopicUpdate}}

\begin{algorithm}
\caption{NonMyopicUpdate}
\label{algo_update_nonmyopic}
\begin{algorithmic}
\REQUIRE $\mathcal{Z}, \tilde{\bvec{L}}_t, \bvec{L}_t = \{L_1, L_2, \ldots L_{\bar{I}_t} \}$
\ENSURE $\mathcal{Z}$
\STATE $\bvec{L}' \leftarrow \bvec{L}_t \cap \tilde{\bvec{L}}_t$
\STATE $\mathcal{Z} \leftarrow  \text{\sc{Remove}}(\tilde{\bvec{L}}_t \setminus \bvec{L}')$
\STATE $\mathcal{Z} \leftarrow \text{\sc{MyopicUpdate}}(\mathcal{Z}, \bvec{L}_t \setminus \bvec{L}')$
\RETURN $\mathcal{Z}$
\end{algorithmic}
\end{algorithm}

We show the pseudo code of Function {\sc NonMyopicUpdate} in Algorithm~\ref{algo_update_nonmyopic}. 
This function updates the set of executable trip-plans $\mathcal{Z}$ by using the information of a set of trip-plans $\bvec{L}_t$ that are actually reported at time $t$. It first takes an intersection $\bvec{L}' = \bvec{L}_t \cap \tilde{\bvec{L}}_t$ of the actual plans and virtual plans, and then changes the unrealized virtual plans $\tilde{\bvec{L}}_t \setminus \bvec{L}'$ to the unexpected actual plans $\bvec{L}_t \setminus \bvec{L}'$ in the executable trip-plans $\mathcal{Z}$.

This function does not execute anything if the estimated virtual demands are precise, that is $\tilde{\bvec{L}}_t = \bvec{L}_t$, while it coincides with the Function {\sc MyopicUpdate} if the virtual demands are not provided, that is $\tilde{\bvec{L}}_t = \emptyset$.
Thus, it is reasonable to divide the computational efforts into two functions, {\sc Enumerate} that can be processed in advance and {\sc NonMyopicUpdate} that requires adaptive processes to the real-time requests. In cases in which the operator estimates the constraints of users precisely, the computational efforts in real-time processes can be decreased by using the Function {\sc Enumerate} effectively.

%関数{\sc NonMyopicUpdate}では、前述の関数{\sc Enumerate}で生成した仮想エージェントの実行プラン$\tilde{\bvec{L}}$のうち、時刻$t$に行動を開始する集合$\tilde{\bvec{L}}_t$を、実際の時刻$t$のレポート$\bvec{L}_t$と比較する。
%まずその共通部分$\bvec{L}' = \bvec{L}_t \cap \tilde{\bvec{L}}_t$をとり、過剰な仮想エージェントのプラン$\tilde{\bvec{L}}_t \setminus \bvec{L}'$を$\mathcal{Z}$から取り除く。その後、見積もられていなかったエージェント$\bvec{L}_t \setminus \bvec{L}'$の実行プランを$\mathcal{Z}$に追加する。

%この関数は、見積もりが全くない、すなわち、 $\tilde{\bvec{L}}_t = \emptyset$のときには、関数{\sc MyopicUpdate}と一致する。また、見積もりが正確、すなわち$\bvec{L}_t = \tilde{\bvec{L}}_t$であれば何も実行しない。

%見積もりの精度に応じて、計算負荷を関数{\sc Enumerate}と関数{\sc NonMyopicUpdate}に分散することが望まれる。関数{\sc NonMyopicUpdate}はオンライン処理が必要であるのに対して、関数{\sc Enumerate}は事前のオフライン処理で対応できるため、予測精度が高いエージェントについてはなるべく関数{\sc Enumerate}で対応することで、オンラインの計算負荷を削減することが可能である。

%%%%%%%%%
\subsubsection{Function{\sc NonMyopicOptimalAllocation}}

\begin{algorithm}
\caption{NonMyopicOptimalAllocation}
\label{algo_allocation_nonmyopic}
\begin{algorithmic}
\REQUIRE $\tilde{\mathcal{D}}, \mathcal{Z}, \bvec{R}_t , \bvec{s}_t$
\ENSURE $\bvec{\pi}_t$
\STATE $\Gamma(\bvec{s_t}) \leftarrow \text{\sc{JointAction}}(\bvec{s}_t)$
\STATE $\forall{\bvec{a} \in \Gamma(\bvec{s_t})}: Q_{\bvec{a}} \leftarrow 0$
\FOR{$m = 1$ to $M$}
	\STATE $\tilde{\bvec{R}} \leftarrow \text{\sc{SampleRewards}}(\tilde{\mathcal{D}})$
	\STATE $\mathcal{Z}^w \leftarrow \text{\sc{Weight}}(\mathcal{Z}, (\bvec{R}_t \cup \tilde{\bvec{R}}))$
	\FOR{$\bvec{a} \in \Gamma(\bvec{s_t})$}
		\STATE $Q_{\bvec{a}} \leftarrow Q_{\bvec{a}} + \text{\sc{GetQValue}}(\bvec{s}_t, \bvec{a}, \mathcal{Z}^w)$
	\ENDFOR
\ENDFOR
\STATE $\bvec{\pi}_t \leftarrow \argmax_{\bvec{a} \in \Gamma(\bvec{s_t})} Q_{\bvec{a}}$
\RETURN $\bvec{\pi}_t$
\end{algorithmic}
\end{algorithm}

We show the pseudo code of Function {\sc NonMyopicOptimalAllocation} in Algorithm~\ref{algo_allocation_nonmyopic}. 
Comparing it to the Function {\sc MyopicOptimalAllocation}, this function takes a multi-scenario approach \citep{chang2000line} to calculate the efficient allocation that maximizes the discounted social welfare based on given estimated demands $\tilde{\mathcal{D}}$.

It samples $M$ sets of reward functions $\tilde{\bvec{R}}$ for future agents based on the demand model $\tilde{\mathcal{D}}$. Then, it unifies those samples with already-reported reward functions $\bvec{R}_t$ of agents until the current time step and weights them to the set of plans $Z$. The obtained $Z^w$ is a weighted graph that is saved in the ZDD structure. Using this, the algorithm calculates the Q-value given by Eq.~\ref{eq_def_Q_value} for each joint action $\bvec{a} \in \Gamma(\bvec{s_t})$. After repeating this process for all $M$ sample scenarios, the algorithm selects the action that has the maximum Q-value on average as the efficient allocation at the current time step.

Since the function {\sc GetQValue} needs optimized calculation, the number of optimized calculations that this function needs is the number of sample scenario times the number of joint actions, including in $\Gamma(\bvec{s_t})$. However, the graphs structure of $Z$ is not changed throughout these processes, and only the weight is changed for each scenario. Thus, it can be calculated efficiently by using the ZDD.

%%%%%%%%%%%%%%%%%%%%%
\subsection{\label{sec_approximation_mechanism}Approximation algorithm}

Although the algorithm proposed in the previous part can be treated in a compact form by using ZDD, it requires large computation effort as the number of agents increases, because of the combinatorial nature of the problem. In this part, we show approximation algorithms to decrease the computational burden.

%%%%%%%%
\subsubsection{Making decisions for each agent}

In the Function {\sc MyopicOptimalAllocation} (Algorithm~\ref{algo_allocation}) or the Function {\sc NonMyopicOptimalAllocation} (Algorithm~\ref{algo_allocation_nonmyopic}), the joint actions $\Gamma(\bvec{s_t}) $ of all agents existing at time $t$ are considered. However, the size of $\Gamma(\bvec{s_t}) $ suffers from the combinatorial explosion in the number of agents. To avoid this, we propose an approximate algorithm that makes decisions for each agent separately. The pseudo code of this approximate algorithm for myopic allocation is shown in Algorithm~\ref{algo_allocation_approximate}. The approximation algorithm for non-myopic allocation can also be considered similarly. We call this approximation algorithm as \emph{Per agent}. 

\begin{algorithm}
\caption{MyopicAllocation - PerAgent}
\label{algo_allocation_approximate}
\begin{algorithmic}
\REQUIRE $\mathcal{Z}, \bvec{R}_t , \bvec{s}_t$
\ENSURE $\bvec{\pi}_t$
\STATE $\mathcal{Z}^w \leftarrow \text{\sc{Weight}}(\mathcal{Z}, \bvec{R}_t)$
\FOR{$i \in I_t$}
	\STATE $\Gamma(s_{i,t}) \leftarrow \text{\sc{JointAction}}(s_{i,t})$
	\STATE $\pi_{i,t} \leftarrow \text{\sc{Optimize}}(s_{i,t}, \Gamma(s_{i,t}), \mathcal{Z}^w)$
	\STATE $\bvec{\pi}_t \leftarrow \text{\sc{Set}}(\pi_{i,t})$
	\STATE $\mathcal{Z}^w \leftarrow \text{\sc{NarrowDown}}(\mathcal{Z}^w, \pi_{i,t})$
\ENDFOR
\RETURN $\bvec{\pi}_t$
\end{algorithmic}
\end{algorithm}

Although this algorithm makes decisions for each agent's action separately, it collects the information of the reward function $\bvec{R}_t$ for all agents that reports at time $t$ beforehand and uses it to achieve efficient allocation.

%%%%%%%%
\subsubsection{Branch cutting}

By introducing the limitation on the number of options, we can decrease the computational effort reasonably. We introduce the number of maximum branches $N_{branch}^{max}$ in our mechanism. Namely, the number of options that the enumerated plan $\mathcal{Z}$ always keeps;
\begin{equation}
	|\mathcal{Z}| \leq N_{branch}^{max}. 
\end{equation}
If the number of options exceeds $N_{branch}^{max}$, the only $N_{branch}^{max}$ options are left and the rest are cut off. There are various ways to select $N_{branch}^{max}$ options, for example, to select the highest $N_{branch}^{max}$ options in terms of rewards, to select $N_{branch}^{max}$ options randomly, or to use a combination of these two methods. We call this approximate algorithm as \emph{Branch cutting}.

%%%%%%%%%%%%%%%%%%%%%
\subsection{Property of proposed solution algorithms}

In this part, we show the property of algorithms stated in the previous parts.

\begin{prop} 
The exact algorithm shown in Algorithm~\ref{algo_main} is RC-optimal in myopic settings without any information related to future agents.
\label{thm_myopic_optimal}
\end{prop}

\begin{proof}
The function {\sc MyopicUpdate} shown in Algorithm~\ref{algo_update} provides a complete set of executable trip-plans for all agents that have reported, which satisfies the space--time constraints as defined by Eq.~\ref{eq_dynamic_constraints} and capacity constraints as defined by Eq.~\ref{eq_dynamic_cap_constraints}. Then, the decisions are made by the function MyopicOptimalAllocation as shown in Algorithm~\ref{algo_allocation} that optimize the discounted social welfare based on Eq.~\ref{eq_obj_DSW} by a full-search using ZDD. Thus, the solution of the Algorithm~\ref{algo_main} is RC-optimal, in myopic settings.
\end{proof}

\begin{prop} 
Assuming that the set of estimated trip-plans $\tilde{\bvec{L}}$ obtained by the function {\sc Enumerate} is precise, namely $\tilde{\bvec{L}} = \bvec{L}$, and the number of sample scenarios $M$ is set as infinitely large, the exact algorithm shown in Algorithm~\ref{algo_main_nonmyopic} is RC-optimal in non-myopic settings under the given stochastic demand $\tilde{\mathcal{D}}$.
\label{thm_myopic_optimal}
\end{prop}

\begin{proof}
If the set of estimated trip-plans $\tilde{\bvec{L}}$ obtained by the function {\sc Enumerate} is precise, the function {\sc NonMyopicUpdate} does nothing. The function {\sc NonMyopicAllocation} with infinitely large number of samples maximizes the expected discounted social welfare as defined by Eq.~\ref{eq_def_DSW} precisely, by maximizing Q-value given by Eq.~\ref{eq_def_Q_value} using ZDD. The solution is obtained within the set of executable trip plans $\mathcal{Z}$ that  satisfies the space--time constraints as defined by Eq.~\ref{eq_dynamic_constraints} and capacity constraints as defined by Eq.~\ref{eq_dynamic_cap_constraints}. Thus, this mechanism is RC-optimal under the given stochastic demand $\tilde{\mathcal{D}}$.
\end{proof}

\begin{prop} 
All algorithms introduced in this section, specifically, myopic algorithm shown in Algorithm~\ref{algo_main} and non-myopic algorithm shown in Algorithm~\ref{algo_main_nonmyopic}, arbitrarily combined with Per agent and Branch cutting approximation is RC-feasible. The number of samples $M$ and the size of initially enumerated samples in non-myopic algorithm can also be set arbitrarily.
\label{thm_RC_feasible}
\end{prop}

\begin{proof}
In all algorithms, including the approximation ones, the solution is obtained within the set of executable trip plans $\mathcal{Z}$ that  satisfies the space--time constraints as defined by Eq.~\ref{eq_dynamic_constraints} and capacity constraints as defined by Eq.~\ref{eq_dynamic_cap_constraints}. Thus, all of these mechanisms are RC-feasible.
\end{proof}

\begin{table}[t]
  \begin{center}
    \caption{Summary of the proposed framework of the solution algorithms}
    \begin{tabular}{c|c|c|c|c|c} \hline
	Myopic/ & Reports at & Maximum & Computational & Efficiency & \\
	Non-Myopic & a time-step & options to keep & complexity & of allocation & \\ \hline
	& &1 & \multirow{2}{*}{small} & \multirow{2}{*}{low} & FCFS \\ 
	& Sequential&N & & & \\
	\multirow{2}{*}{Myopic}& &$\infty$ &\multirow{2}{*}{$\updownarrow$} &\multirow{2}{*}{$\updownarrow$} & \\ \cline{2-3}
	& &1 & & & \\
	& Simultaneous&N &\multirow{2}{*}{large} &\multirow{2}{*}{high} & \\
	& &$\infty$ & & &\multirow{2}{*}{$\updownarrow$}  \\ \cline{1-5}
	& &1 & \multirow{2}{*}{middle} & \multirow{2}{*}{middle} & \\ 
	& Sequential&N & & & \\
	\multirow{2}{*}{Non-Myopic}& &$\infty$ &\multirow{2}{*}{$\updownarrow$} &\multirow{2}{*}{$\updownarrow$} & \\ \cline{2-3}
	& &1 & & & \\
	& Simultaneous&N &\multirow{2}{*}{extra large} &\multirow{2}{*}{extra high} &\\
	& &$\infty$ & & &Optimal \\ \hline
    \end{tabular}
\label{tab_summary_algorithms}
  \end{center}
\end{table}

We summarize the proposed framework of the solution algorithms, including the approximation algorithms in Table~\ref{tab_summary_algorithms}. We provide myopic/non-myopic algorithms that process the information reported from agents at the same time-step sequentially or simultaneously and keep the pre-defined maximum options. The framework provides a wide range of trade-offs between the computational complexity and the efficiency of allocation. At the simplest, the myopic algorithm that sequentially processes reports at a time-step and keeps only one option coincides with the FCFS algorithm that requires a low computational complexity, but is not efficient. In contrast, a non-myopic algorithm that simultaneously processes reports at a time-step and does not limit the number of options, achieves the optimal allocation that maximizes the discounted social welfare given by Eq.~\ref{eq_def_DSW} but requires extremely high computational efforts.
As shown in Table~\ref{tab_summary_algorithms}, we provide wide-range algorithms between these two extreme cases. All of these algorithms are RC-feasible and can be interrelatedly implemented by a graph algorithm such as the ZDD.

%%%%%%%%%%%%%%%%%%%%%%%%%%%%%%%%%%%%%%%%%%%%%%%%%%%%%%%%%%%%%%%%%%%%%%%%%%%%%%%%%%%%%%%%%%%%%%%%%%%%%%%%%%%%%%%%
%
\section{\label{sec_numerical}Numerical Study}									
%
%
%%%%%%%%%%%%%%%%%%%%%%%%%%%%%%%%%%%%%%%%%%%%%%%%%%%%%%%%%%%%%%%%%%%%%%%%%%%%%%%%%%%%%%%%%%%%%%%%%%%%%%%%%%%%%%

Although we provide a wide-range of RC-feasible algorithms in the previous section, they have trade-offs between efficiency and computational time. In this section, we explore these trade-offs using numerical studies. In a numerical analysis, we consider a simple network representing the setting where the trip demands of tourists interferes with the background trip demand of commuters. 
In the study, we use Graphillion\footnote{\url{https://github.com/takemaru/graphillion}}, a Python software package on search, optimization, and enumeration using ZDD.

%%%%%%%%%
\subsection{Experimental setup}

\begin{figure}[t]
\begin{center}
	    \includegraphics[width = 0.5\hsize]{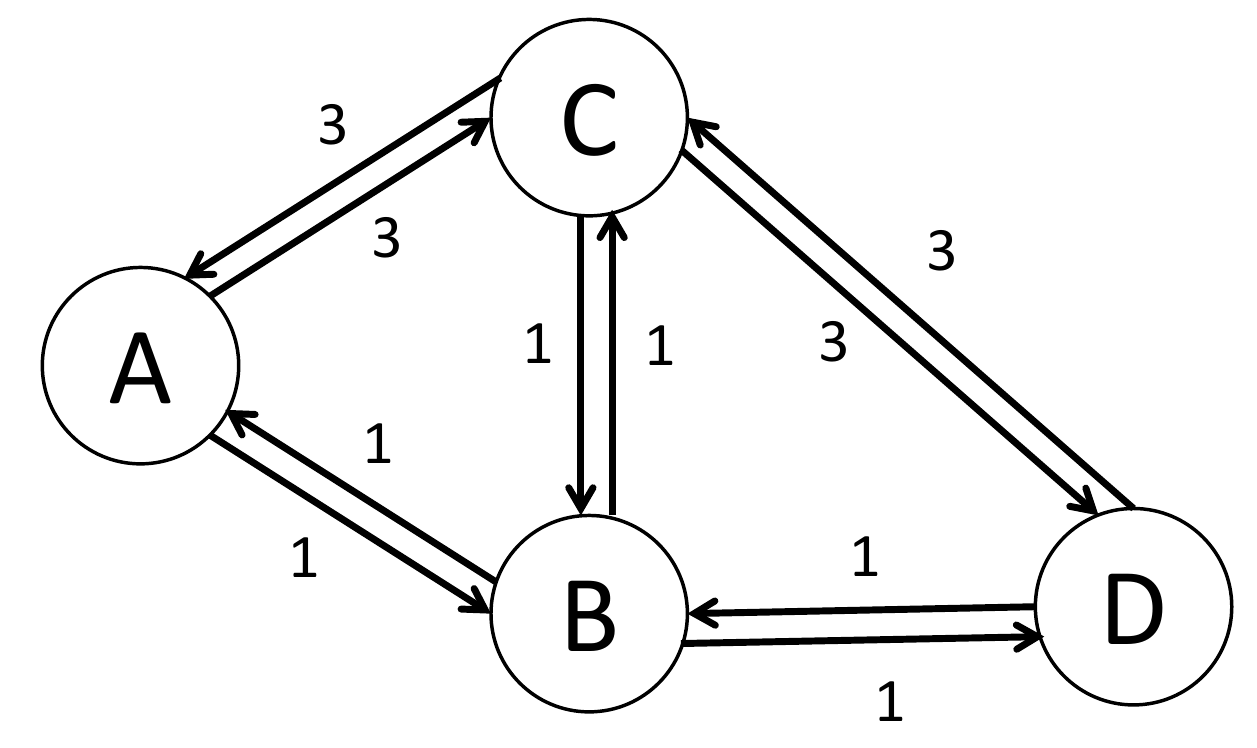}\\
          \caption{Capacity of traffic resources for each edge} 
	    \label{fig_SampleNetwork_capacity}
\end{center}
\end{figure}

%\begin{figure}[b]
%\begin{center}
%	\begin{minipage}{0.49\hsize}
%	\begin{center}
%	    \includegraphics[width = 0.8\hsize]{Figures/SampleNetwork_Cost}\\
%	    (a) Link cost
%	\end{center}
%	\end{minipage}
%	%2
%	\begin{minipage}{0.49\hsize}
%	\begin{center}
%	    \includegraphics[width = 0.8\hsize]{Figures/SampleNetwork_Capacity}\\
%	    (b) Link capacity
%	\end{center}
%	\end{minipage}
%          \caption{Sample network} 
%%	    \label{fig_SampleNetwork}
%\end{center}
%\end{figure}

Fist, we introduce a simple experimental setting to evaluate the basic performance of the proposed algorithm. In this setting, we set $\bar{T} = 8$. Specifically, we consider discrete time $T= \{1,2, \ldots, 8\}$ in this analysis. We also assume that the time discount rate is $\beta = 1$. We show the experimental setup for the \emph{static agent-type setting} in Table~\ref{tab_numerical_conditions}, which we state in detail in the following parts.

%%%%%%%%%
\subsubsection{Network}

We use the simple network that has already been shown in Fig.~\ref{fig_SampleNetwork}. The numbers in brackets in the figure show the required travel time on each edge. The capacity of traffic resources, that is, the possible number of agents using the edge simultaneously at one time step is shown in Fig.~\ref{fig_SampleNetwork_capacity}. We assume that the capacity is constant over time.
In this numerical study, we consider Nodes A and D to be residential and office areas, and Nodes B and C to be amusement areas. We set the active time for the facility on Node B as $b(B) = \{3,4,5\}$, while the facilities on other nodes are set to be active for the entire time duration in $T$.

\begin{table}[t]
\small
  \begin{center}
    \caption{Experimental setup in the \emph{static agent-type} setting}
    \begin{tabular}{lll} \hline
      Network & Set of nodes	& $\{A,B,C,D\} \in \mathcal{N}$\\
			&Set of edges& $\{AB, BA, AC, CA, BC, CB, BD, DB, CD, DC\} \in \mathcal{E}$\\
	&Travel time			& $\tau_{AB}=\tau_{BA}=\tau_{AC}=\tau_{CA}=\tau_{BC}=\tau_{CB}=\tau_{BD}=\tau_{DB}=1,$\\
						&&$\tau_{CD}=\tau_{DC}=2$, as shown in Fig.~\ref{fig_SampleNetwork}.\\
	&Traffic resource capacity &$\forall{t}: C_{AB,t}=C_{BA,t}=C_{BC,t}=C_{CB,t}=C_{BC,t}=C_{CB,t}=1,$\\		
						&&$\forall{t}: C_{AC,t}=C_{CA,t}=C_{CD,t}=C_{DC,t}=3$, as shown in Fig.~\ref{fig_SampleNetwork_capacity}.\\
	&Active time for nodes & $b(B) = \{3,4,5\}$\\
	\hline
	
	Agents&\\
     	\hspace{0.2cm} Passing agents $I_P^{AD}$ 	&Origin						&Node A\\
										&Destination					&Node D\\
										&Value of time and location		&$v_i^A = 0,\:\:  v_i^B = 0,\:\: v_i^C = 0,\:\:  v_i^D \sim N(250, 10000)$\\
										&Departure time				&$t_i^B$ is distributed uniformly in $\{0,1,2\}$\\
										&Deadline					&$t_i^E= t_i^B + 4$\\
     	\hspace{0.2cm} Passing agents $I_P^{DA}$ 	&Origin						&Node D\\
										&Destination					&Node A\\
										&Value of time and location		&$v_i^A \sim N(250, 10000),\:\:  v_i^B = 0,\:\: v_i^C = 0,\:\:  v_i^D = 0$\\
										&Departure time				&$t_i^B$ is distributed uniformly in $\{2,3,4\}$\\
										&Deadline					&$t_i^E= t_i^B + 4$\\
     	\hspace{0.2cm} Cruising agents $I_C$ 		&Origin						&Node A\\
										&Destination					&Node A\\
										&Value of time and location		&$v_i^A \sim N(50, 400),\:\:  v_i^B \sim N(200, 6400),\:\: v_i^C \sim N(100, 1600),\:\:  v_i^D \sim N(50, 400)$\\
										&Departure time				&$t_i^B$ is distributed uniformly in $\{0,1,2\}$\\
										&Deadline					&$t_i^E= t_i^B + 6$\\
										&Additional constraints			&Spend at least one time step at Node B during active time $b(B)$\\
	\hline
    \end{tabular}
\label{tab_numerical_conditions}
  \end{center}
\end{table}

%%%%%%%%%
\subsubsection{Passing and cruising agents}

We consider two types of agents, passing and cruising agents, denoted by $I_P$ and $I_C$, respectively, on the condition that $I_P \cup I_C = I$. The main objective of passing agents is to move between places. The origin and destination of passing agents are different. In contrast, the main objective of cruising agents is spending time at some places, enriching their experience by cruising around. The origin and destination of cruising agents are the same. A key motivating example for this setting is the traffic services in tourist sites\footnote{This setting can be expanded to demand response services at tourist sites, by introducing congestion pricing. We discuss this in Section~\ref{sec_discussions}. }. Passing agents express the background traffic demand that is represented by morning or evening commuters, and the cruising agents express tourists cruising multiple areas in the network. In the following section, we describe the agents in this analysis in detail.

\begin{itemize}
\item Passing Agents

%\paragraph{Passing Agents}

In our analysis, we introduce passing agents as background traffic. We consider two types of passing agents, namely, traveling from Node A to Node D and vice versa, representing morning and evening commuters. Sets of these agents are denoted by $I_P^{AD}$ and $I_P^{DA}$, respectively, on the condition that $I_P^{AD} \cup I_P^{DA} = I_P$. For each agent $i \in I_P^{AD}$ traveling from Node A to Node D with the value $\bvec{v_i} = \{ v_i^A, v_i^B, v_i^C, v_i^D,\}$, we set the values as follows:
\begin{equation}
	v_i^A = 0,\:\:  v_i^B = 0,\:\: v_i^C = 0,\:\:  v_i^D \sim N(250, 10000),
\label{eq_value_p1}
\end{equation}
where $N(\mu, \sigma^2)$ means the normal distribution with the mean of $\mu$ and the variance of $\sigma^2$. Similarly, for each agent $i \in I_P^{DA}$, we set the values as follows:
\begin{equation}
	v_i^A \sim N(250, 10000),\:\:  v_i^B = 0,\:\: v_i^C = 0,\:\:  v_i^D = 0.
\label{eq_value_p2}
\end{equation}
We assume that the departure time $t_i^B$ of each agent $i \in I_P^{AD}$ is distributed uniformly in $\{0,1,2\}$ whereas the departure time $t_i^B$ of each agent $i \in I_P^{DA}$ is distributed uniformly in $\{2,3,4\}$. The agents cannot stay at the origin, in that they have to start immediately after they request the trip.
We further assume that the deadline for the trip $t_i^E$ is set by $t_i^E= t_i^B + 4$ for all passing agents $i \in I_P$, meaning that passing agents have an upper-limit travel time of $4$. However, they have to spend the final time-step to stay at the destination.

\item Cruising Agents
%\paragraph{Cruising Agents}

Unlike passing agents, we introduce cruising agents as visitors that have a large amount of flexibility on their trips. The efficiency of traffic services can be improved by using this flexibility effectively.
We assume that the origin and destination of all cruising agents is Node A, representing the visitor staying at Node A and intending to cruise the area nearby. The main purpose of the cruising agents is visiting the facility on Node B. Thus, we set the constraints so that all the cruising agents must spend at least one time step at Node B during active time $b(B)$, except when the agents cannot obtain a permit and thus, cancel their trips. We set the value of these agents as follows:
\begin{equation}
	v_i^A \sim N(50, 400),\:\:  v_i^B \sim N(200, 6400),\:\: v_i^C \sim N(100, 1600),\:\:  v_i^D \sim N(50, 400).
\label{eq_value_c}
\end{equation}
We assume that the departure time $t_i^B$ of each agent $i \in I_C$ is distributed uniformly in $\{0,1,2\}$ and the deadline for trip $t_i^E$ is set by $t_i^E= t_i^B + 6$, meaning that the cruising agents have an upper-limit travel time of $6$.

\end{itemize}

In this setting, passing agents traveling from Node A to Node D interfere with the cruising agents when they depart from Node A; passing agents traveling from Node D to Node A interfere with cruising agents when they arrive at Node D. Thus, appropriate control of traffic resources is desired. 

%%%%%%%%%
\subsubsection{\label{sec_dynamic_type}Experimental setup for the dynamic agent-type setting}

As stated in Section~\ref{sec_flexible_behavior}, we basically consider the \emph{static agent-type setting}. However, in Section~\ref{sec_results_dynamic_type}, we explore our proposed mechanism in the \emph{dynamic agent-type setting}, where the reward and constraints of agents are changed by the situation.
%Given these networks and agents, we consider two settings, namely, static and dynamic types of agents. In the setting with the static type of agents, the reward and constraints of agents remain a constant. We call this as \emph{Static agent-type setting}.
%In contrast, the setting with the dynamic type of agents, the reward, and constraints of agents are changed by the situation.
In the numerical study in Section~\ref{sec_results_dynamic_type}, we specifically assume that cruising agents that visit Node B change their mind with a probability of $1/2$, such that,
${v'}_i^C = 3 \times v_i^C$
, and ${t'}_i^E= t_i^E + 2$.
Namely, the value ${v'}_i^C$ of staying at Node C becomes three times as high as the previous value $v_i^C$ and the deadline of the trip ${t'}_i^E$ extends by two time steps, when compared to the previous deadline $t_i^E$. This captures the setting where a part of the visitors who visit Node B are eager to visit Node C, additionally.
%We call this setting as \emph{Dynamic agent-type setting}.

%%%%%%%%%%%%%%%
\subsection{Benchmark Mechanism}

In our numerical analysis, we adopted the offline-optimal and FCFS mechanisms as benchmarks. The former is the optimal assignment stated in Definition~\ref{def_offline_optimal}, assuming that all future agents are given to the mechanism in advance. The social welfare achieved by the offline-optimal mechanism is considered the upper-limit of the given situation. FCFS is commonly introduced in settings that are considered in this study. It is not efficient in that traffic resources that were once allocated to users cannot be reallocated. 

We evaluate a wide-range of RC-feasible algorithms over these benchmarks. In myopic settings, our proposed exact algorithm shown in Algorithm~\ref{algo_main}, which maximizes the discounted social welfare is called as \emph{Proposed (myopic)} in this section. We also evaluate the approximate algorithms. Our proposed algorithm with \emph{Per agent} approximation as shown in Algorithm~\ref{algo_allocation_approximate} is called \emph{Per agent (myopic)}. The algorithm with \emph{Branch cutting} approximation is stated by the maximum number of options, like \emph{Per agent (myopic) with maximum branch of $N_{branch}^{max}$.}
Similarly, in non-myopic settings, our proposed algorithm shown in Algorithm~\ref{algo_main_nonmyopic} is called as \emph{Proposed (non-myopic)}. The approximation algorithms \emph{Per agent} and \emph{Branch cutting} are also considered and called \emph{Per agent (non-myopic) with maximum branch of $N_{branch}^{max}$.} The number of sample scenarios used in these algorithms is set as $M=10$.

We assume that the offline-optimal mechanism has perfect knowledge of the space--time constraints and reward functions of future agents, while our proposed non-myopic algorithms knows the stochastic demand $\tilde{\mathcal{D}}$ that includes the deterministic space--time constraints and stochastic model of the reward function as expressed by Eqs.~\ref{eq_value_p1}, \ref{eq_value_p2}, and \ref{eq_value_c}. In contrast, FCFS and our proposed myopic algorithms have no knowledge of future agents at all.
For \emph{dynamic agent-type setting}, we assume that the offline-optimal mechanism additionally has perfect knowledge of the changing mind of agents, specifically in terms of the individual's identity and how that agent changes in the future. In contrast, the FCFS mechanism and our proposed myopic and non-myopic algorithms have no knowledge of the changing minds of agents.

%%%%%%%%%%%%%%%
\subsection{Results}

In this part, we present the results of the numerical analysis and discuss the advantage of our proposed mechanism over the benchmark mechanism.

%%%%%%%
\subsubsection{Efficiency}
%
%\begin{figure}[t]
%\begin{center}
%	    \includegraphics[width = 0.9\hsize]{Figures/Efficiency}\\
%
%          \caption{Efficiency of proposed and FCFS mechanisms \textcolor{red}{(To be fixed)}} 
%	    \label{fig_Efficiency}
%\end{center}
%\end{figure}

\begin{figure}[t]
\begin{center}
%Source: 0523_2337
	%1
	\begin{minipage}{0.49\hsize}
	\begin{center}
	    \includegraphics[width = \hsize]{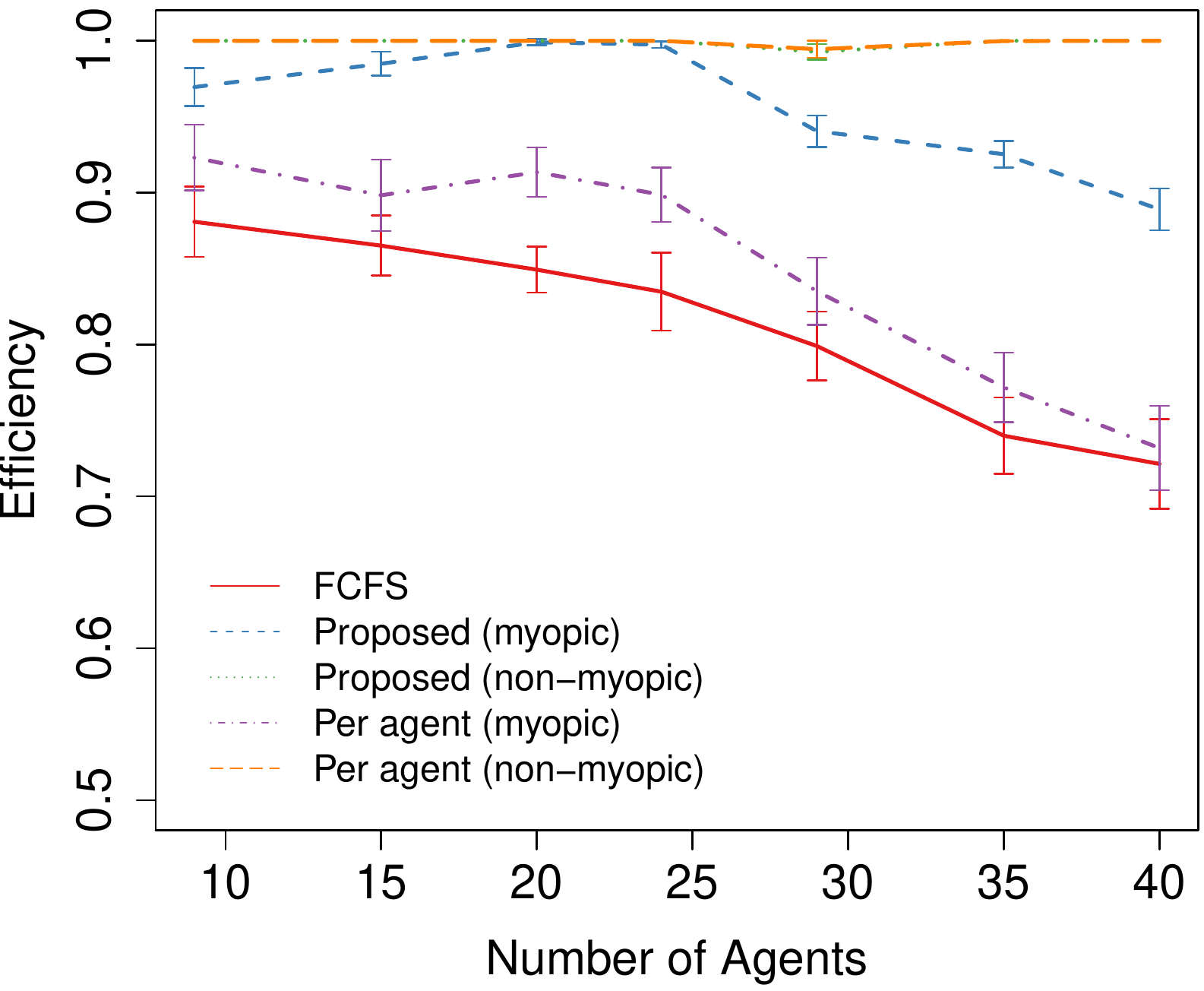}\\

	(a) Fraction of cruising agents = 0.1
	\end{center}
	\end{minipage}
	%2
	\begin{minipage}{0.49\hsize}
	\begin{center}
	    \includegraphics[width = \hsize]{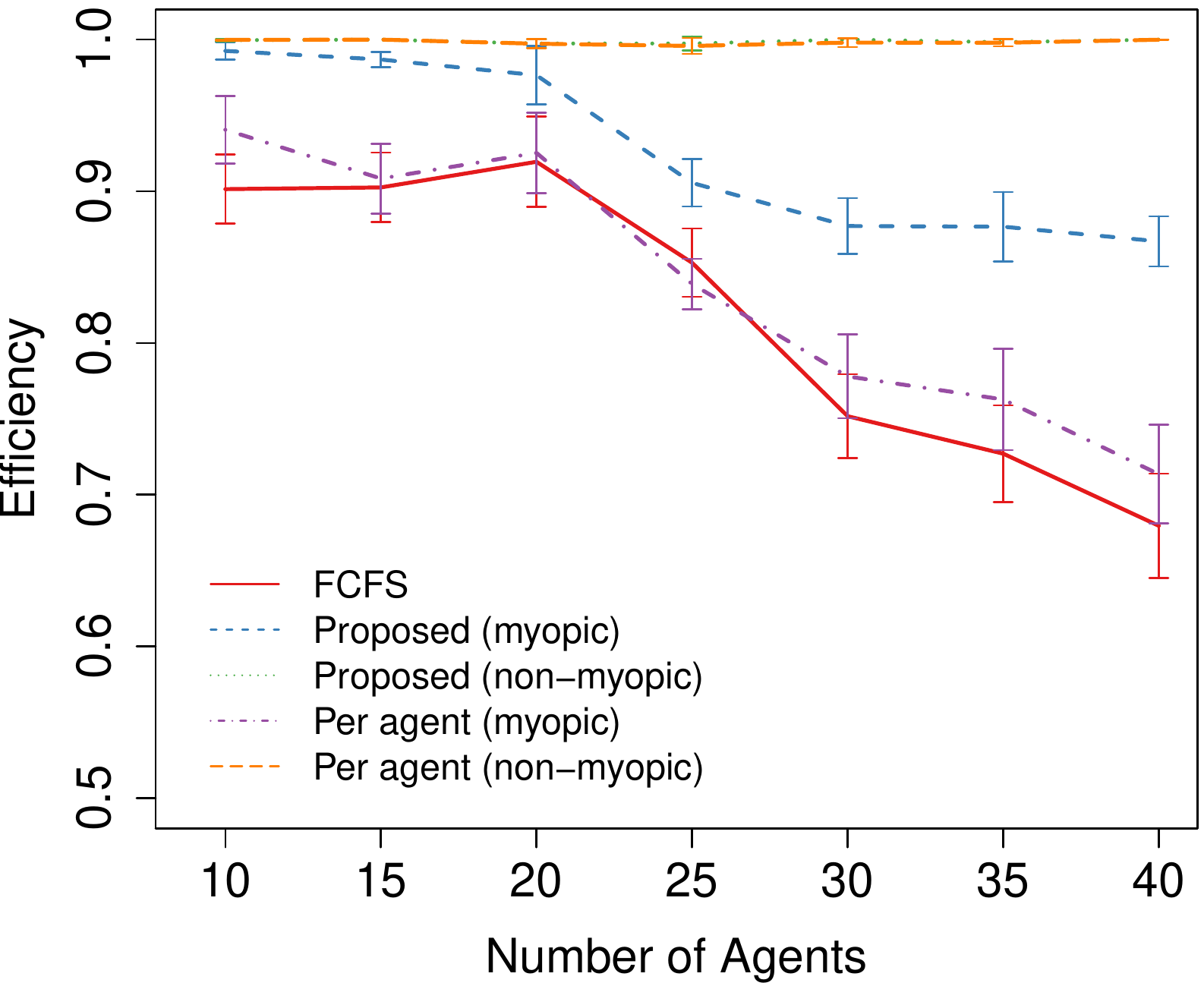}\\

	(b) Fraction of cruising agents = 0.2
	\end{center}
	\end{minipage}

          \caption{Efficiency in \emph{static agent-type} setting} 
	    \label{fig_Efficiency}
\end{center}
\end{figure}
%
%\begin{figure}[t]
%\begin{center}
%%\footnotesize
%	%1
%	\begin{minipage}{0.49\hsize}
%	\begin{center}
%	    \includegraphics[width = \hsize]{Figures/Efficiency_0_1}\\
%
%	(a) Fraction of cruising agents = 0.1
%	\end{center}
%	\end{minipage}
%	%2
%	\begin{minipage}{0.49\hsize}
%	\begin{center}
%	    \includegraphics[width = \hsize]{Figures/Efficiency_0_2}\\
%
%	(b) Fraction of cruising agents = 0.2
%	\end{center}
%	\end{minipage}
%
%          \caption{Efficiency of the proposed and FCFS mechanisms (\emph{Static agent-type setting}) \textcolor{red}{(To be fixed)}} 
%	    \label{fig_Efficiency}
%\end{center}
%\end{figure}

First, we evaluate the \emph{efficiency} of our proposed algorithms in \emph{static agent-type setting}. Here, we define \emph{efficiency} as the total social welfare achieved by each mechanism as a proportion of the offline-optimal welfare. We consider two cases, with the ratio of $10\%$ or $20\%$ of cruising agents over all agents, and vary the number of agents from $10$ to $40$. The experiment was repeated 10 trials for each setting. The results of \emph{Proposed} and \emph{Per agent} algorithms in myopic and non-myopic settings are shown in Fig.~\ref{fig_Efficiency}. 
The results in settings where the fraction of cruising agents is 0.1 and 0.2 are shown in Fig.~\ref{fig_Efficiency}(a) and Fig.~\ref{fig_Efficiency}(b), respectively. The efficiency of the FCFS mechanism is also shown in the same figure. The plots in these figures show the mean value of the trials and the error bars express 95\% confidence. 

As seen in Fig.~\ref{fig_Efficiency}, the efficiency of \emph{FCFS} decreases with the interference of agents. On the other hand, \emph{Proposed (myopic)} achieves better efficiency than \emph{FCFS} does, and \emph{Proposed (non-myopic)} retains almost the same quality as the optimal allocation does.
While the proposed mechanism guarantees the minimum service quality to the early-coming agents with space--time prism constraints, it takes later-coming high-valued agents into consideration in its decision-making and thus, achieves high efficiency.

As for the \emph{Per agent} approximation, \emph{Per agent (myopic)} results in losses in the efficiency that is almost near the level of \emph{FCFS}, while \emph{Per agent (non-myopic)} does not result in losses in efficiency. 
Although the decision-making process in \emph{Per agent} algorithm shown in Algorithm~\ref{algo_allocation_approximate} is divided for each agent, the information of all agents reporting at the same time-step is taken into consideration in the non-myopic algorithm and thus, the efficiency loss is suppressed.

%%%%%%%
%The number of executable plans

\begin{figure}[t]
\begin{center}
	    \includegraphics[width = 0.5\hsize]{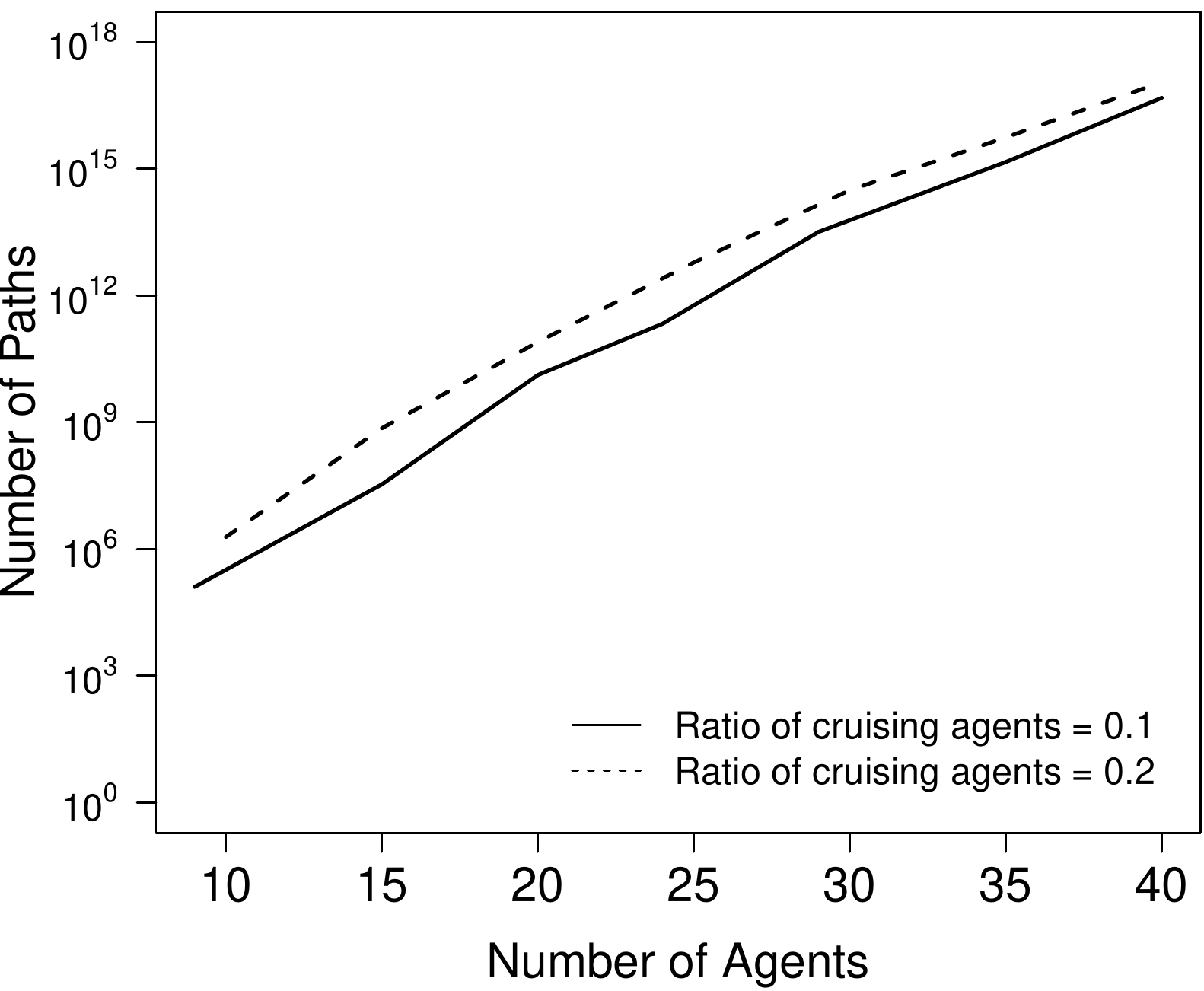}\\

          \caption{The total number of executable plans for each setting} 
	    \label{fig_num_paths}
\end{center}
\end{figure}

In Fig.~\ref{fig_num_paths}, we show the number of enumerated executable plans that keep space--time constraints and capacity constraints considered in the numerical analysis stated above. As the number of agents increases, the number of executable plans also increases extensively, owing to the combinatorial nature of the problem. It reaches $1.13 \times 10^{17}$ when the number of agents is $40$ and the fraction of cruising agents is $0.2$.
\emph{Proposed (non-myopic)} offers a full-search of all these plans for making decisions and thus achieves the highest efficiency. On the other hand, \emph{Proposed (myopic)} makes decisions myopically, but keeps all executable plans in the future, under the given myopic states with all agents reported until that time. This results in higher efficiency when compared to the FCFS.

%%%%%%%
%Branch Cutting

\begin{figure}[t]
\begin{center}
	\begin{minipage}{0.49\hsize}
	\begin{center}
	    \includegraphics[width = \hsize]{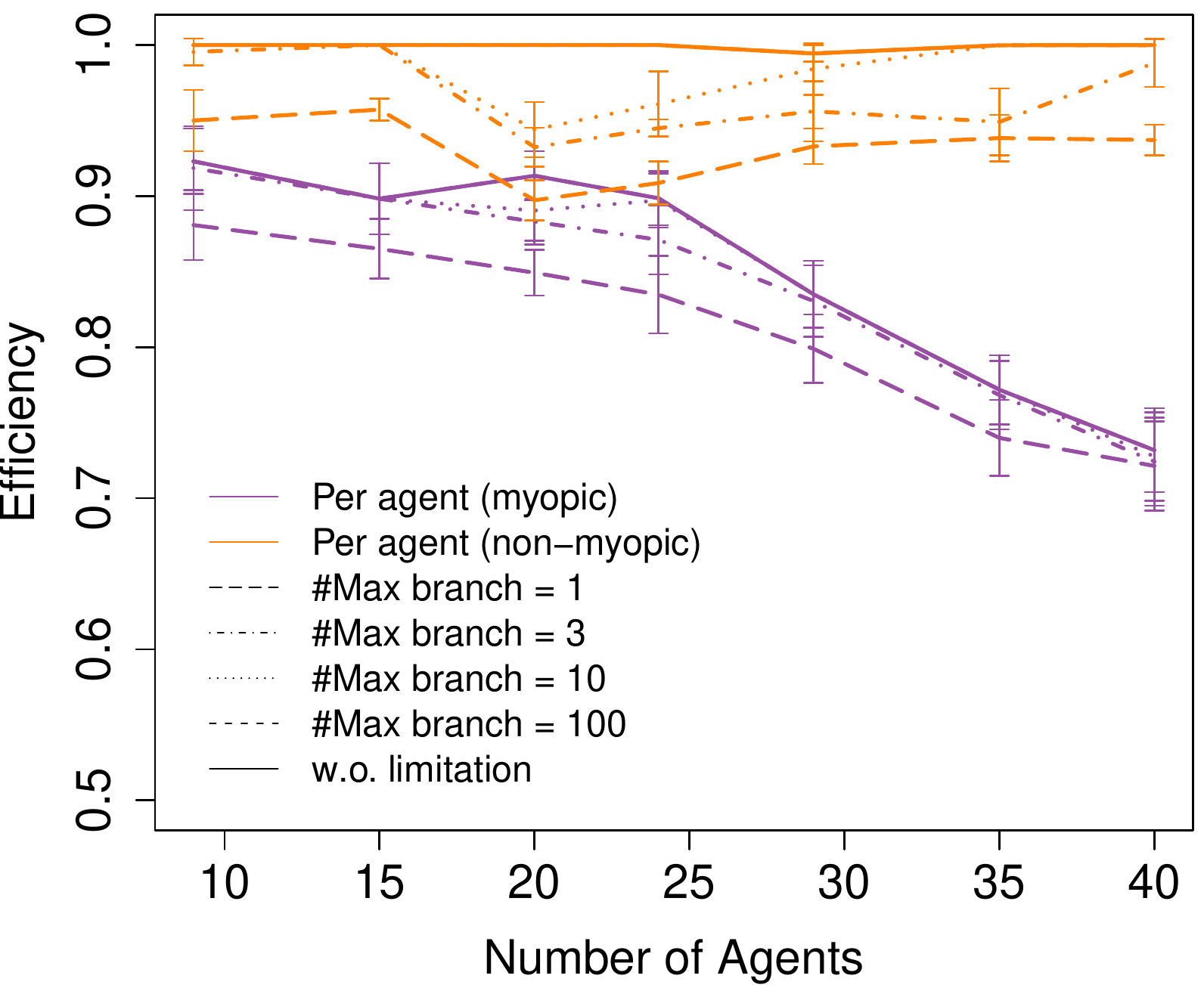}\\

	(a) Fraction of cruising agents = 0.1
	\end{center}
	\end{minipage}
	%2
	\begin{minipage}{0.49\hsize}
	\begin{center}
	    \includegraphics[width = \hsize]{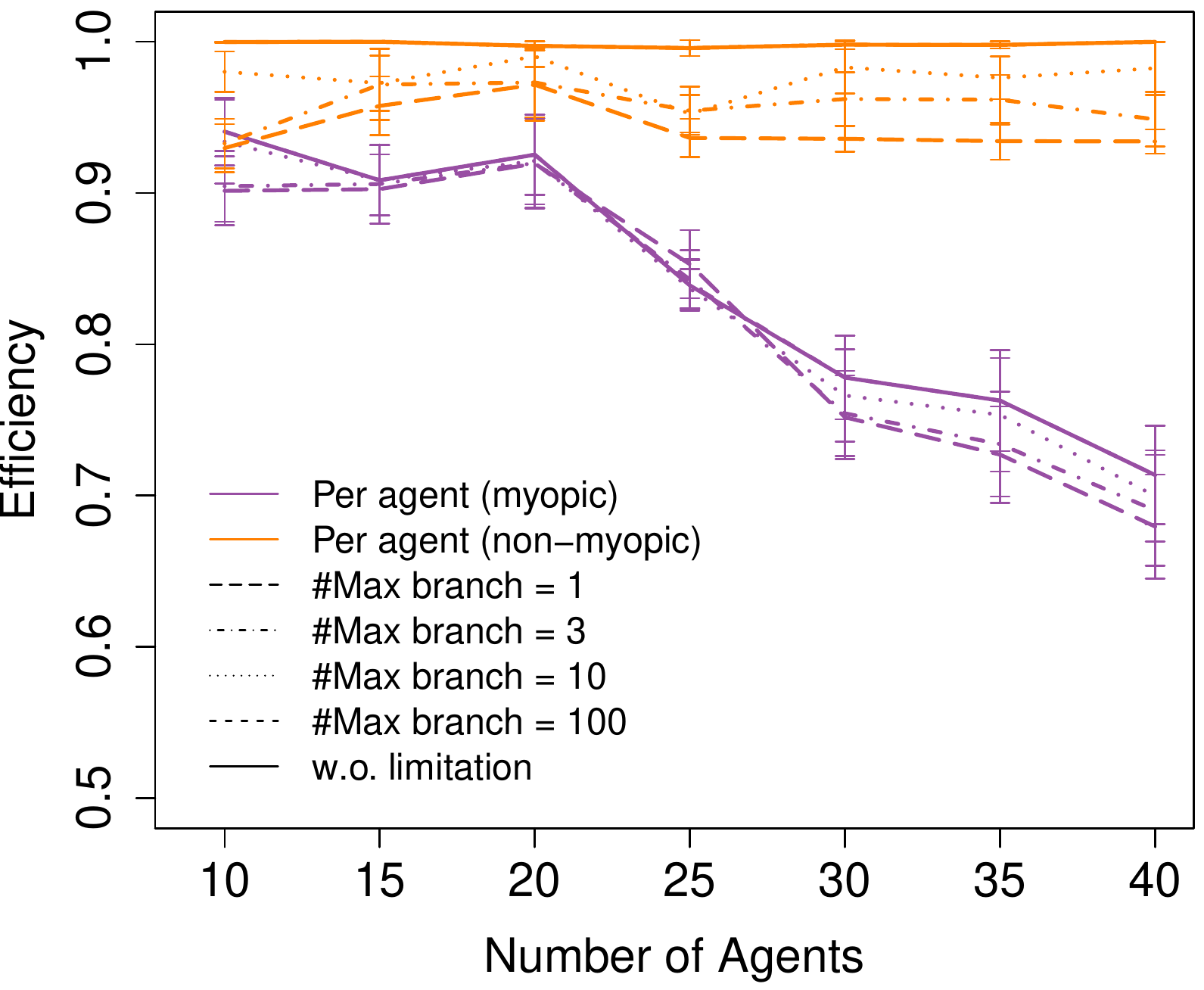}\\

	(b) Fraction of cruising agents = 0.2
	\end{center}
	\end{minipage}
\end{center}

          \caption{Efficiency of \emph{Branch cutting} algorithms} 
	    \label{fig_efficiency_Branch cutting}
\end{figure}

Then, we evaluate the performance of \emph{Branch cutting} approximation as introduced into \emph{Per agent (myopic)} and \emph{Per agent (non-myopic)} algorithms, running a numerical experiment in the same condition as in Fig.~\ref{fig_Efficiency}.
For each algorithm, we vary the maximum branch $N_{branch}^{max}$ from 1 to 100, and evaluate the performance. 
In this analysis, we choose the highest $N_{branch}^{max}$ branches in terms of the expected discounted social welfare at each step of the mechanism.
The results are shown in Fig.~\ref{fig_efficiency_Branch cutting}. Note that, in this figure, the \emph{Per agent (myopic)} algorithm with maximum branch $N_{branch}^{max} = 1$ coincides with the FCFS mechanism.
As we can see from this figure, the efficiency decreases as the maximum branch decreases. %
In these settings, the performance of the \emph{Branch cutting} algorithm with maximum branch $N_{branch}^{max} = 10$ is fairly close to the performance, without the limitation of the maximum branch in myopic settings, but it still has a gap in non-myopic settings. 
In the setting where the fraction of cruising agents is 0.1 (Fig.~\ref{fig_efficiency_Branch cutting}(a)), the results with maximum branch $N_{branch}^{max} = 3$ outperform the results of FCFS in myopic settings. It implies that our proposed algorithms may achieve considerably higher efficiency than FCFS, by keeping just a few options, depending on the situation.

%%%%%%%
\subsubsection{Rejection rate}

%
%\begin{figure}[t]
%\begin{center}
%	    \includegraphics[width = 0.6\hsize]{Figures/RejectRate}\\
%
%          \caption{Rejected rate\textcolor{red}{(To be fixed)}} 
%	    \label{fig_rejected_rate}
%
%\end{center}
%\end{figure}

\begin{figure}[t]
\begin{center}
% 生データ：0523_2337
	%1
	\begin{minipage}{0.49\hsize}
	\begin{center}
	    \includegraphics[width = \hsize]{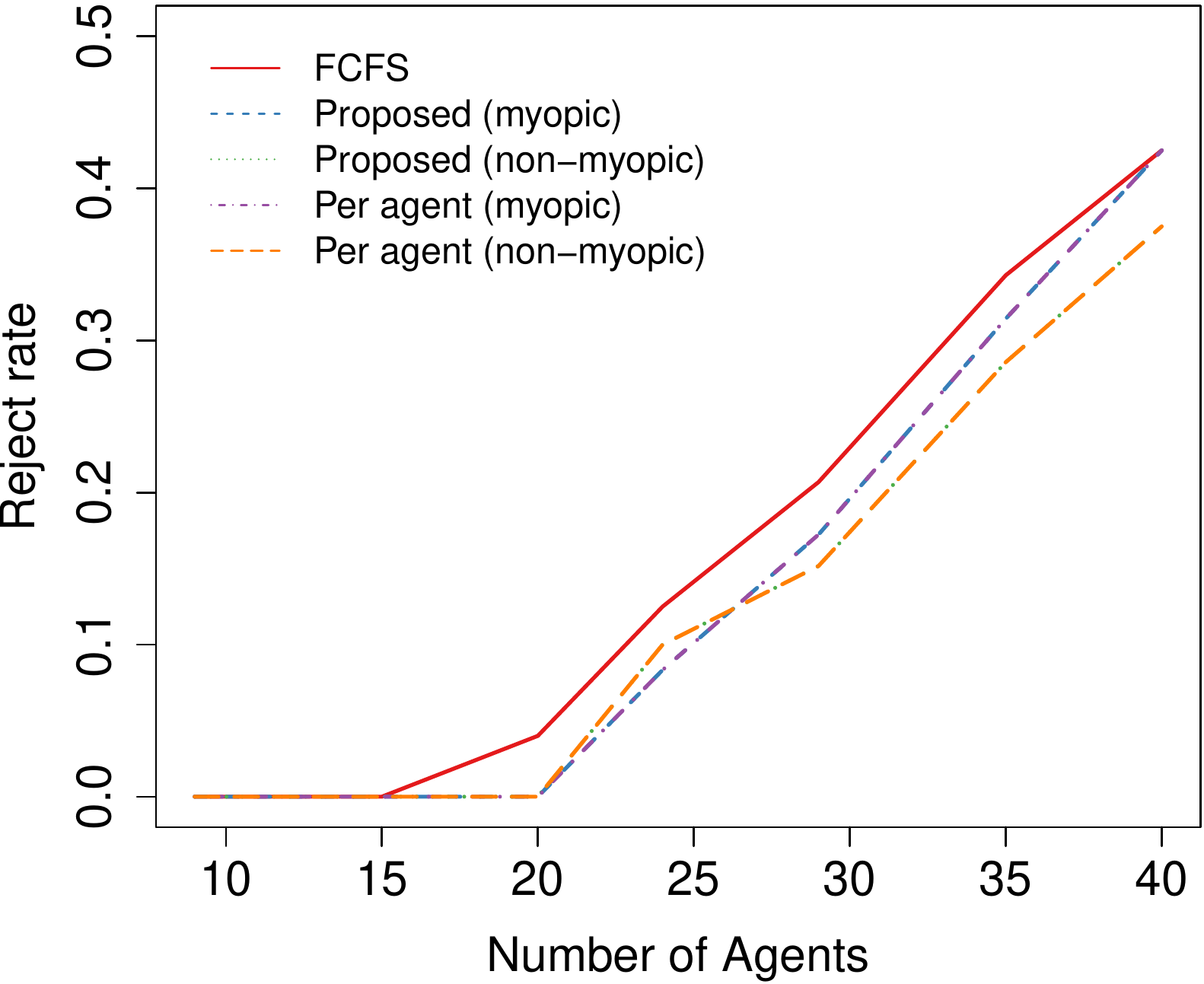}\\

	(a) Fraction of cruising agents = 0.1
	\end{center}
	\end{minipage}
	%2
	\begin{minipage}{0.49\hsize}
	\begin{center}
	    \includegraphics[width = \hsize]{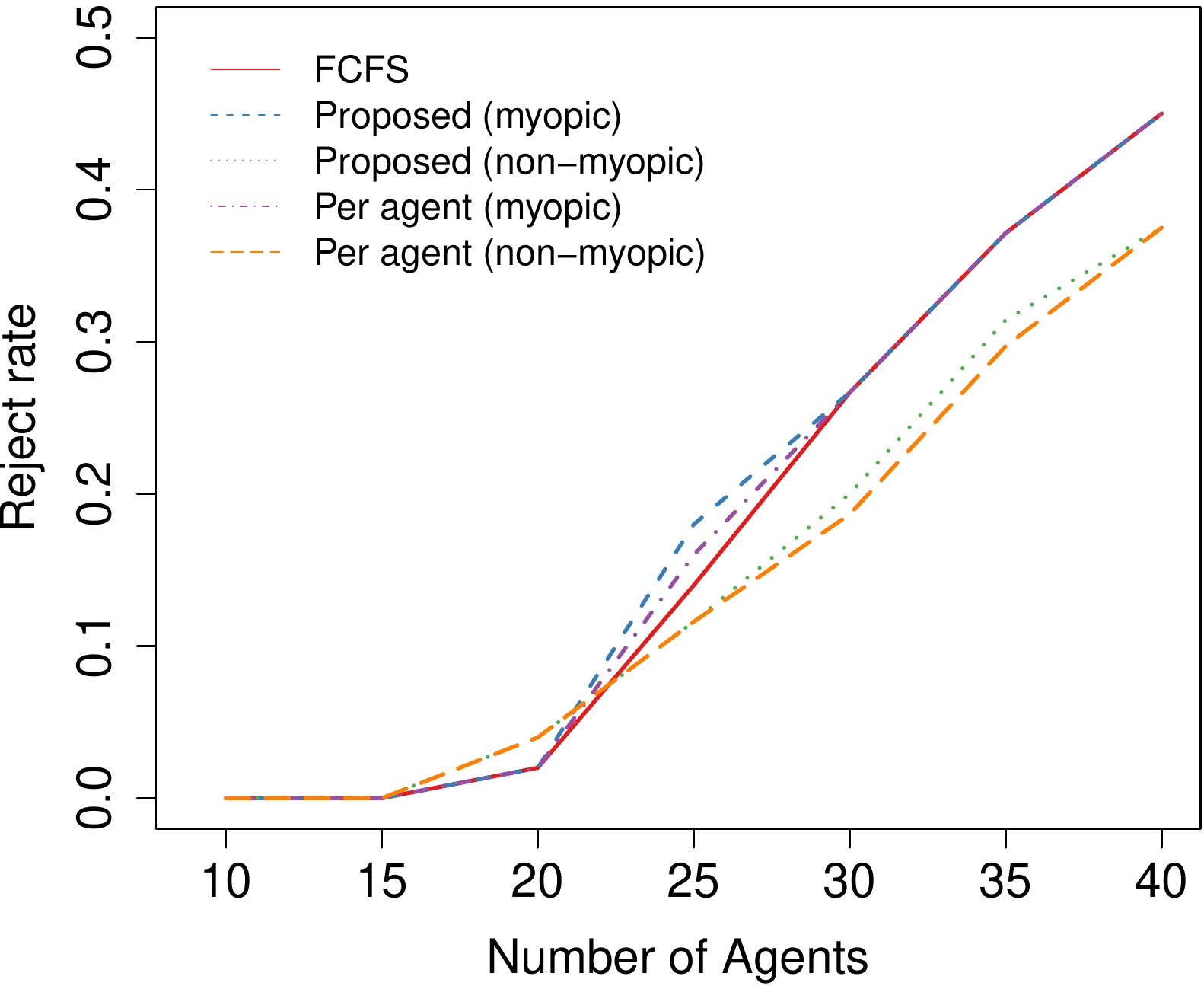}\\

	(b) Fraction of cruising agents = 0.2
	\end{center}
	\end{minipage}

          \caption{Rejection rate} 
	    \label{fig_rejected_rate}
\end{center}
\end{figure}

In Fig.~\ref{fig_rejected_rate}, we show the rejection rate as the fraction of the number of rejected agents over the number of all agents, by each algorithm.
In Fig.~\ref{fig_rejected_rate}(a), where the fraction of cruising agents is $0.1$, the rejection rates of our proposed algorithms are lower than the FCFS. In contrast, in Fig.~\ref{fig_rejected_rate}(b), where the fraction of cruising agents is $0.2$, the rejection rates of our proposed non-myopic algorithms are lower, but those of our proposed myopic algorithms are higher than FCFS. 
%It can be considered that the non-myopic algorithm is effective when it comes to accepting more customers in settings where the interference of customers is not as large, but is not in the setting with large interference. 
In both cases, the rejection rate of our proposed non-myopic algorithms are higher than myopic algorithms when the number of agents is small, in the range from 20 to 25 in Fig.~\ref{fig_rejected_rate}(a) or from 15 to 20 in Fig.~\ref{fig_rejected_rate}(b). We can see from this, that the non-myopic mechanism makes decisions to reject early-coming agents to keep traffic resources for later-coming high-value agents in order to achieve high efficiency.

As we can seen from Fig.~\ref{fig_rejected_rate}, in our experimental settings, the rejection rate is considerably high and exceeds 0.4 in the largest settings. It seems to be too high if the discussion is based on an emerged demand, for example, the actual records of taxi services in the real world. However, in this paper, we aim to consider hidden demands that have not emerged because users give up on traveling before they request the service, making judgments that the trip would possibly suffer their space--time constraints. Considering that such potential users may become customers if the mobility services are improved, the rejection rates in Fig.~\ref{fig_rejected_rate} are rational. In such settings with high rejection rates, it is difficult to solve the combined problems of finding a combination of feasible agents and allocating traffic resources for those agents using ILP. Graph-based algorithms, including our proposed algorithms using the ZDD, are appropriate in these settings.

%このようなリジェクト率が高い条件の下で、Feasible なエージェントの組み合わせと最適リソース配分を同時に解く問題は、ILPでは難しい。全列挙を使う提案手法の長所が出ている。

%%%%%%%
\subsubsection{\label{sec_calc_time}Calculation time}

\begin{figure}[t]
\begin{center}
% 生データ：0523_2337
	%1
	\begin{minipage}{0.49\hsize}
	\begin{center}
	    \includegraphics[width = \hsize]{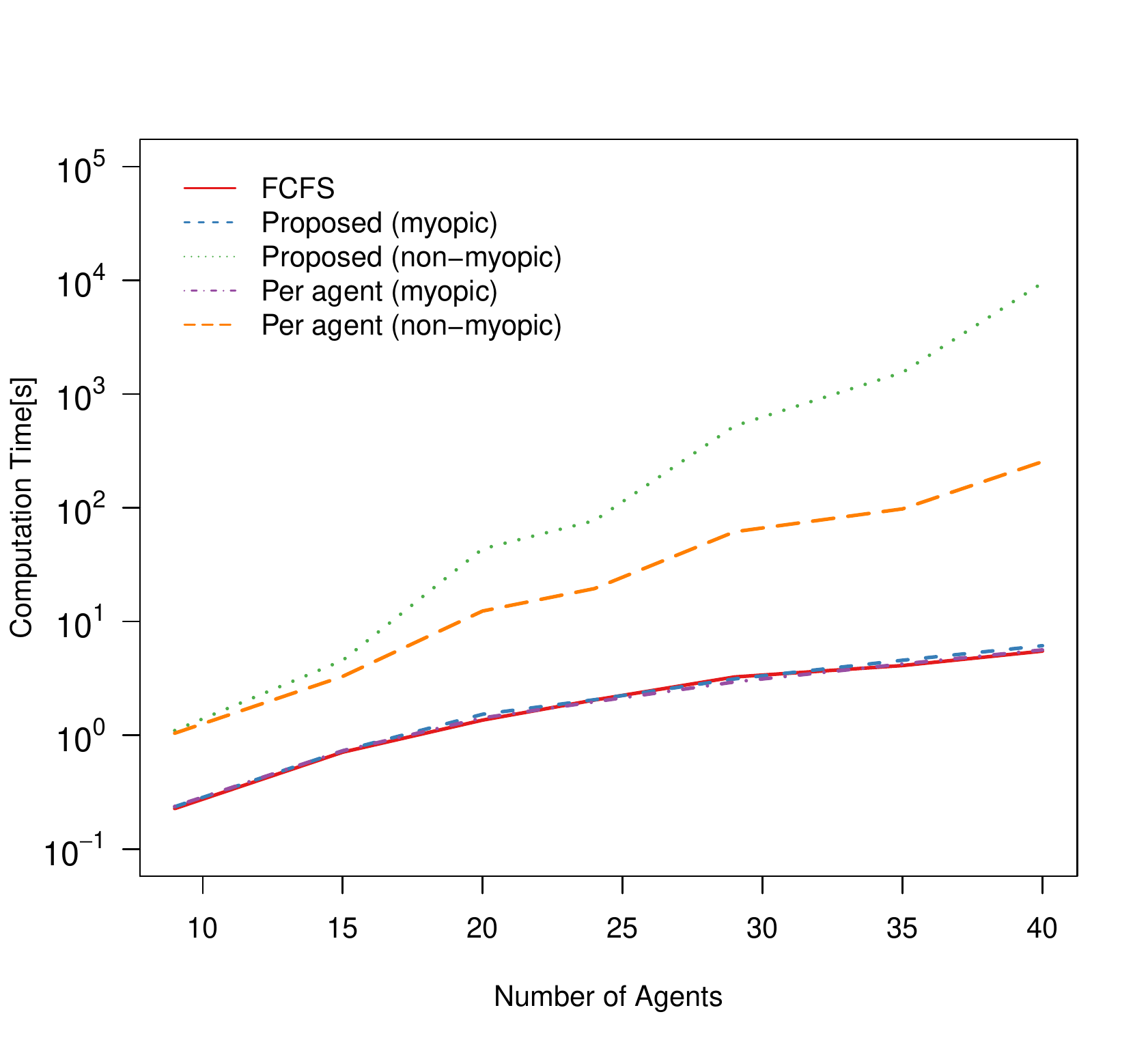}\\

	(a) Fraction of cruising agents = 0.1
	\end{center}
	\end{minipage}
	%2
	\begin{minipage}{0.49\hsize}
	\begin{center}
	    \includegraphics[width = \hsize]{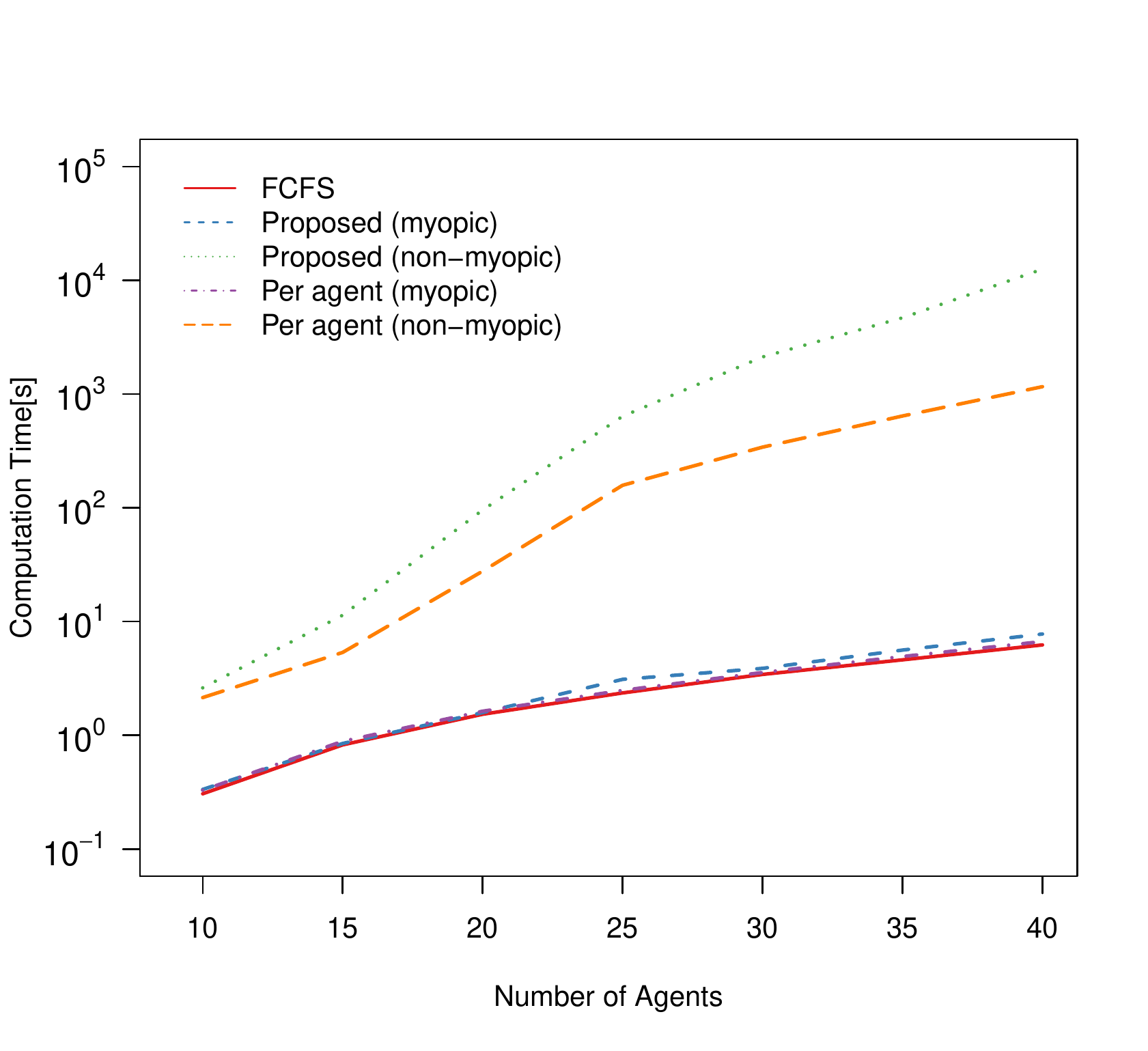}\\

	(b) Fraction of cruising agents = 0.2
	\end{center}
	\end{minipage}

          \caption{Calculation Time} 
	    \label{fig_CalcTime}
\end{center}
\end{figure}

We show the calculation time of the proposed mechanism in Fig.~\ref{fig_CalcTime}. All the analysis was performed on a machine with an Intel(R)Xeon(R)E5-2690v4 CPU@2.60GHz, using a single core, and 56GB of RAM.
As we can see from these figures, the calculation time of the proposed non-myopic algorithm becomes large as the number of agents increases, while that of the proposed myopic algorithm stays small. However, considering that the non-myopic algorithm treats more than a trillion options (as shown in Fig.~\ref{fig_num_paths}) of the future combinatorial behavior of agents and repeats the optimization process for each action plan and each sample scenario shown in Algorithm~\ref{algo_allocation}, the calculation time is still small, owing to the nature of the ZDD.

%, owing to the nature of the ZDD and thus may be acceptable depending on each case.
As we have shown, our proposed non-myopic algorithm takes a large amount of computation time, but has high performance, while the basic FCFS mechanism can compute faster, but results in poor performance. Within the trade-off between the performance and computation time, our proposed framework offers a wide range of options.
We can achieve better performance by adopting a myopic mechanism with an appropriate maximum branch. Depending on the situation, such as for example, a car-sharing service for limited pre-described customers in which we have enough time to compute, we can introduce a non-myopic mechanism to achieve better performance.

%%%%%%%%%
\subsubsection{\label{sec_results_dynamic_type} Efficiency in the dynamic agent-type setting}

\begin{figure}[t]
\begin{center}
%Source:0523_1549
	%1
	\begin{minipage}{0.49\hsize}
	\begin{center}
	    \includegraphics[width = \hsize]{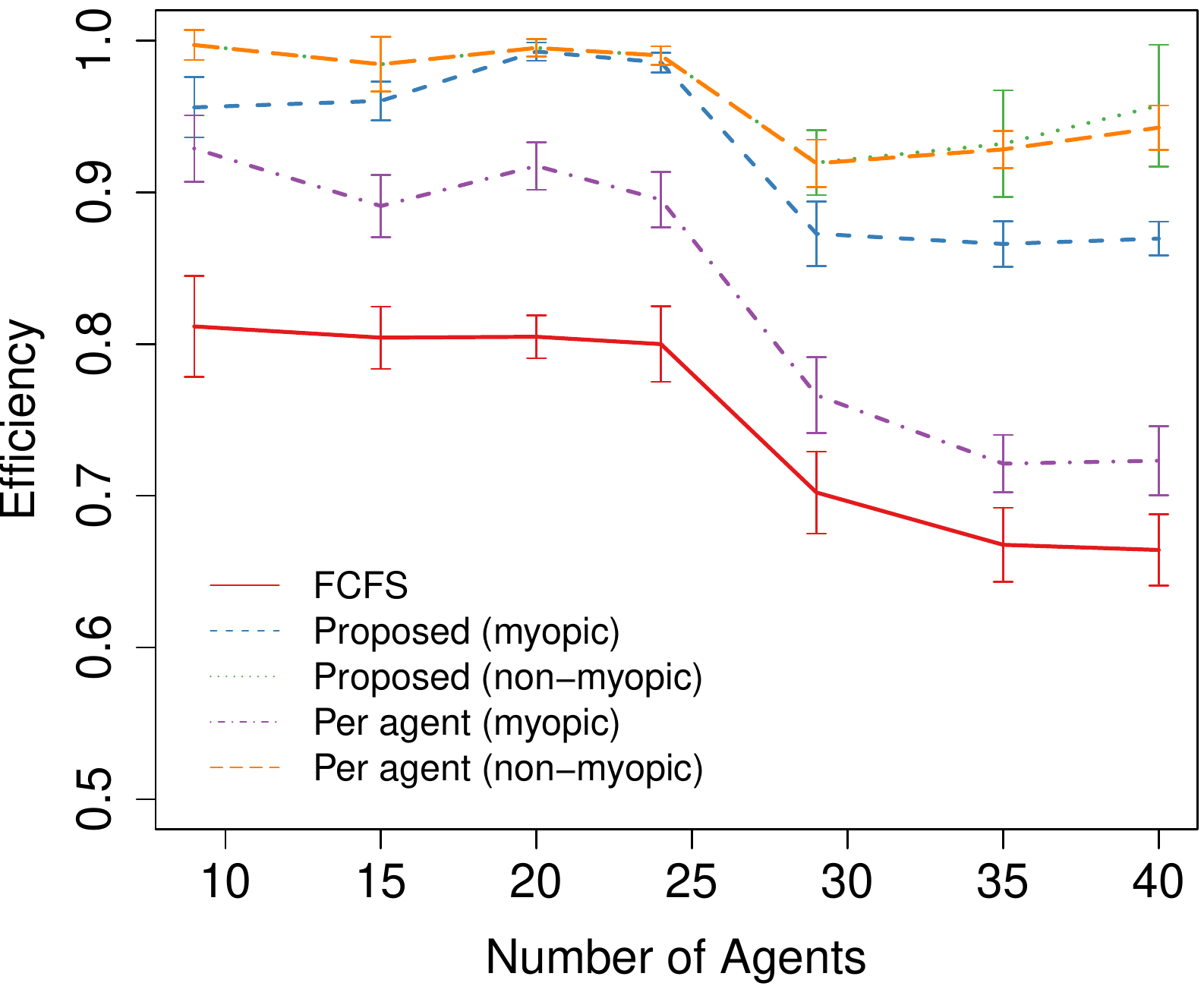}\\

	(a) Fraction of cruising agents = 0.1
	\end{center}
	\end{minipage}
	%2
	\begin{minipage}{0.49\hsize}
	\begin{center}
	    \includegraphics[width = \hsize]{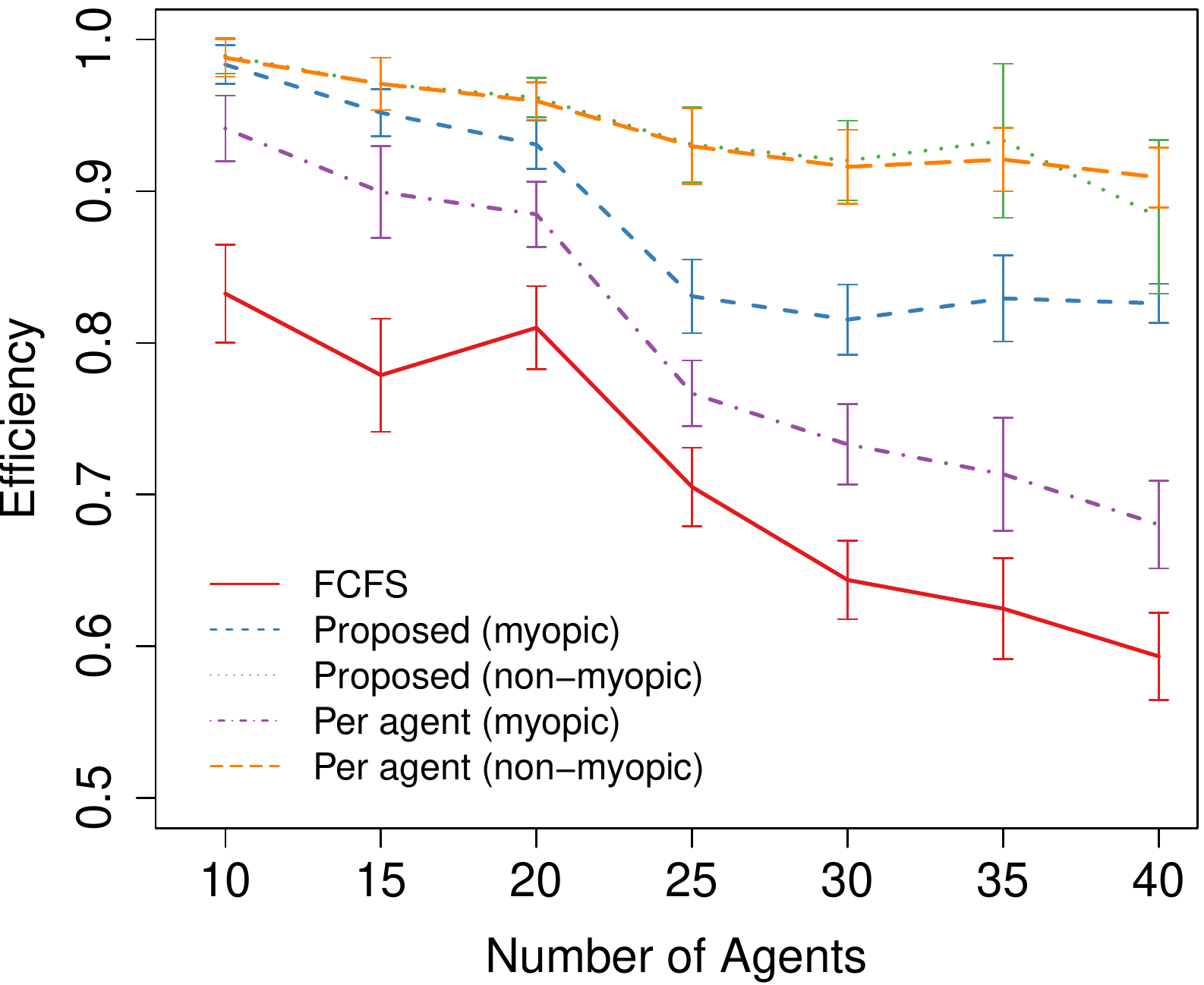}\\

	(b) Fraction of cruising agents = 0.2
	\end{center}
	\end{minipage}

          \caption{Efficiency in \emph{dynamic agent-type setting} } 
	    \label{fig_dynamic}
\end{center}
\end{figure}

Then, we evaluate the \emph{efficiency} in the \emph{dynamic agent-type setting} as stated in Section~\ref{sec_dynamic_type}. The performance of \emph{Proposed}, \emph{Per agent} and FCFS algorithms in both, myopic and non-myopic settings, are shown in Fig.~\ref{fig_dynamic}.
The overall trends in the results are similar to the results in \emph{Static agent-type setting} as shown in Fig.~\ref{fig_Efficiency}, although the non-linearity of the performance increases, because of the settings that include the non-linear mind-change on part of the users.
It is noteworthy that the gap between the performance of \emph{Per agent (myopic)} and that of FCFS becomes larger in a \emph{dynamic agent-type setting} than in a \emph{static agent-type setting}. As we can see from this result, considering the changing minds of agents, the proposed mechanism that sequentially makes a decision at each time, performs well. 
In a non-myopic setting, the performance of \emph{Per agent} algorithm is almost as good as that of \emph {Proposed} algorithm.
Thus, it can be said that a \emph{Per agent} approximation is effective in both, myopic and non-myopic mechanisms, used in dynamic situations.

%\color{blue}
%
%問題設定： 地点B を訪れたエージェントのうち半数は、ノードCのvalueが3倍になり、時間制限が+2になる。
%Offline optimal：　心変わりも含めて既知
%Model-based, Model-free, FCFS は、全て、心変わりは想定していないものとする。
%
%エージェント数を(6,6)で固定して、割引率を振って計算する？
%
%
%\textcolor{red}{要・追加実験}

%%%%%%%%%%%%%%%
\subsubsection{Evaluation on the large number of agents}

\begin{figure}[t]
\begin{center}
	    \includegraphics[width = 0.5\hsize]{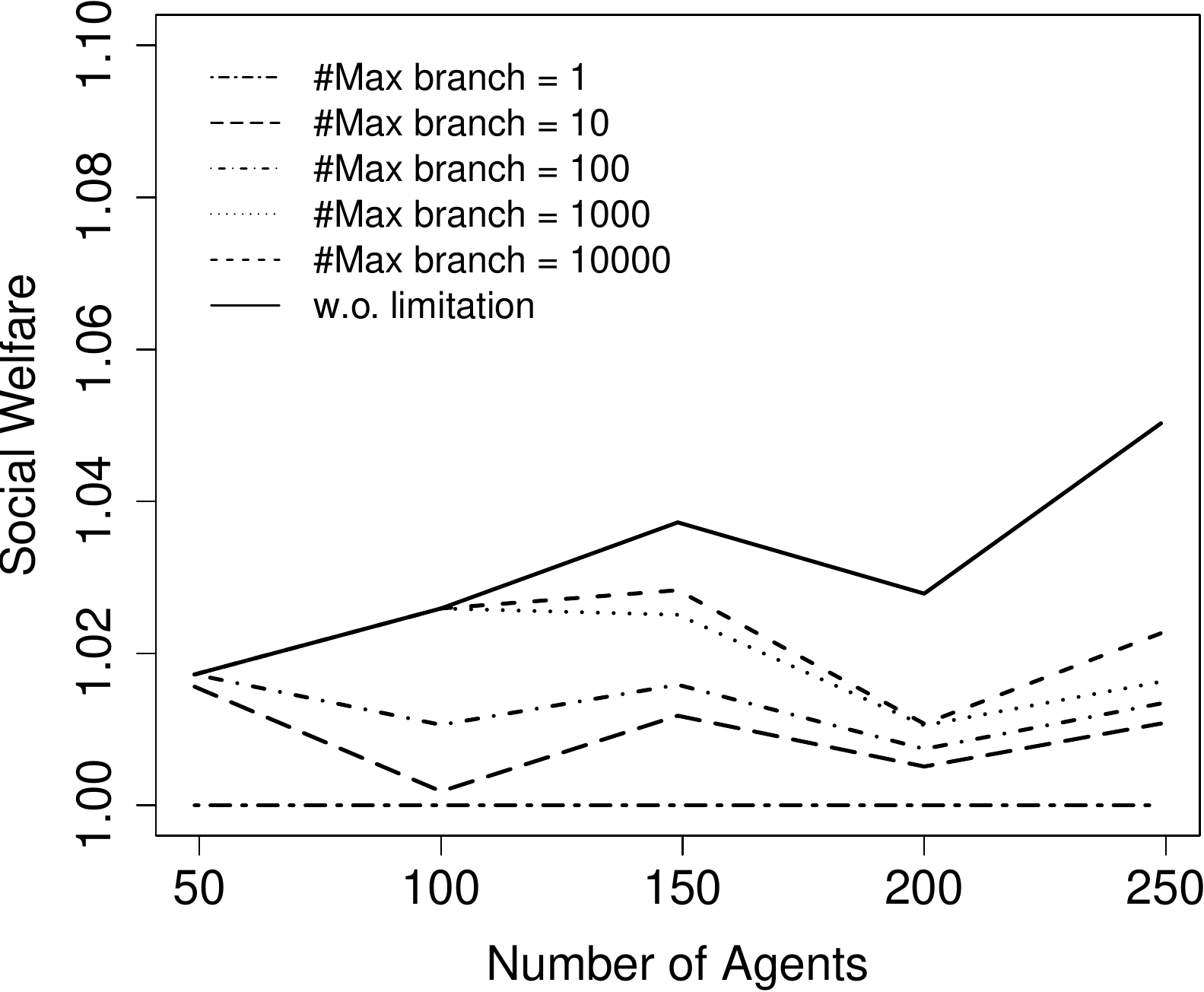}\\

          \caption{Achieved social welfare (vs. FCFS)} 
	    \label{fig_results_LargeScale}
\end{center}
\end{figure}

%
%\begin{figure}[t]
%\begin{center}
%%	    \includegraphics[width = 0.5\hsize]{figures/}\\
%
%		\textcolor{red}{To be set.}
%
%          \caption{Calculation time of large scale simulation} 
%	    \label{fig_CalcTime_Large}
%\end{center}
%\end{figure}

Finally, we apply our proposed framework to a setting with a large number of agents. 
The network considered in this analysis is the same as the network shown in Fig.~\ref{fig_SampleNetwork}. However, the edge capacity is enlarged. Specifically, the capacity of all edges are multiplied by the number of agents divided by $50$ together with the basic capacity as shown in Fig.~\ref{fig_SampleNetwork_capacity}.
We vary the number of agents from 50 to 250, setting the fraction of cruising agents to $0.1$, and apply our proposed myopic mechanism to the approximation algorithms \emph{Per agent} and \emph{Branch Cutting}. We execute a numerical study and evaluate the achieved total social welfare, varying the max branch $N_{branch}^{max}$ from 1 to 10000. Note the case where the max branch is set to $N_{branch}^{max}=1$ coincides with $FCFS$. In Fig.~\ref{fig_results_LargeScale}, we show the social welfare achieved by each number of max branch divided by that achieved by FCFS. As we can see from this figure, higher social welfare is achieved when we set a higher branch. For example, the proposed algorithm performs about 2 to 3\% higher social welfare when the max branch is set to $N_{branch}^{max}=10000$, while it reaches 2 to 5\% without the limitation of the max branch. 

In our computational environment shown in Section~\ref{sec_calc_time}, it takes only 1500 seconds to calculate the settings with 250 agents. However, owing to the nature of ZDD, the computation requires a large size of RAM. Using an environment with large sized RAM, our proposed framework can be implemented for a service with hundreds of agents.

%計算時間：Fig.~\ref{fig_CalcTime_Large}

%\clearpage
%%%%%%%%%%%%%%%%%%%%%%%%%%%%%%%%%%%%%%%%%%%%%%%%%%%%%%%%%%%%%%%%%%%%%%%%%%%%%%%%%%%%%%%%%%%%%%%%%%%%%%%%%%%%%%%%
%
\section{\label{sec_discussions}Conclusions and discussions}									
%
%
%%%%%%%%%%%%%%%%%%%%%%%%%%%%%%%%%%%%%%%%%%%%%%%%%%%%%%%%%%%%%%%%%%%%%%%%%%%%%%%%%%%%%%%%%%%%%%%%%%%%%%%%%%%%%%%

In this study, we addressed dynamic traffic resource allocation problems with strict capacity constraints and elastic demand of users with space--time prism constraints who require the guarantee of service quality in the worst cases. We characterized this problem by using activity-based user model, in which the relationship between the successive transfers generated from activities was considered. In many settings, this problem includes many non-linear constraints. Thus we take an approach that does not use the ILP, but uses graph algorithms.
%some users may want to go to restaurant A at noon or go to restaurant B at 1:00 p.m.

In such a setting, we characterize a class of RC-feasible mechanisms that strictly keep both, space--time prism constraints of customers and capacity constraints of traffic resources. We also showed the RC-optimal mechanism that maximizes the discounted social welfare by keeping these constraints.
Based on the RC-optimal and RC-feasible mechanisms, we propose a floating booking system that allocates traffic resources efficiently by rationally reallocating previously assigned traffic resources to late-coming high-value customers. We showed exact algorithms for the RC-optimal mechanism using the ZDD for both, myopic and non-myopic settings. We also showed multiple approximation algorithms that are RC-feasible. The proposed framework provides a wide range of options from a simple FCFS algorithm to the algorithm that achieves economically dynamic social optimal states in which discounted social welfare is maximized at each time. Thus, we can find an appropriate algorithm within a trade-off between computational costs and efficiency, depending on applications.

In several studies, we showed that our proposed model can keep more than trillions of combinatorial trip options of current and future agents in rational computational time, and can thus be used effectively in settings with high rejection rates, meaning that this mechanism can be used for focusing on the behavior of latent customers who have not used mobility services so far. Specifically, our proposed framework is effective as a demand prediction tool that can consider induced demands depending on the service quality.

We also showed that our proposed \emph{Per agent} approximation algorithm when introduced to the non-myopic algorithm, is effective. Commonly, the operator can achieve high efficiency by treating agents reporting in certain durations simultaneously in its decision-making, but it requires high computational costs. Our proposed mechanism solves this problem by making decisions at each agent using information from multiple agents. This approach is effective in designing mobility services with a fixed small number of customers, such as car-sharing systems with specified residents. In such settings, the space--time prism constraints of customers can be highly predictable by observing the daily active patterns of specified customers and thus, the non-myopic algorithm is suitable.

Moreover, we showed that the \emph{Per agent} approximation algorithm introduced to the myopic algorithm performs much better than the common FCFS mechanism does, in settings where the type of agents changes dynamically. 
%\textcolor{red}{The approximation algorithm is computable for a large number of agents. For a setting with 100 agents, it achieves 10\% higher performance than the FCFS does when the maximum branch is set to 10000.}
This approach is effective in designing mobility services for unspecified many customers, such as shared taxi systems in large cities. In such settings, our proposed floating booking system can provide a lot of flexibility to customers. Customers can arbitrarily change the booking as far as it does not bother other users, and in case the change is rejected, the space--time constraints that the customer originally registered are still guaranteed.

Our proposed framework brings up various directions for discussion.
The first is the discussion on the speculative behavior of customers. In the floating type of bookings, customers can be better-off by reporting untruthful demand on purpose and thus, the pricing algorithm that promotes customers' truthful reports should be discussed. To address these problems, \citet{bergemann2010dynamic} proposed the pivot mechanism by which strategy-proofness of users is guaranteed, in that customers can never be better-off by misreporting. Our proposed framework that describes decision-making of all agents at a normalized time-step is consistent with this mechanism at the point that both of them are based on a basic economic model. Indeed, we can establish the pivot mechanism by combining our proposed non-myopic RC-optimal mechanism that maximizes the expected discounted social welfare and the well-defined pivot pricing\citep{bergemann2010dynamic}. Given that, we can design a mobility system that achieves social optimal states by the best response strategy of self-interested customers \citep{hayakawa2018auction}.
In considering the pricing algorithm, it is natural to focus not only on social welfare, but also on the profit of the operator. In MaaS systems, the operator may not hold a fixed number of traffic resources but procure them flexibly. In such a setting, the combined procurement, pricing, and allocation problems should be considered, which bring up a trade-off between profit of the operator and benefit of its customers \citep{Hayakawa2018AAMAS, Hayakawa2015}. 
 
The second is the discussion on the applications. Our proposed framework presents the possibilities of many emerging applications. One instance is the electric vehicle (EV) dispatch problems for car-sharing or ride-sharing services. In such systems, the space--time constraints of vehicles should also be considered. EVs have a limited cruising distance depending on their battery capacities and thus, they have to return to a charging station before the battery runs out. Thus, the executable trajectory of vehicles is expressed in space--time prism constraints based on a series of customers and charging stations. The problem is formulated as a combination of the vehicle assignment to customers, routing, and selecting charging stations problems. This setting provides a lot of non-linear constraints and thus, our proposed framework is effective. So far, many studies have discussed this type of a problem based on trip-based formulation and not based on activity-based formulation. If it only considers myopic decision-making based on the emerged deterministic demands, the myopic trip-based approach can be reasonable. However, to provide non-myopic decision-making based on the elastic and potentially hidden demands, our proposed framework with the activity-based model will play an important role.

The third is the discussion on the various types of space--time prism constraints. \citet{hagerstraand1970people} introduced three types of space--time prism constraints, namely, ``capability constraints'', ``coupling constraints'' and ``authority constraints.'' To make the discussion simple, we only refer to the ``capability constraints'' in this paper. However, by introducing two other types of constraints, our proposed framework becomes more resourceful. Considering ``coupling constraints'' between people and people, we can express the activities with families and friends. Our proposed framework can easily be extended to such situations by introducing reward functions depending on the person doing the same activities. Considering ``coupling constraints'' between people and objects, such as vehicles, machines, and so on, we can consider another type of application. The EV application stated above is an instance considering the ``coupling constraints'' between people and vehicles. In addition, we can extend the model to logistical problems by introducing this type of constraints. Finally, if we introduce ``authority constraints,'' we can discuss far wider problems. For example, it can be used as a tool to design premium charge for limited priority members. It may also be useful to discuss the diversity of mobility services aiming to realize services for all people, including the elderly, children, handicapped persons, and so on. Emerging technologies in the coming age, including automated vehicles, can provide mobility services for people who cannot have their own driver's licenses, and our proposed framework can be extended to evaluate the possibility of such achievements by new services.

However, our study still has many limitations. The first and the most important problem to be discussed is the method to measure the utility of activities. 
In the earlier study, \citet{kitamura1997empiricall} presented the concept of ``temporal utility profiles'' of activities and travel, and empirically discussed it using data obtained in the San Diego Zoo, and after that, many studies approached this problem. With recent progress in information technologies, a data driven approach for this problem is promising.
Another limitation to be discussed is solution algorithms to solve the problems more efficiently. Depending on applications, the RC-feasible algorithms should be refined and the specific algorithms for this should be explored. For example, $A^{\ast}$-based algorithms considered in MAPF settings can potentially be expanded to our settings.
\color{black}
In the future, we plan to explore these problems.

%\balance
%\clearpage
%\bibliographystyle{jsce}
%\bibliographystyle{junsrt}
%\bibliographystyle{plain}

%\bibliographystyle{apalike}
%\bibliography{references_2015_tr}

%\bibliography{references_add}

%\clearpage
%\lastpagecontrol[0cm]{0cm}
%\lastpagecontrol[1cm]{13.7cm}
%\lastpagesettings[0cm]

\end{document}